\documentclass[12pt,notitlepage]{amsart}%
\usepackage{amssymb}
\usepackage{amsfonts}
\usepackage{graphicx}
\usepackage{amscd}
\usepackage{amsmath}
\usepackage{amsxtra}
\usepackage{booktabs}
\usepackage{rotating}
\usepackage{footmisc}%
\usepackage{hyperref}
\usepackage{float}
\usepackage{xcolor}
\usepackage{soul}
\sethlcolor{yellow}

\usepackage{algpseudocode}
\usepackage{colortbl}
\usepackage{longtable}
\usepackage{adjustbox}
\usepackage{tabularx}
\usepackage{threeparttable}
\usepackage{makecell}
\usepackage{array}

\newcolumntype{Y}{>{\centering\arraybackslash}X}

\usepackage{subcaption}

\usepackage{tikz}
\usetikzlibrary{shapes.geometric,arrows.meta,positioning,backgrounds,calc,fit}

\newtheorem{theorem}{Theorem}
\theoremstyle{plain}

\newtheorem{algorithm}{Algorithm}

\newtheorem{lemma}{Lemma}

\newtheorem{proposition}{Proposition}

\numberwithin{equation}{section}
\numberwithin{theorem}{section}
\numberwithin{lemma}{section}
\numberwithin{proposition}{section}
\numberwithin{corollary}{section}

\ifx\pdfoutput\relax\let\pdfoutput=\undefined\fi
\newcount\msipdfoutput
\ifx\pdfoutput\undefined\else
\ifcase\pdfoutput\else
\ifx\paperwidth\undefined\else
\ifdim\paperheight=0pt\relax\else\pdfpageheight\paperheight\fi
\ifdim\paperwidth=0pt\relax\else\pdfpagewidth\paperwidth\fi
\fi\fi\fi
\begin{document}
\title[Patterns in Quantum Hierarchical CNNs]{Pattern Formation in Quantum Hierarchical Cellular Neural Networks}
\author[Z\'{u}\~{n}iga-Galindo]{W. A. Z\'{u}\~{n}iga-Galindo}
\address{University of Texas Rio Grande Valley\\
School of Mathematical \& Statistical Sciences\\
One West University Blvd\\
Brownsville, TX 78520, United States}
\email{wilson.zunigagalindo1@utrgv.edu}
\author[Zambrano-Luna]{B. A. Zambrano-Luna}
\address{ Universidad Sergio Arboleda\\ Escuela de Ciencias Exactas e Ingeniería\\
Calle 74 No. 14-14, Bogotá, Colombia}
\email{brian.zambrano@usa.edu.co}
\author[Indoung]{Chayapuntika Indoung}
\address{University of Texas Rio Grande Valley\\
School of Mathematical \& Statistical Sciences\\
One West University Blvd\\
 Edinburg, TX 78539, United States}
\email{chayapuntika.indoung01@utrgv.edu}
\begin{abstract}
We present a new class of quantum neural networks (QNNs) whose states are
solutions of $p$-adic Schr\"{o}dinger equations with a non-local potential
that controls the interaction between the neurons. These equations are
obtained as Wick rotations of the state equations of $p$-adic cellular neural networks (CNNs). The $p$-adic CNNs arise as continuous limits of large discrete hierarchical neural
networks (NNs). The CNNs are bio-inspired by the
Wilson--Cowan model, which describes the macroscopic dynamics of large
populations of neurons. We provide a detailed study of the discretization of
the new $p$-adic Schr\"{o}dinger equations, which allows the construction of
new QNNs on simple graphs. We also conduct detailed numerical simulations,
offering a clear insight into the functioning of the new QNNs. At a
mathematical level, we show the existence of global solutions for the new $p$%
-adic Schr\"{o}dinger equations.

\end{abstract}
\maketitle

\section{\label{Section_0}Introduction}

In the 1980s, Chua and Yang introduced a new natural computing paradigm,
cellular neural networks (CNNs), which include cellular automata as a special case. CNNs are a parallel computing paradigm similar to neural networks, where
each cell (a processing unit) only communicates directly with its nearest
neighbors; this localized, analog processing enables the real-time, high-speed
solutions of complex computational problems, particularly those involving
spatial and temporal data, such as image and video processing; see e.g.,
\cite{Chua-Tamas}-\cite{Slavova}, and the references therein. The CNN model is
bioinspired by the Wilson--Cowan model, which describes the evolution of
excitatory and inhibitory activity in a synaptically coupled neuronal network;
see, e.g., \cite{Wilson-Cowan 1}-\cite{Neural-Fields} and the references therein.

Quantum cellular neural networks (QCNNs) are computational models that
integrate quantum mechanics with the architecture of classical cellular neural
networks; they were introduced by Toth et al. in the 1990s, \cite{Toth et al}.
In the original proposal, quantum-dot cellular automata were used as
fundamental processing units. The dynamics of the QCNN cells are governed by a
time-dependent Schr\"{o}dinger equation, enabling the network to operate on
continuous values rather than just binary ones. In this paper, we initiate the rigorous study of the $p$-adic QCNNs as tools for quantum computing. 

Before introducing the formal construction, it is useful to explain intuitively why $p$-adic numbers are a natural domain for hierarchical neural networks. The $p$-adic integers $\mathbb{Z}_p$ can be visualized as the boundary of an infinite rooted tree in which every node has $p$ children: two points $x,y\in\mathbb{Z}_p$ are ``close'' -- in the $p$-adic sense, $\left\vert x-y\right\vert _p$ small -- exactly when they share a long common branch of that tree, i.e., when their $p$-adic digit expansions agree on many leading digits. This is an ultrametric, hierarchical notion of proximity, fundamentally different from the Euclidean one: closeness is organized by shared ancestry in the tree rather than by distance along a line. Such a structure is a natural mathematical idealization of systems that are organized in nested scales or clusters -- for instance, the hierarchical organization of neurons into cells, columns, areas, and lobes in biological neural systems, which is precisely the Wilson--Cowan-inspired motivation behind the classical $p$-adic CNNs recalled below. This intuition is what motivates working on $\mathbb{Z}_p$ rather than on a Euclidean domain $\mathbb{R}^n$, as is customary in most existing QNN proposals.

The proposal of constructing quantum neural networks (QNNs) based on the
Schr\"{o}dinger equation is not new, see \cite{Nakajima et al}-\cite{Schuld et
al}. In \cite{Behera et al}, the authors introduced a recurrent quantum neural
network (RQNN) model for eye tracking of moving targets based on a non-linear
Schr\"{o}dinger equation. On the other hand, the non-linear
Schr\"{o}dinger equation has been used extensively in image processing; see
e.g. \cite{Singh et al}-\cite{Honigman et al}.

The construction of mathematical models for hierarchical neural networks (NNs) is an area of
intense activity. Here, it is relevant to note that the $p$-adic numbers are
naturally organized in a tree-like structure, which can be used to construct
models of hierarchical NNs. In \cite{Zambrano-Zuniga-1}, the first two authors
introduced the $p$-adic cellular neural networks, which are mathematical
generalizations of classical CNNs. The new networks have an infinite number of
cells organized hierarchically in a tree-like structure, with an infinite
number of hidden layers. Intuitively, the $p$-adic CNNs occur as limits of
large hierarchical discrete CNNs. In \cite{Zambrano-Zuniga-2}-\cite{Zuniga et
al}, they developed image processing algorithms using $p$-adic CNNs. In this
paper, we introduce the quantum counterparts of the $p$-adic CNNs.

The continuous-time quantum walks (CTQWs) are fundamental tools for
developing efficient quantum algorithms, simulating complex quantum systems,
and potentially achieving quantum advantages in various computational tasks, \cite{Farhi-Gutman}-\cite{Childs et al}. In \cite{Zuniga-QM-2}, see also
\cite{Zuniga-Mayes}-\cite{Zuniga-Chacon}, the first author showed that a large
class of $p$-adic Schr\"{o}dinger equations is the scaling limit of certain
continuous-time quantum Markov chains (CTQMCs). In practice, discretizing such
an equation yields a CTQMC. This construction includes, as a particular case,
the CTQWs constructed using adjacency matrices.

In this paper, we introduce a new class of $p$-adic non-linear Schr\"{o}dinger
equations obtained as the Wick rotation of the state equations of the $p$-adic
CNNs. The equations have the form
\begin{equation}
i\frac{\partial}{\partial t}\Psi\left(  x,t\right)  =\Psi\left(  x,t\right)
-J\left(  \left\vert x\right\vert _{p}\right)  \ast\Psi\left(  x,t\right)  +%
{\displaystyle\int\limits_{\mathbb{Z}_{p}}}
W\left(  x,y\right)  \phi\left(  \Psi\left(  y,t\right)  \right)  dy+Z(x,t),
\label{Model 3A}%
\end{equation}
where $J:\left[  0,1\right]  \rightarrow\mathbb{R}_{\geq0}$, and
$\int_{\mathbb{Z}_{p}}J\left(  \left\vert x\right\vert _{p}\right)  dx=1$, and
$W\left(  x,y\right)  $, $Z(x,t)$, $\Psi\left(  x,t\right)  $ are complex-valued
functions; $\mathbb{Z}_{p}$ is the $p$-adic unit ball, which is an additive
group, and $dx$ denotes a Haar measure on this group. When $W\left(
x,y\right)  =0$, (\ref{Model 3A}) is a  CTQMC. In this case, if
$\left\Vert \Psi\left(  \cdot,0\right)  \right\Vert _{2}=1$, then $\left\Vert
\Psi\left(  \cdot,t\right)  \right\Vert _{2}=1$ for any $t\geq0$, which means
that (\ref{Model 3A}) describes a unitary evolution. Rigorous discretizations
of (\ref{Model 3A}) (with $W\left(  x,y\right)  =0$) are discussed in
\cite{Zuniga-QM-2}-\cite{Zuniga-Chacon}. On the other hand, when $W\left(  x,y\right)  \neq0$,
(\ref{Model 3A}) describes a non-unitary evolution. Decoherence is the process
where a quantum system loses its unique quantum properties (such as
superposition and entanglement) due to interaction with its surrounding environment. Non-unitary evolution (described, for instance, by Lindblad master equations) is the mathematical framework used to model that process. Based on these ideas, we propose that (\ref{Model 3A}) is a new kind of master equation that describes an open quantum network. This interpretation has been confirmed in all our numerical simulations.

By expressing the functions $W\left(  x,y\right)  $, $Z(x,t)$, $\Psi\left(
x,t\right)  $ in terms of their real and imaginary parts, the
integro-differential equation (\ref{Model 3A}) corresponds to a
Wilson-Cowan-type system for $\operatorname{Re}\left(  \Psi\left(  x,t\right)
\right)  $ and $\operatorname{Im}\left(  \Psi\left(  x,t\right)  \right)  $.
This shows directly that (\ref{Model 3A}) is a quantum version of a neural
activity model; see \cite{Zuniga-Entropy}, and the references therein.

Suitable discretizations of (\ref{Model 3A}) give new quantum networks on
simple graphs. Let $(\mathcal{G},V,E)$ be a simple graph, with adjacency
matrix $\boldsymbol{A}=\left[  A_{I,K}\right]  _{I,K\in V}$. We set
$\boldsymbol{D}=\left[  D_{I,K}\right]  _{I,K\in V}$ to be the diagonal matrix
defined as $D_{I,K}=val(K)\delta_{I,K}$, where $val(K)$ denotes the valence
of vertex $K$, and $\delta_{I,K}$ is the Kronecker symbol. A new example of a QNN (obtained using our theory) is given by the following discrete Schr\"{o}dinger equation on
$\mathcal{G}$:%
\begin{gather}
i\frac{\partial}{\partial t}\left[  \Psi_{I}\left(  t\right)  \right]  _{I\in
V}=-p^{-l}\left(  \boldsymbol{A}-\boldsymbol{D}\right)  \left[  \Psi_{I}\left(  t\right)
\right]  _{I\in V}\label{Model 1}\\
+\left[  p^{-l}%
{\displaystyle\sum\limits_{K\in V}}
W_{I,K}\phi\left(  \Psi_{K}\left(  t\right)  \right)  +Z_{I}(t)\right]  _{I\in
V};\nonumber
\end{gather}
where $\Psi_{I}\left(  t\right)  \in\mathbb{C}$, for $t\geq0$, $I\in V$, is
the quantum state of vertex $I$, and the column vector $\left[  \Psi
_{I}\left(  t\right)  \right]  _{I\in V}$ is the wavefunction of the network;
$W_{I,K}\in\mathbb{C}$, for $I,K\in V$, is the weight of the connection
between the neurons; $Z_{I}(t)\in\mathbb{C}$, for $t\geq0$, $I\in V$, is the
bias of the neuron $I$ at time $t$, and $\phi:\mathbb{R}\rightarrow\mathbb{R}%
$, with $\phi\left(  a+ib\right)  =\phi\left(  a\right)  +i\phi\left(
b\right)  $, is the activation function. Finally, $l$ is a positive integer
and $p$ is a prime number. The quantity $p^{-l}$ is a scale factor. The network
(\ref{Model 1}) is a quantum version of a $p$-adic CNN. If $\phi=0$, the
Schr\"{o}dinger equation%
\[
i\frac{\partial}{\partial t}\left[  \Psi_{I}\left(  t\right)  \right]  _{I\in
V}=-p^{-l}\left( \boldsymbol{A}- \boldsymbol{D}\right)  \left[  \Psi_{I}\left(  t\right)
\right]  _{I\in V}%
\]
describes the standard CTQWs on graphs; see, e.g. \cite{Farhi-Gutman}%
-\cite{Zuniga-Chacon}.

It is relevant to mention that there is another type of Schr\"{o}dinger
equation on graphs; see \cite{Noja} and the references therein. In this
context, a Schr\"{o}dinger equation on a graph is a collection of ordinary
Schr\"{o}dinger equations parametrized by the graph's edges, where the
vertices encode boundary conditions for the equations. To the best of our
understanding, this type of Schr\"{o}dinger equation is not related at all to
the equations considered here. On the other hand, in \cite{Verdon et al}, see
also \cite{Ceschini et al}, the quantum graph neural networks (QGNNs) were
introduced. In this approach, a discrete Hamiltonian is constructed from a graph, and the network is defined by an evolution operator obtained as
products of exponentials of this Hamiltonian. These networks are completely different from the ones considered here.

The broader rationale for introducing yet another QNN architecture is best explained through the limitations of the existing proposals just reviewed. Variational quantum circuits, the dominant near-term paradigm, are finite-dimensional, strictly unitary, and carry no intrinsic hierarchical structure beyond what the circuit layout imposes; non-linearity enters only through classical post-processing of measurement outcomes. Quantum graph neural networks are likewise unitary and discrete-time, with no explicit non-linear activation function built into the quantum dynamics itself. Schr\"{o}dinger-equation-based QNNs defined on Euclidean domains $\mathbb{R}^n$ do not carry a native multiscale or hierarchical geometry; any such structure must be imposed externally. The $p$-adic QCNNs introduced here address these three limitations simultaneously: the domain $\mathbb{Z}_p$ carries a tree-like topology, non-linearity enters directly into the Schr\"{o}dinger dynamics through the activation function $\phi$, and a single governing equation, \mbox{\eqref{Model 3A}} below, interpolates continuously between closed (unitary) and open (non-unitary) evolution depending on whether the weight kernel $W$ vanishes. A detailed, criterion-by-criterion comparison against six representative QNN families -- variational quantum circuits, the original quantum-dot QCNNs of \mbox{\cite{Toth et al}}, quantum graph neural networks, Euclidean Schr\"{o}dinger-based QNNs, continuous-time quantum walks, and the classical $p$-adic CNNs from which the present model is derived -- is given in Section \mbox{\ref{Section_6A}}; the comparison there also identifies the specific gap in the literature that motivates the present construction.

It is also useful to state at the outset how the proposed framework relates to practical quantum computing, since it is formulated as a Schr\"{o}dinger-equation-based dynamical system rather than as a variational quantum circuit implemented on quantum hardware. The present work should be understood as a rigorous mathematical and theoretical foundation, with a partially implementable linear sector: the free evolution (obtained when $W=0$ and $Z=0$) is genuinely unitary and, as discussed in Section \mbox{\ref{Section_7A}}, admits in principle a gate-based realization via a quantum-Fourier-transform-type decomposition (for $p=2$, this coincides with the Walsh--Hadamard transform on $l$ qubits). The non-linear and open-system regimes ($W\neq0$ or $Z\neq0$), which are central to the pattern-formation and decoherence phenomena studied numerically in Section \mbox{\ref{Section_6}}, do not currently correspond to sequences of unitary gates; at this stage they constitute a theoretical, quantum-inspired model of an open quantum network, and their hardware realization is identified as an open problem in Section \mbox{\ref{Section_7A}} and in the Limitations discussion of Section \mbox{\ref{Section_7}}.

Beyond its relevance to quantum computing, the proposed framework is also of intrinsic interest to non-linear science. Equation \mbox{\eqref{Model 3A}} is a genuinely non-linear, non-local dynamical system on a non-Archimedean (ultrametric) space, so its pattern-forming behavior is of interest independently of any quantum-computational motivation. Moreover, since setting $W=0$ and $\phi=0$ recovers the standard continuous-time quantum walk on a graph, the framework strictly generalizes CTQWs and connects naturally to the graph-based quantum-algorithms literature. Finally, because the governing equation is obtained from the classical Wilson--Cowan-inspired $p$-adic CNNs by a Wick rotation, it provides a concrete mathematical bridge between classical neural pattern formation and a quantum analogue, which is of independent interest for the study of quantum models of neural and brain dynamics (see Section \mbox{\ref{Section_7}}).

We conclude this Introduction by summarizing the manuscript's objectives and principal contributions. The primary objective is to construct and rigorously analyze a class of quantum neural networks whose states solve $p$-adic Schr\"{o}dinger equations on a hierarchical (ultrametric) domain, and to understand, both analytically and numerically, how a non-local weight kernel drives the transition between closed (unitary) and open (non-unitary) quantum dynamics. The main contributions are: (a) a new class of nonlinear $p$-adic Schr\"{o}dinger equations, obtained as the Wick rotation of the state equation of the $p$-adic CNNs of \mbox{\cite{Zambrano-Zuniga-1}}; (b) a global-existence and uniqueness result for these equations (Appendix); (c) a rigorous discretization scheme that produces finite-dimensional QNNs on simple graphs (Sections \mbox{\ref{Section_4}}--\mbox{\ref{Section_5}}); (d) a systematic, criterion-by-criterion comparison against six existing QNN families, clarifying the novelty of the present approach (Section \mbox{\ref{Section_6A}}); (e) a detailed numerical study characterizing the unitary-to-non-unitary transition and the resulting decoherence and pattern-formation phenomena (Section \mbox{\ref{Section_6}}); and (f) an analysis of computational complexity, memory cost, and scalability, together with a discussion of the prospects for quantum hardware implementation (Section \mbox{\ref{Section_7A}}).

The paper is organized as follows. In Section \ref{Section_2}, we present some basic results from $p$-adic analysis that are needed here. In Section
\ref{Section_3}, we introduce the $p$-adic QNNs; Section \ref{Section_4} is
dedicated to the discretization of the QNNs. In Section \ref{Section_5}, new
QNNs on simple graphs are presented. Section \ref{Section_6A} presents a discussion comparing quantum machine learning (QML) against the new $p$-adic QCNNs. Section \ref{Section_6} is dedicated to numerical simulations. Section \ref{Section_7A} discusses the computational complexity of the $p$-adic QCNNs. Finally, in Section \ref{Section_7}, we present the conclusions and a general discussion. The Appendix is dedicated to the existence of global solutions for certain $p$-adic Schr\"{o}dinger equations. 

Figure~\ref{fig:workflow} provides a schematic overview of the entire
construction, from the Wilson--Cowan biological model to the final QNN on a
simple graph, and may serve as a guide for the reader throughout the paper.

\begin{figure}[t]
\centering
\tikzset{
  biobox/.style={rectangle, rounded corners=7pt,
    draw=green!55!black, line width=1.1pt, fill=green!8,
    text width=7.6cm, minimum height=1.5cm, align=center},
  classicbox/.style={rectangle, rounded corners=7pt,
    draw=teal!65!black, line width=1.1pt, fill=teal!7,
    text width=7.6cm, minimum height=1.5cm, align=center},
  contqbox/.style={rectangle, rounded corners=7pt,
    draw=blue!65!black, line width=1.1pt, fill=blue!6,
    text width=7.6cm, minimum height=1.5cm, align=center},
  discqbox/.style={rectangle, rounded corners=7pt,
    draw=violet!60!black, line width=1.1pt, fill=violet!5,
    text width=7.6cm, minimum height=1.5cm, align=center},
  ctqwbox/.style={rectangle, rounded corners=7pt,
    draw=orange!65!black, line width=1.1pt, fill=orange!7,
    text width=3.2cm, minimum height=1.5cm, align=center},
  mainarr/.style={-{Stealth[length=8pt,width=6pt]},
    line width=1.2pt, color=gray!60!black},
  sidarr/.style={-{Stealth[length=8pt,width=6pt]},
    line width=1.2pt, color=orange!70!black, dashed},
  albl/.style={font=\scriptsize\itshape,
    fill=white, inner sep=2pt, text=gray!60!black}
}
\scalebox{0.80}{%
\begin{tikzpicture}[every node/.style={font=\small}]

\node[biobox] (WC) {%
  \textbf{Wilson--Cowan Neural Field Model}\\[3pt]
  {\scriptsize Neural activity on $(\Omega,d\mu)$:
  $\;d\boldsymbol{y}/dt = -\boldsymbol{H}\boldsymbol{y}
    + \int_\Omega W\phi(\boldsymbol{y})\,d\mu + \boldsymbol{\eta}$}\\[2pt]
  {\scriptsize\textit{Real-valued; dissipative (Markovian) dynamics;
  eq.~\eqref{Eq_1}}}%
};

\node[classicbox, below=1.6cm of WC] (pCNN) {%
  \textbf{$p$-Adic CNN on $\mathbb{Z}_p$}\\[3pt]
  {\scriptsize $\Omega = \mathbb{Z}_p$, $d\mu = dx$ (Haar);\;
  $\boldsymbol{H} = J\bigl(|x|_p\bigr){*}\cdot - \mathrm{Id}$}\\[2pt]
  {\scriptsize\textit{Hierarchical tree topology; infinitely many
  layers; \cite{Zambrano-Zuniga-1}--\cite{Zuniga et al}}}%
};

\node[contqbox, below=1.6cm of pCNN] (pQNN) {%
  \textbf{$p$-Adic QNN on $\mathbb{Z}_p$}\\[3pt]
  {\scriptsize $i\,\partial_t\Psi = \boldsymbol{H}\Psi
    + \int_{\mathbb{Z}_p} W(x,y)\,\phi(\Psi(y,t))\,dy + Z(x,t)$}\\[2pt]
  {\scriptsize\textit{Complex-valued; unitary if $W\!=\!Z\!=\!0$;\;
  open system if $W\!\neq\!0$ or $Z\!\neq\!0$;\;
  eq.~\eqref{EQ_0A}}}%
};

\node[discqbox, below=1.6cm of pQNN] (discQNN) {%
  \textbf{Discrete QNN on $G_l = \mathbb{Z}_p/p^l\mathbb{Z}_p$}\\[3pt]
  {\scriptsize $i\,\dot{\Psi}_I
    = -[\boldsymbol{J}^{(l)}\boldsymbol{\Psi}]_I
    + p^{-l}\!\sum_{K\in G_l} W_{I,K}\,\phi(\Psi_K) + Z_I$,\;
  $I \in G_l$}\\[2pt]
  {\scriptsize\textit{$N\!=\!p^l$ neurons;\;
  $\boldsymbol{J}^{(l)}$ = convolution matrix on
  $G_l\!\cong\!\mathbb{Z}/p^l\mathbb{Z}$;\;
  eq.~\eqref{Network_Type_I}}}%
};

\node[discqbox, below=1.6cm of discQNN] (graphQNN) {%
  \textbf{QNN on Simple Graph $\mathcal{G} = (V,E)$}\\[3pt]
  {\scriptsize $J^{(l)}_{I,K} = p^{-l}A_{I,K}$,\quad
  $J^{(l)}_{I,I} = -p^{-l}\mathrm{val}(I)$,\quad $I,K\in G_l^0$}\\[2pt]
  {\scriptsize\textit{Vertices $=$ neurons;\;
  adjacency matrix $A$ determines $\boldsymbol{J}^{(l)}$;\;
  eq.~\eqref{Matrix_CTQW}}}%
};

\node[ctqwbox, right=1.8cm of graphQNN] (CTQW) {%
  \textbf{CTQW on $\mathcal{G}$}\\[3pt]
  {\scriptsize $i\,\dot{\Psi}_I
    = -[\boldsymbol{J}^{(l)}\boldsymbol{\Psi}]_I$}\\[3pt]
  {\scriptsize $W\!=\!0$,\; $\phi\!=\!0$}\\[2pt]
  {\scriptsize \cite{Farhi-Gutman}--\cite{Zuniga-Chacon}}%
};

\draw[mainarr] (WC.south) --
  node[albl, right=4pt]
    {$\Omega = \mathbb{Z}_p$;\;
     $\boldsymbol{H} = J(|x|_p){*}\cdot - \mathrm{Id}$
     \quad (Sec.~\ref{Section_3})}
  (pCNN.north);

\draw[mainarr] (pCNN.south) --
  node[albl, right=4pt]
    {Wick rotation:\; $t \;\to\; it$
     \quad (Sec.~\ref{Section_3})}
  (pQNN.north);

\draw[mainarr] (pQNN.south) --
  node[albl, right=4pt]
    {Discretize:\;
     $\Psi\in\mathcal{D}_l(\mathbb{Z}_p)\cong\mathbb{C}^{p^l}$
     \quad (Sec.~\ref{Section_4})}
  (discQNN.north);

\draw[mainarr] (discQNN.south) --
  node[albl, right=4pt]
    {Set $J^{(l)}_{I,K}$ from adjacency matrix of $\mathcal{G}$
     \quad (Sec.~\ref{Section_5})}
  (graphQNN.north);

\draw[sidarr] (graphQNN.east) --
  node[albl, above=3pt] {$W\!=\!0$,\; $\phi\!=\!0$}
  (CTQW.west);

\begin{scope}[on background layer]
  \node[rounded corners=10pt, fill=gray!4, draw=gray!25,
        line width=0.6pt,
        fit=(WC)(pCNN)(pQNN)(discQNN)(graphQNN),
        inner sep=8pt] {};
\end{scope}

\end{tikzpicture}%
}%
    \caption{Schematic overview of the construction of the $p$-adic QCNNs.
  \textbf{Step~1:} The neurons are placed on the $p$-adic integers $\mathbb{Z}_p$
  (a group with tree-like ultrametric topology), and the diffusion generator
  $\mathbf{H} = J_\alpha{*}\cdot - \mathrm{Id}$ is chosen, yielding the $p$-adic CNN
  (Section~3).
  \textbf{Step~2:} The Wick rotation $t\to it$ promotes the real-valued,
  dissipative CNN dynamics to complex-valued Schr\"odinger dynamics,
  giving the $p$-adic QNN (3.3) (Section~3).
  \textbf{Step~3:} Restricting solutions to the finite-dimensional space
  $\mathcal{D}_l(\mathbb{Z}_p)\cong\mathbb{C}^{p^l}$ produces a discrete
  quantum network of $N=p^l$ neurons on the quotient group
  $G_l = \mathbb{Z}_p/p^l\mathbb{Z}_p$ (Section~4).
  \textbf{Step~4:} Choosing the Hamiltonian entries from the adjacency matrix
  of a simple graph $G$ embeds the network on $G$ (Section~5).
  Setting $W=0$ and $\phi=0$ in any of the lower two models recovers
  the standard continuous-time quantum walk (CTQW) on $G$
  (\cite{Farhi-Gutman}--\cite{Zuniga-Chacon}).}
    \label{fig:workflow}
\end{figure}

\clearpage

\section{ \label{Section_2}Basic concepts of \texorpdfstring{$p$}{p}-adic analysis}

In this section, we fix the notation and collect some basic results on
$p$-adic analysis that we use in this paper. For a detailed exposition on
$p$-adic analysis, the reader may consult \cite{V-V-Z}-\cite{Kochubei}.

\subsection{$p$-Adic numbers}

From now on, we use $p$ to denote a fixed prime number. Any non-zero $p$-adic number $x$ has a unique expansion of the form

\begin{equation}
x=x_{-k}p^{-k}+x_{-k+1}p^{-k+1}+\ldots+x_{0}+x_{1}p+\ldots,\text{
}\label{p-adic-number}%
\end{equation}
with $x_{-k}\neq0$, where $k$ is an integer, and the $x_{j}$s are numbers
from the set $\left\{  0,1,\ldots,p-1\right\}  $. The set of all possible
sequences of the form (\ref{p-adic-number}) constitutes the field of $p$-adic
numbers $\mathbb{Q}_{p}$. There are natural field operations, sum and
multiplication, on series of form (\ref{p-adic-number}). There is also a norm
in $\mathbb{Q}_{p}$ defined as $\left\vert x\right\vert _{p}=p^{-ord(x)}$,
where $ord_{p}(x)=ord(x)=-k$, for a nonzero $p$-adic number $x$. By definition
$ord(0)=\infty$. The field of $p$-adic numbers with the distance induced by
$\left\vert \cdot\right\vert _{p}$ is a complete ultrametric space. The
ultrametric property refers to the fact that $\left\vert x-y\right\vert
_{p}\leq\max\left\{  \left\vert x-z\right\vert _{p},\left\vert z-y\right\vert
_{p}\right\}  $ for any $x$, $y$, $z$ in $\mathbb{Q}_{p}$. The $p$-adic
integers, which are sequences of the form (\ref{p-adic-number}) with $-k\geq0$,
constitute the unit ball $\mathbb{Z}_{p}$. The unit ball is an infinite rooted
tree with fractal structure. As a topological space $\mathbb{Q}_{p}$ is
homeomorphic to a Cantor-like subset of the real line, see, e.g.,
\cite{V-V-Z}, \cite{Alberio et al}. $\mathbb{Q}_{p}$ is a paramount example of a totally disconnected space, with a very rich mathematical structure.

\subsection{$p-$Adic balls and spheres}

The ball with center at $a\in\mathbb{Q}_{p}$ and radius $p^{-l}$ is the set
\[
B_{-l}(a)=\left\{  x\in\mathbb{Q}_{p};\left\vert x-a\right\vert _{p}\leq
p^{-l}\right\}  =a+p^{l}\mathbb{Z}_{p},
\]
and the sphere with center at $a\in\mathbb{Q}_{p}$ and radius $p^{-l}$ is the
set%
\[
S_{-l}(a)=\left\{  x\in\mathbb{Q}_{p};\left\vert x-a\right\vert _{p}%
=p^{-l}\right\}  =a+p^{l}\mathbb{Z}_{p}^{\times},
\]
where
\[
\mathbb{Z}_{p}^{\times}=\left\{  x\in\mathbb{Z}_{p};\left\vert x\right\vert
_{p}=1\right\}  .
\]
The ball $B_{-l}(0)$ is the unit ball $\mathbb{Z}_{p}$. If $a\in\mathbb{Z}%
_{p}$ and $l$ is a non-negative integer,%
\[
S_{-l}(a)\subset B_{-l}(a)\subset\mathbb{Z}_{p}.
\]
Furthermore,%
\[
B_{-l}(a)=S_{-l}(a)%
{\displaystyle\bigsqcup}
B_{-\left(  l+1\right)  }(a),
\]
where $%
{\textstyle\bigsqcup}
$\ denotes disjoint union.

\subsection{Test functions}

A function $\varphi:\mathbb{Q}_{p}\rightarrow\mathbb{C}$ is called locally
constant if for any $a\in\mathbb{Q}_{p}$, there is an integer $l=l(a)$, such
that
\[
\varphi\left(  a+x\right)  =\varphi\left(  a\right)  \text{ for any }%
|x|_{p}\leq p^{l}.
\]
The set of functions for which $l=l\left(  \varphi\right)  $ depends only on
$\varphi$ forms a $\mathbb{C}$-vector space denoted as $\mathcal{U}%
_{loc}\left(  \mathbb{Q}_{p}\right)  $. We call $l\left(  \varphi\right)  $
the exponent of local constancy. If $\varphi\in\mathcal{U}_{loc}\left(
\mathbb{Q}_{p}\right)  $ has compact support, we say that $\varphi$ is a test
function. Any function of this type is a linear combination of characteristic
functions of balls. We denote by $\mathcal{D}(\mathbb{Q}_{p})$ the complex
vector space of test functions. There is a natural integration theory so that
$\int_{\mathbb{Q}_{p}}\varphi\left(  x\right)  dx$ gives a well-defined
complex number. The measure $dx$ is the Haar measure of $\mathbb{Q}_{p}$,
\cite{Halmos}.

\subsection{Lebesgue spaces}

Let $U\subset\mathbb{Q}_{p}$ be an open subset. We denote by $\mathcal{D}(U)$ the $\mathbb{C}$-vector space of test functions with supports
contained in $U$; then, $\mathcal{D}(U)$ is dense in
\[
L^{\rho}(U)=\left\{  \varphi:U\rightarrow\mathbb{C};\left\Vert
\varphi\right\Vert _{\rho}=\left\{
{\displaystyle\int\limits_{U}}
\left\vert \varphi\left(  x\right)  \right\vert ^{\rho}dx\right\}  ^{\frac
{1}{\rho}}<\infty\right\}  ,
\]
for $1\leq\rho<\infty$, see, e.g., \cite[Section 4.3]{Alberio et al}.

We now take $U=\mathbb{Z}_{p}$ and $\rho=2$. The density of $\mathcal{D}%
(\mathbb{Z}_{p})$ in $L^{2}\left(  \mathbb{Z}_{p}\right)  $ means that given
$f\in L^{2}\left(  \mathbb{Z}_{p}\right)  $ and $\epsilon>0$, there exists a
nonnegative integer $l$, and a function
\[
\varphi^{\left(  l\right)  }\left(  x\right)  =\sum_{I\in G_{l}}\varphi
_{I,l}\Omega\left(  p^{l}\left\vert x-I\right\vert _{p}\right)  \in
\mathcal{D}(\mathbb{Z}_{p}),
\]
such that $\left\Vert f-\varphi^{\left(  l\right)  }\right\Vert _{2}<\epsilon
$. Here, $I=I_{0}+\ldots+I_{l-1}p^{l-1}$, with $I_{k}\in\left\{
0,\ldots,p-1\right\}  $, $\varphi_{I,l}\in\mathbb{C}$, and $\Omega\left(
p^{l}\left\vert x-I\right\vert _{p}\right)  $ is the characteristic function
of the ball $I+p^{l}\mathbb{Z}_{p}$. The function $\varphi^{\left(  l\right)
}$ is a discretization of $f$; the functions $\varphi^{\left(  l\right)  }$
constitute a $\mathbb{C}$-vector space $\mathcal{D}_{l}(\mathbb{Z}_{p})$,
which is isometric to the finite-dimensional Hilbert space $\mathbb{C}^{p^{l}%
}$. The functions $\left\{  p^{\frac{l}{2}}\Omega\left(  p^{l}\left\vert
x-I\right\vert _{p}\right)  \right\}  _{I\in G_{l}}$ form an orthonormal basis.

Since notation is introduced progressively across several sections, Table \mbox{\ref{tab:notation}} collects the principal symbols used throughout the paper, together with a pointer to where each is first defined, for ease of reference.

\begin{minipage}{\textwidth}
\begin{center}
\footnotesize
\begin{tabular}{@{}llp{9.3cm}@{}}
\toprule
\textbf{Symbol} & \textbf{First used} & \textbf{Meaning} \\
\midrule
$p$ & Sec.~\ref{Section_2} & prime number defining the $p$-adic field $\mathbb{Q}_p$ \\
$\mathbb{Z}_p$ & Sec.~\ref{Section_2} & ring of $p$-adic integers (compact, boundaryless) \\
$\left\vert\cdot\right\vert_p$ & Sec.~\ref{Section_2} & $p$-adic absolute value (ultrametric) \\
$l$ & Sec.~\ref{Section_4} & discretization level (depth of the truncated $p$-adic tree) \\
$G_l=\mathbb{Z}_p/p^l\mathbb{Z}_p$ & Sec.~\ref{Section_4} & finite quotient group used to discretize $\mathbb{Z}_p$ \\
$N=p^l$ & Sec.~\ref{Section_4}, \ref{Section_7A} & number of neurons in the discretized network \\
$\Psi(x,t)$ & Sec.~\ref{Section_3} & complex-valued state (wavefunction) of the network \\
$J_\alpha(\left\vert x\right\vert_p)$ & Sec.~\ref{Section_6} & radial diffusion kernel with exponent $\alpha$, defining the free Hamiltonian \\
$\boldsymbol{J}^{(l)}(\alpha)$ & \eqref{Matrix_J_1}--\eqref{Matrix_J_2} & discretized free-evolution (kernel) matrix on $G_l$ \\
$\boldsymbol{H}_l$ & Sec.~\ref{Section_7A} & free Hamiltonian on $G_l$; convolution operator with entries from $\boldsymbol{J}^{(l)}$ \\
$W(x,y)$, $W_{I,K}$ & Sec.~\ref{Section_3}, \ref{Section_4} & non-local interaction (weight) kernel coupling neurons $x,y$ (or $I,K$) \\
$W_{cat}$ & Sec.~\ref{Section_6} & $p$-adic approximation of the cat-cortex connectivity matrix of \cite{Zuniga-Entropy} \\
$Z(x,t)$, $Z_I(t)$ & Sec.~\ref{Section_3}, \ref{Section_4} & external bias/input applied to neuron $x$ (or $I$) at time $t$ \\
 $\phi$ & Sec.~\ref{Section_3} & nonlinear activation function \\
$\Psi_0$ & Sec.~\ref{Section_6} & initial condition (state at $t=0$) \\
$\delta t$ & Sec.~\ref{Section_6} & reporting time step used in the numerical simulations \\
$\left\Vert\Psi(\cdot,t)\right\Vert_2^2$ & Sec.~\ref{Section_6} & squared $L^2$-norm; conserved iff the evolution is unitary \\
\bottomrule
\end{tabular}
\end{center}
\captionof{table}{Principal notation used throughout the paper.}
\label{tab:notation}
\end{minipage}
\bigskip

\section{\label{Section_3}CNNs and QNNs}

Denote by $\Omega$ the space of neurons, which we assume is an infinite set
endowed with a measure $d\mu$ ($\left(  \Omega,d\mu\right)  $ is a measure
space), and denote by $\boldsymbol{y}(x,t)$ a real-valued function
representing the neural activity at the location $x\in\Omega$ and time $t$. A
large class of continuous neural networks is defined by an
integro-differential equation of the form%
\begin{gather}
\frac{d\boldsymbol{y}(x,t)}{dt}=-\gamma\boldsymbol{Hy}\left(  x,t\right)  +
{\displaystyle\int\limits_{\Omega}}
\boldsymbol{A}(x,z)\boldsymbol{y}\left(  z,t\right)  d\mu\left(  z\right)
\label{Eq_1}\\
+
{\displaystyle\int\limits_{\Omega}}
\boldsymbol{B}(x,z)\boldsymbol{x}\left(  z,t\right)  d\mu\left(  z\right)  +%
{\displaystyle\int\limits_{\Omega}}
\boldsymbol{W}\left(  x,z\right)  \phi\left(  \boldsymbol{y}\left(
z,t\right)  \right)  d\mu\left(  z\right) \nonumber\\
+%
{\displaystyle\int\limits_{\Omega}}
\boldsymbol{U}\left(  x,z\right)  \phi\left(  \boldsymbol{x}\left(
z,t\right)  \right)  d\mu\left(  z\right)  +\boldsymbol{\eta}\left(
x,t\right)  \text{, }\nonumber
\end{gather}
where $x\in\Omega$,\ $t\in\mathbb{R}$. The kernels $\boldsymbol{A}(x,z)$,
$\boldsymbol{B}(x,z)$, $\boldsymbol{W}\left(  x,z\right)  $, and
$\boldsymbol{U}\left(  x,z\right)  $ define the network architecture, $\phi$
is an activation function, $\boldsymbol{x}\left(  z,t\right)  $ is
the input and $\boldsymbol{\eta}\left(  x,t\right)  $ is the bias. We assume
that $\frac{d\boldsymbol{y}(x,t)}{dt}=-\gamma\boldsymbol{Hy}\left(
x,t\right)  $, $\gamma>0$, is a diffusion equation, which means that it describes a
random motion in $\Omega\times\mathbb{R}_{\geq0}$ (a Markov process). So,
(\ref{Eq_1}) is a reaction-diffusion equation.

If the network is stochastic, the functions $\boldsymbol{A}(x,z)$,
$\boldsymbol{B}(x,z)$, $\boldsymbol{W}\left(  x,z\right)  $, and
$\boldsymbol{U}\left(  x,z\right)  $ are realizations of generalized Gaussian
random variables in $L^{2}\left(  \Omega\times\Omega\right)  $, with mean
zero, and $\boldsymbol{\eta}\left(  x,t\right)  $ is a realization of a
generalized Gaussian noise in $L^{2}(\Omega\times\mathbb{R})$. When
$\boldsymbol{H}$ is the identity operator, this type of network was considered
in \cite{Grosvenor-Jefferson}-\cite{Sompolinsky et al}.

If the network is deterministic, the functions $\boldsymbol{A}(x,z)$,
$\boldsymbol{B}(x,z)$, $\boldsymbol{W}\left(  x,z\right)  $, and
$\boldsymbol{U}\left(  x,z\right)  $ belong to a fixed, suitable function
space. If $\boldsymbol{H}$ is the identity operator $\boldsymbol{I}$, the
system (\ref{Eq_1}) is a generalization of a continuous cellular neural network
(CNN), \cite{Chua-Tamas}, \cite{Slavova}, \cite{Zambrano-Zuniga-1}%
-\cite{Zambrano-Zuniga-2}.

For simplicity, we set $\gamma=1$. By performing the Wick rotation
$t\rightarrow it$, $i=\sqrt{-1}$, and taking $\boldsymbol{y}(x,it)=\Psi
\left(  x,t\right)  $, $-\boldsymbol{A}(x,z)=A(x,z)$, $-\boldsymbol{B}%
(x,z)=B(x,z)$, $-\boldsymbol{W}(x,z)=W(x,z)$, $-\boldsymbol{U}(x,z)=U(x,z)$,
$\boldsymbol{x}\left(  z,it\right)  =\Phi\left(  z,t\right)  $,
$-\boldsymbol{\eta}\left(  x,it\right)  =Z\left(  x,t\right)  $, (\ref{Eq_1})
becomes%
\begin{gather}
i\frac{d\Psi\left(  x,t\right)  }{dt}=\boldsymbol{H}\Psi\left(  x,t\right)  +%
{\displaystyle\int\limits_{\Omega}}
A(x,z)\Psi\left(  z,t\right)  d\mu\left(  z\right) \label{Eq_2}\\
+%
{\displaystyle\int\limits_{\Omega}}
B(x,z)\Phi\left(  z,t\right)  d\mu\left(  z\right)  +%
{\displaystyle\int\limits_{\Omega}}
W\left(  x,z\right)  \phi\left(  \Psi\left(  z,t\right)  \right)  d\mu\left(
z\right) \nonumber\\
+%
{\displaystyle\int\limits_{\Omega}}
U\left(  x,z\right)  \phi\left(  \Phi\left(  z,t\right)  \right)
d\mu\left(  z\right)  +Z\left(  x,t\right)  \text{, }\nonumber
\end{gather}
where the parameters $A(x,z)$, $B(x,z)$, $W\left(  x,z\right)  $, $U\left(
x,z\right)  $, $Z\left(  x,t\right)  $ are complex-valued functions, and by
definition, $\phi\left(  a+ib\right)  =\phi\left(  a\right)  +i\phi\left(
b\right)$. We argue that (\ref{Eq_2}) is a quantum neural
network (QNN). This paper is dedicated to developing this proposal. For
simplicity, here we consider networks of type%
\begin{equation}
i\frac{d\Psi\left(  x,t\right)  }{dt}=\boldsymbol{H}\Psi\left(  x,t\right)  +%
{\displaystyle\int\limits_{\Omega}}
W\left(  x,y\right)  \phi\left(  \Psi\left(  y,t\right)  \right)  d\mu\left(
y\right)  +Z\left(  x,t\right)  \text{. } \label{EQ_0A}%
\end{equation}

There are two basic choices for $\Omega$: $\mathbb{R}$ and $\mathbb{Q}_{p}$
(or $\mathbb{Z}_{p}$). In the first case, the neurons are organized in a
straight line, whereas in the second case, they are organized in a tree-like
structure with an infinite number of layers. From now on, we set
$\Omega=\mathbb{Z}_{p}$, the $p$-adic unit ball, and $d\mu=dx$ is the normalized Haar measure of $\mathbb{Q}_{p}$.

For practical applications and numerical simulations, a discrete version of (\ref{EQ_0A}) is necessary. The discretization is obtained by assuming that the parameters  $W, Z$ belong to a finite-dimensional vector space, and by approximating $\boldsymbol{H}$ by a matrix. This is the goal of the next section.

\section{\label{Section_4}Discretization of \texorpdfstring{$p$}{p}-adic QNNs}

The equation \mbox{\eqref{Model 3A}} introduced in Section \mbox{\ref{Section_3}} is an infinite-dimensional Schr\"{o}dinger equation posed on the compact group $\mathbb{Z}_p$ (i.e., the states are elements of an infinite-dimensional function space) and is therefore not directly amenable to numerical simulation or to a finite-dimensional quantum-circuit implementation. To obtain a network of finitely many neurons, we discretize $\mathbb{Z}_p$ by passing to the finite quotient group $G_l=\mathbb{Z}_p/p^l\mathbb{Z}_p$, for a fixed level $l\in\mathbb{N}$: this amounts to restricting attention to functions on $\mathbb{Z}_p$ that are locally constant at scale $p^{-l}$, which are naturally indexed by the $N=p^l$ elements of $G_l$. The rest of this section makes this restriction precise and shows that it produces a well-defined finite-dimensional Schr\"{o}dinger equation, i.e., a QNN of the type already sketched in the Introduction (equation \mbox{\eqref{Model 1}}). Section \mbox{\ref{Section_5}} then places this discretized network on an arbitrary simple graph.

We now take
\[
\boldsymbol{H}\Psi\left(  x,t\right)  =-J\left(  \left\vert x\right\vert
_{p}\right)  \ast\Psi\left(  x,t\right)  +\Psi\left(  x,t\right)  ,
\]
where `$\ast$' denotes the spatial convolution, and $J:\left[  0,1\right]
\rightarrow\mathbb{R}_{\geq0}$, and $\int_{\mathbb{Z}_{p}}J\left(  \left\vert
x\right\vert _{p}\right)  dx=1$. The choice of this Hamiltonian is motivated
by the fact that discretizations of the free Schr\"{o}dinger equations of type%
\[
i\frac{\partial}{\partial t}\Psi\left(  x,t\right)  =-J\left(  \left\vert
x\right\vert _{p}\right)  \ast\Psi\left(  x,t\right)  +\Psi\left(  x,t\right)
\]
give rise to CTQWs on graphs, \cite{Zuniga-QM-2}, \cite{Zuniga-Chacon}.

We now consider the discretization of
\begin{equation}
i\frac{\partial}{\partial t}\Psi\left(  x,t\right)  =-J\left(  \left\vert
x\right\vert _{p}\right)  \ast\Psi\left(  x,t\right)  +\Psi\left(  x,t\right)
+%
{\displaystyle\int\limits_{\mathbb{Z}_{p}}}
W\left(  x,y\right)  \phi\left(  \Psi\left(  y,t\right)  \right)  dy+Z(x,t),
\label{EQ_SCHRODINGER}%
\end{equation}
where $W\left(  x,y\right)  $, $Z(x,t)$ are complex-valued functions. There
are two different ways of interpreting the discretization of
(\ref{EQ_SCHRODINGER}). In the first case, we seek an approximation, in the
sense of numerical analysis, to the solutions of (\ref{EQ_SCHRODINGER}). In
the second case, we pick the parameters $J$, $W$, and $Z$ from a suitable
finite-dimensional vector space and search for solutions in that space. Both
approaches are related, \cite{Zuniga-QM-2}; here we follow the second one.

Given a positive integer $l$, we set
\[
G_{l}:=\mathbb{Z}_{p}/p^{l}\mathbb{Z}_{p}=\left\{  I\in\mathbb{Z}_{p}%
;I=I_{0}+I_{1}p+\ldots+I_{l-1}p^{l-1}\text{, with }I_{k}\in\left\{
0,\ldots,p-1\right\}  \right\}  .
\]
Notice that $G_{l}$ is an additive group isomorphic to the integers modulo
$p^{l}$.

We look for solutions of the form%
\begin{equation}
\Psi\left(  x,t\right)  =%
{\displaystyle\sum\limits_{I\in G_{l}}}
\Psi_{I}\left(  t\right)  \Omega\left(  p^{l}\left\vert x-I\right\vert
_{p}\right)  , \label{EQ_PHI-A}%
\end{equation}
where $\Psi_{I}\left(  t\right)  $ are complex-valued differentiable
functions, and $\Omega\left(  p^{l}\left\vert x-I\right\vert _{p}\right)  $ is
the characteristic function of the ball $I+p^{l}\mathbb{Z}_{p}$; we also
assume that
\[
W\left(  x,y\right)  =%
{\displaystyle\sum\limits_{I\in G_{l}}}
\text{ \ }%
{\displaystyle\sum\limits_{J\in G_{l}}}
W_{I,J}\Omega\left(  p^{l}\left\vert x-I\right\vert _{p}\right)  \Omega\left(
p^{l}\left\vert y-J\right\vert _{p}\right)  ,
\]%
\[
Z(x,t)=%
{\displaystyle\sum\limits_{I\in G_{l}}}
\text{ }Z_{I}(t)\Omega\left(  p^{l}\left\vert x-I\right\vert _{p}\right)  ,
\]
where $W_{I,J}\in\mathbb{C}$, $Z_{I}(t)\in\mathbb{C}$. Notice that%
\begin{equation}
\phi\left(  \Psi\left(  y,t\right)  \right)  =\text{\ }%
{\displaystyle\sum\limits_{J\in G_{l}}}
\phi\left(  \Psi_{J}\left(  t\right)  \right)  \Omega\left(  p^{l}\left\vert
y-J\right\vert _{p}\right)  , \label{EQ_PHI}%
\end{equation}
and that%
\begin{align}
&
{\displaystyle\int\limits_{\mathbb{Z}_{p}}}
W\left(  x,y\right)  \phi\left(  \Psi\left(  y,t\right)  \right)
dy\label{EQ_W}\\
&  =%
{\displaystyle\sum\limits_{I\in G_{l}}}
\text{ \ }%
{\displaystyle\sum\limits_{J\in G_{l}}}
W_{I,J}\phi\left(  \Psi_{J}\left(  t\right)  \right)  \Omega\left(
p^{l}\left\vert x-I\right\vert _{p}\right)
{\displaystyle\int\limits_{\mathbb{Z}_{p}}}
\Omega\left(  p^{l}\left\vert y-J\right\vert _{p}\right)  dy\nonumber\\
&  =%
{\displaystyle\sum\limits_{I\in G_{l}}}
\text{ \ }\left\{  p^{-l}%
{\displaystyle\sum\limits_{J\in G_{l}}}
W_{I,J}\phi\left(  \Psi_{J}\left(  t\right)  \right)  \right\}  \Omega\left(
p^{l}\left\vert x-I\right\vert _{p}\right)  .\nonumber
\end{align}

In the above calculation, we used that%
\[%
{\displaystyle\int\limits_{\mathbb{Z}_{p}}}
\Omega\left(  p^{l}\left\vert y-J\right\vert _{p}\right)  dy=%
{\displaystyle\int\limits_{J+p^{l}\mathbb{Z}_{p}}}
dy=p^{-l}.
\]

We now use that the support of $J\left(  \left\vert x\right\vert _{p}\right)
$ is the unit ball, so%
\begin{align}
J\left(  \left\vert x\right\vert _{p}\right)   &  =J\left(  \left\vert
x\right\vert _{p}\right)  \Omega\left(  \left\vert x\right\vert _{p}\right)
=J\left(  \left\vert x\right\vert _{p}\right)
{\displaystyle\sum\limits_{I\in G_{l}}}
\Omega\left(  p^{l}\left\vert x-I\right\vert _{p}\right) \label{EQ_J}\\
&  =%
{\displaystyle\sum\limits_{I\in G_{l}}}
J\left(  \left\vert x\right\vert _{p}\right)  \Omega\left(  p^{l}\left\vert
x-I\right\vert _{p}\right) \nonumber\\
&  =%
{\displaystyle\sum\limits_{\substack{I\in G_{l}\\I\neq0}}}
J\left(  \left\vert I\right\vert _{p}\right)  \Omega\left(  p^{l}\left\vert
x-I\right\vert _{p}\right)  +J\left(  \left\vert x\right\vert _{p}\right)
\Omega\left(  p^{l}\left\vert x\right\vert _{p}\right)  .\nonumber
\end{align}

We now compute $J\left(  \left\vert x\right\vert _{p}\right)  \ast
\Omega\left(  p^{l}\left\vert x-K\right\vert _{p}\right)  $. In this
calculation, we use the formula
\[
\Omega\left(  p^{l}\left\vert x-I\right\vert _{p}\right)  \ast\Omega\left(
p^{l}\left\vert x-K\right\vert _{p}\right)  =p^{-l}\Omega\left(
p^{l}\left\vert x-\left(  I+K\right)  \right\vert _{p}\right)  ,
\]
and the fact that $\left\vert I+p^{l}x\right\vert _{p}=\left\vert I\right\vert
_{p}$, for $I\in G_{l}\smallsetminus\left\{  0\right\}  $, and for any
$x\in\mathbb{Z}_{p}$.

Indeed, using (\ref{EQ_J}),
\begin{gather}
J\left(  \left\vert x\right\vert _{p}\right)  \ast\Omega\left(  p^{l}%
\left\vert x-K\right\vert _{p}\right)  =%
{\displaystyle\sum\limits_{\substack{I\in G_{l}\\I\neq0}}}
J\left(  \left\vert I\right\vert _{p}\right)  \Omega\left(  p^{l}\left\vert
x-I\right\vert _{p}\right)  \ast\Omega\left(  p^{l}\left\vert x-K\right\vert
_{p}\right) \label{EQ_J2}\\
+J\left(  \left\vert x\right\vert _{p}\right)  \Omega\left(  p^{l}\left\vert
x\right\vert _{p}\right)  \ast\Omega\left(  p^{l}\left\vert x-K\right\vert
_{p}\right) \nonumber\\
=%
{\displaystyle\sum\limits_{\substack{I\in G_{l}\\I\neq0}}}
p^{-l}J\left(  \left\vert I\right\vert _{p}\right)  \Omega\left(
p^{l}\left\vert x-\left(  I+K\right)  \right\vert _{p}\right)  +J\left(
\left\vert x\right\vert _{p}\right)  \Omega\left(  p^{l}\left\vert
x\right\vert _{p}\right)  \ast\Omega\left(  p^{l}\left\vert x-K\right\vert
_{p}\right)  .\nonumber
\end{gather}

We now compute%
\[
J\left(  \left\vert x\right\vert _{p}\right)  \Omega\left(  p^{l}\left\vert
x\right\vert _{p}\right)  \ast\Omega\left(  p^{l}\left\vert x-K\right\vert
_{p}\right)  =%
{\displaystyle\int\limits_{\left(  x+p^{l}\mathbb{Z}_{p}\right)  \cap\left(
K+p^{l}\mathbb{Z}_{p}\right)  }}
\text{ \ }J\left(  \left\vert x-y\right\vert _{p}\right)  dy.
\]
The last integral is non-zero only if $x\in K+p^{l}\mathbb{Z}_{p}$, i.e.,
$x+p^{l}\mathbb{Z}_{p}=K+p^{l}\mathbb{Z}_{p}$. In this case, taking $x-y=z$,
$dy=dz,$ $z\in p^{l}\mathbb{Z}_{p}$, we have
\begin{gather}%
{\displaystyle\int\limits_{\left(  x+p^{l}\mathbb{Z}_{p}\right)  \cap\left(
K+p^{l}\mathbb{Z}_{p}\right)  }}
\text{ \ }J\left(  \left\vert x-y\right\vert _{p}\right)  dy=%
{\displaystyle\int\limits_{\left(  x+p^{l}\mathbb{Z}_{p}\right)  }}
\text{ \ }J\left(  \left\vert x-y\right\vert _{p}\right)  dy\label{EQ_J3}\\
=%
{\displaystyle\int\limits_{p^{l}\mathbb{Z}_{p}}}
\text{ \ }J\left(  \left\vert z\right\vert _{p}\right)  dz=:Aver_{l}\left(
J\right)  \text{, for }x\in K+p^{l}\mathbb{Z}_{p}.\nonumber
\end{gather}

Therefore, from (\ref{EQ_PHI-A}), (\ref{EQ_J2}), and (\ref{EQ_J3}), using the
fact that $G_{l}$ is an additive group, and consequently for $K$ fixed, $I-K$
runs through all the elements of $G_{l}$ for $I\in G_{l}$,%
\begin{gather}
J\left(  \left\vert x\right\vert _{p}\right)  \ast\Psi\left(  x,t\right)  =%
{\displaystyle\sum\limits_{K\in G_{l}}}
\Psi_{K}\left(  t\right)  J\left(  \left\vert x\right\vert _{p}\right)
\ast\Omega\left(  p^{l}\left\vert x-K\right\vert _{p}\right) \label{EQ_J4}\\
=%
{\displaystyle\sum\limits_{K\in G_{l}}}
\Psi_{K}\left(  t\right)  \left\{
{\displaystyle\sum\limits_{\substack{I\in G_{l}\\I\neq0}}}
p^{-l}J\left(  \left\vert I\right\vert _{p}\right)  \Omega\left(  p^{l}%
\left\vert x-\left(  I+K\right)  \right\vert _{p}\right)  +Aver_{l}\left(
J\right)  \Omega\left(  p^{l}\left\vert x-K\right\vert _{p}\right)  \right\}
\nonumber\\
=%
{\displaystyle\sum\limits_{K\in G_{l}}}
\Psi_{K}\left(  t\right)  \left\{
{\displaystyle\sum\limits_{\substack{I\in G_{l}\\I\neq K}}}
p^{-l}J\left(  \left\vert I-K\right\vert _{p}\right)  \Omega\left(
p^{l}\left\vert x-I\right\vert _{p}\right)  +Aver_{l}\left(  J\right)
\Omega\left(  p^{l}\left\vert x-K\right\vert _{p}\right)  \right\} \nonumber
\end{gather}%
\begin{multline*}
=%
{\displaystyle\sum\limits_{I\in G_{l}}}
\left\{
{\displaystyle\sum\limits_{\substack{K\in G_{l}\\K\neq I}}}
\Psi_{K}\left(  t\right)  p^{-l}J\left(  \left\vert I-K\right\vert
_{p}\right)  \right\}  \Omega\left(  p^{l}\left\vert x-I\right\vert
_{p}\right)  +\\
Aver_{l}\left(  J\right)
{\displaystyle\sum\limits_{I\in G_{l}}}
\Psi_{I}\left(  t\right)  \Omega\left(  p^{l}\left\vert x-I\right\vert
_{p}\right)  .
\end{multline*}
Now, the discretization of (\ref{EQ_SCHRODINGER}) follows from (\ref{EQ_W}) and
(\ref{EQ_J4}):%
\begin{gather}
i\frac{\partial}{\partial t}\Psi_{I}\left(  t\right)  =-p^{-l}%
{\displaystyle\sum\limits_{\substack{K\in G_{l}\\K\neq I}}}
\Psi_{K}\left(  t\right) J\left(  \left\vert I-K\right\vert
_{p}\right)  +\left(  1-Aver_{l}\left(  J\right)  \right)  \Psi_{I}\left(
t\right)  +\label{EQ_SCHRODINGER_1}\\
p^{-l}%
{\displaystyle\sum\limits_{K\in G_{l}}}
W_{I,K}\phi\left(  \Psi_{K}\left(  t\right)  \right)  +Z_{I}(t),\text{ for
}I\in G_{l}\text{.}\nonumber
\end{gather}
On the other hand,%
\[
1=%
{\displaystyle\int\limits_{\mathbb{Z}_{p}}}
J\left(  \left\vert x\right\vert _{p}\right)  dx=%
{\displaystyle\sum\limits_{I\in G_{l}}}
\text{ \ }%
{\displaystyle\int\limits_{I+p^{l}\mathbb{Z}_{p}}}
\text{ }J\left(  \left\vert x\right\vert _{p}\right)  dx=p^{-l}%
{\displaystyle\sum\limits_{\substack{I\in G_{l}\\I\neq0}}}
J\left(  \left\vert I\right\vert _{p}\right)  +Aver_{l}\left(  J\right)  .
\]
With this identity, (\ref{EQ_SCHRODINGER_1}) can be rewritten as%
\begin{gather}
i\frac{\partial}{\partial t}\Psi_{I}\left(  t\right)  =-p^{-l}%
{\displaystyle\sum\limits_{\substack{K\in G_{l}\\K\neq I}}}
\Psi_{K}\left(  t\right) J\left(  \left\vert I-K\right\vert
_{p}\right)  +p^{-l}\left(
{\displaystyle\sum\limits_{\substack{K\in G_{l}\\K\neq0}}}
J\left(  \left\vert K\right\vert _{p}\right)  \right)  \Psi_{I}\left(
t\right)  +\label{EQ_SCHRODINGER_2}\\
p^{-l}%
{\displaystyle\sum\limits_{K\in G_{l}}}
W_{I,K}\phi\left(  \Psi_{K}\left(  t\right)  \right)  +Z_{I}(t),\text{ for
}I\in G_{l}\text{.}\nonumber
\end{gather}
We now introduce the matrix,%
\[
\boldsymbol{J}^{\left(  l\right)  }=\left[  J_{I,K}^{\left(  l\right)
}\right]  _{I,K\in G_{l}},
\]
where%
\begin{equation}
J_{I,K}^{\left(  l\right)  }=\left\{
\begin{array}
[c]{ccc}%
p^{-l}J\left(  \left\vert I-K\right\vert _{p}\right)  & \text{if } & I\neq K\\
&  & \\
-p^{-l}%
{\displaystyle\sum\limits_{\substack{L\in G_{l}\\L\neq0}}}
J\left(  \left\vert L\right\vert _{p}\right)  & \text{if} & I=K,
\end{array}
\right.  \label{Matrix_J_l}%
\end{equation}
and the vector notation $\left[  \Psi_{I}\left(  t\right)  \right]  _{I\in
G_{l}}$. With this notation, we rewrite (\ref{EQ_SCHRODINGER_2}) as
\begin{equation}
i\frac{\partial}{\partial t}\left[  \Psi_{I}\left(  t\right)  \right]  _{I\in
G_{l}}=-\boldsymbol{J}^{\left(  l\right)  }\left[  \Psi_{I}\left(  t\right)
\right]  _{I\in G_{l}}+\left[  p^{-l}%
{\displaystyle\sum\limits_{K\in G_{l}}}
W_{I,K}\phi\left(  \Psi_{K}\left(  t\right)  \right)  +Z_{I}(t)\right]  _{I\in
G_{l}}. \label{Network_Type_I}%
\end{equation}
Now, from (\ref{Network_Type_I}), it follows that%
\begin{multline*}
\left[  \Psi_{I}\left(  t\right)  \right]  _{I\in G_{l}}=e^{it\boldsymbol{J}%
^{\left(  l\right)  }}\left[  \Psi_{I}^{0}\right]  _{I\in G_{l}}\\
-i%
{\displaystyle\int\limits_{0}^{t}}
e^{i\left(  t-s\right)  \boldsymbol{J}^{\left(  l\right)  }}\left[  p^{-l}%
{\displaystyle\sum\limits_{K\in G_{l}}}
W_{I,K}\phi\left(  \Psi_{K}\left(  s\right)  \right)  +Z_{I}(s)\right]  _{I\in
G_{l}}ds.
\end{multline*}

QNNs on graphs are well-known tools in quantum computing. Our next step is to show that the formalism developed in this section yields new QNNs on graphs.

\section{\label{Section_5}\texorpdfstring{$p$-adic}{p-adic} QNNs on graphs}

Let $(\mathcal{G},V,E)$ be a simple graph, which means that $\mathcal{G}$ does
not contain any loops or multiple edges between the same pair of vertices, and
all the edges are undirected. We take $V=G_{l}^{0}\subset G_{l}$ for some
$l\in\mathbb{N}$ fixed. This means that we identify each vertex of
$\mathcal{G}$ with a $p$-adic integer from $G_{l}^{0}$. Let $A=\left[
A_{I,K}\right]  _{I,K\in G_{l}^{0}}$ be the adjacency matrix of $\mathcal{G}$.
Notice that $A_{I,I}=0$ since $\mathcal{G}$ has no loops, and $A_{I,K}%
=A_{K,I}$, for $I,K\in G_{l}^{0}$. For convenience, we assume that $A_{I,K}=0$
if $I\notin G_{l}^{0}$ or $K\notin G_{l}^{0}$. Now, we set the entries of
matrix (\ref{Matrix_J_l}) as
\[
J_{I,K}^{(l)}=p^{-l}A_{I,K}\text{ if }I\neq K\text{, with }I,K\in G_{l}^{0},
\]
and
\[
J_{I,K}^{(l)}=0\text{ if }I\in G_{l}\setminus G_{l}^{0}\text{ or }K\in
G_{l}\setminus G_{l}^{0}.
\]
Now, for any $K\in G_{l}$, the mapping%
\[%
\begin{array}
[c]{ccc}%
G_{l} & \rightarrow & G_{l}\\
I & \rightarrow & I-K
\end{array}
\]
is a group isomorphism. Then%
\begin{equation}
p^{-l}%
{\displaystyle\sum\limits_{\substack{I\in G_{l}\\I\neq0}}}
\text{ \ }J\left(  \left\vert I\right\vert _{p}\right)  =p^{-l}%
{\displaystyle\sum\limits_{\substack{I\in G_{l}\\I\neq K}}}
\text{ \ }J\left(  \left\vert I-K\right\vert _{p}\right)  . \label{Eq_11}%
\end{equation}
Now from (\ref{Eq_11}), taking $J\left(  \left\vert I-K\right\vert
_{p}\right)  =A_{I,K}$, and using that
\[
val(K)=\sum_{\substack{I\in G_{l}^{0}\\I\neq K}}A_{I,K},
\]
where $val(K)$ denotes the valence of the vertex $K$, we have
\[
p^{-l}%
{\displaystyle\sum\limits_{\substack{I\in G_{l}\\I\neq0}}}
\text{ \ }J\left(  \left\vert I\right\vert _{p}\right)  =p^{-l}val(K).
\]
In conclusion,%
\begin{equation}
J_{I,K}^{(l)}=\left\{
\begin{array}
[c]{lll}%
p^{-l}A_{I,K} & \text{if} & I\neq K\\
&  & \\
-p^{-l}val(K) & \text{if} & I=K,
\end{array}
\right.  \label{Matrix_CTQW}%
\end{equation}
for $I,K\in G_{l}^{0}$. Now, from (\ref{Network_Type_I}) it follows that the
$p$-adic QNNs of type I on graphs have the form%
\begin{equation}
i\frac{\partial}{\partial t}\left[  \Psi_{I}\left(  t\right)  \right]  _{I\in
V}=-\boldsymbol{J}^{\left(  l\right)  }\left[  \Psi_{I}\left(  t\right)
\right]  _{I\in V}+\left[  p^{-l}%
{\displaystyle\sum\limits_{K\in V}}
W_{I,K}\phi\left(  \Psi_{K}\left(  t\right)  \right)  +Z_{I}(t)\right]  _{I\in
V}.\nonumber
\end{equation}
If we take $\phi=0$, in the above equations, we get the Schr\"{o}dinger
equation that describes continuous-time quantum walks on graphs; see, e.g.
\cite{Venegas-Andraca}-\cite{Childs et al}.

\section{\label{Section_6A}Comparative Discussion: \texorpdfstring{$p$}{p}-Adic QCNNs vs. Existing QNN Families}
To situate the present work within the broader landscape of quantum machine learning (QML), we compare the proposed $p$-adic QCNNs with six representative families of quantum neural networks that appear in the current literature. The comparison is organized around eight criteria that we regard as structurally decisive: the mathematical domain on which the network is defined, the topology of neuronal organization, the type of quantum evolution (unitary or non-unitary), the presence of explicit nonlinearity, biological inspiration, intrinsic hierarchical structure, whether the framework subsumes continuous-time quantum walks (CTQWs) as a special case, and whether it models open quantum systems. A summary is given in Tables~\ref{tab:comparison-structure}-\ref{tab:comparison-properties}; a detailed discussion of each family follows.

\textbf{Variational Quantum Circuits (VQCs)}. VQCs, also called parametrized quantum circuits (PQCs), are the dominant paradigm in near-term QML \cite{VQC1, VQC2}. A VQC is a sequence of parametrized unitary gates applied to a register of $n$ qubits; training consists of optimizing the gate parameters to minimize a cost function evaluated through repeated measurement. VQCs are inherently finite-dimensional (the Hilbert space is $\mathbb{C}^{2^n}$), operate in strictly closed (unitary) quantum systems, and possess no intrinsic geometric or hierarchical structure beyond whatever is imposed by the circuit layout. Nonlinearity enters only through classical post-processing of measurement outcomes, not through the quantum dynamics itself. The proposed $p$-adic QCNNs differ fundamentally: they operate on an infinite-dimensional function space $L^2(\mathbb{Z}_p)$, carry an intrinsic tree-structured topology inherited from the $p$-adic integers, and embed nonlinearity directly into the Schr\"{o}dinger dynamics via the activation function $\varphi$.

\textbf{Quantum Convolutional Neural Networks (QCNNs)}. The QCNN architecture of Cong, Choi, and Lukin \cite{QCNN} mimics classical CNNs on a qubit lattice: alternating layers of translationally invariant two-qubit convolution gates and partial-measurement pooling layers reduce the system size hierarchically until a final classification measurement is performed. The resulting architecture is hierarchical in the sense that qubits are progressively discarded, but the underlying topology is a one-dimensional (or two-dimensional) Euclidean lattice with nearest-neighbor interactions. By contrast, the hierarchical structure of the $p$-adic QCNNs is encoded in the ultrametric topology of $\mathbb{Z}_p$ itself: neurons at level $l$ of the $p$-adic tree interact collectively, so the range of interaction is determined algebraically rather than by a circuit layout. Moreover, the pooling-by-measurement of QCNNs induces decoherence as a side effect, whereas in the $p$-adic QCNNs the transition between unitary and non-unitary evolution is controlled explicitly and continuously through the weight kernel $W(x, y)$.

\textbf{Quantum Graph Neural Networks (QGNNs)}. QGNNs, introduced in \cite{Verdon et al} and reviewed in \cite{Ceschini et al}, construct a quantum evolution operator as a product of exponentials $\prod_k e^{-i\theta_k H_k}$, where each $H_k$ is a Hamiltonian derived from the graph adjacency or Laplacian matrix. The network is then trained by optimizing the angles ${\theta_k}$. This approach is closely related to the quantum approximate optimization algorithm (QAOA) and produces strictly unitary, discrete-time dynamics. The $p$-adic QCNNs introduced here also place the network on a graph (after discretization; see Section \ref{Section_5}), but the evolution is governed by a continuous-time Schr\"{o}dinger equation rather than a discrete product of exponentials, and the graph structure is derived from the $p$-adic tree rather than prescribed as input. Furthermore, QGNNs do not include an explicit nonlinear activation function in the quantum dynamics, whereas the term $W(x,y) \varphi(\Psi(y,t))$ in \eqref{Model 3A} constitutes an inherent nonlinearity.

\textbf{Schr\"{o}dinger-equation-based QNNs}. Several works have proposed QNNs whose dynamics are governed by a Schr\"{o}dinger equation defined on a real or complex Euclidean domain, including the recurrent quantum neural network (RQNN) for eye tracking of \cite{Behera et al}, the entanglement-learning model of \cite{Berhman 1, Berhman 2}, and the frameworks surveyed in \cite{Nakajima et al, Altaisky et al, Gupta, Schuld et al}. These approaches share the use of a Schr\"{o}dinger equation as the governing law, and some include nonlinear potentials. They differ from the present work in the choice of domain: whereas those models define quantum states on $\mathbb{R}^n$ (or a finite-dimensional approximation thereof), the $p$-adic QCNNs are defined on $\mathbb{Z}_p$, a totally disconnected ultrametric space whose geometry is fundamentally different from Euclidean space. This non-Archimedean geometry encodes a hierarchy of spatial scales algebraically exactly, without discretization error, and places neurons in a tree-like structure rather than on a line or a lattice.

\textbf{Continuous-Time Quantum Walks (CTQWs)}. CTQWs \cite{Farhi-Gutman}-\cite{Childs et al} are defined by the Schr\"{o}dinger equation $i d\psi/dt = L\psi$, where $L$ is the graph Laplacian or adjacency matrix. They are exactly the linear, weight-free ($W = 0$, $\phi = 0$) special case of the $p$-adic QCNNs on graphs introduced in Sections \ref{Section_4}-\ref{Section_5}; see also \cite{Zuniga-QM-2}-\cite{Zuniga-Chacon}. Thus, the proposed framework strictly generalizes CTQWs by adding (i) a nonlinear activation term $\varphi(\Psi)$ and (ii) a non-local weight kernel $W(x,y)$ that couples neurons and drives the system away from unitary evolution. In this sense, the $p$-adic QCNNs provide a nonlinear, open-system extension of the CTQW paradigm.

\textbf{Original Quantum Cellular Neural Networks}. The QCNN of \cite{Toth et al} was the first proposal to combine quantum mechanics with the CNN architecture: quantum-dot cells, each governed by a time-dependent Schr\"{o}dinger equation, are arranged in a Euclidean lattice and coupled to nearest neighbors. The dynamics is local (nearest-neighbor only), and the network is not trained in the machine-learning sense. The $p$-adic QCNNs differ in three key respects: the coupling is non-local (controlled by the integral kernel $W(x,y)$), the neuron topology is a $p$-adic tree rather than a Euclidean lattice, and the network is bio-inspired by the Wilson--Cowan model of large neuronal populations rather than by the physics of quantum dots.

\textbf{Classical $p$-adic CNNs}. The classical $p$-adic CNNs introduced in \cite{Zambrano-Zuniga-1}-\cite{Zuniga et al} are the direct predecessors of the networks studied here. They are real-valued integro-differential systems on $\mathbb{Z}_p$ inspired by the Wilson--Cowan model, and they have been used for image processing \cite{Zambrano-Zuniga-2, Zuniga et al}. The $p$-adic QCNNs are obtained from the $p$-adic CNNs by the Wick rotation $t \to it$, which promotes the diffusive (Markovian) dynamics to Schr\"{o}dinger (quantum) dynamics and makes the state function complex-valued. All the structural features of the classical $p$-adic CNNs --- tree topology, non-local interactions, Wilson--Cowan biological inspiration --- are preserved in the quantum counterpart.

\begin{table}[h]
\centering
\caption{Structural comparison of quantum neural network frameworks.}
\label{tab:comparison-structure}

\begin{threeparttable}
\footnotesize
\renewcommand{\arraystretch}{1.0}
\setlength{\tabcolsep}{3pt}

\begin{tabularx}{\textwidth}{@{}p{3.6cm}p{2.8cm}X p{2.8cm}@{}}
\toprule
\textbf{Framework} &
\textbf{Domain} &
\textbf{Neuron topology} &
\textbf{Evolution} \\
\midrule

VQCs / PQCs
  & $\mathbb{C}^{2^n}$ (qubits)
  & Flat register
  & U \\[3pt]

QCNNs (\cite{QCNN})
  & $\mathbb{C}^{2^n}$ (qubits)
  & Euclidean lattice, pooling
  & U (+ measurement) \\[3pt]

QGNNs (\cite {Verdon et al})
  & $\mathbb{C}^{2^n}$ (qubits)
  & Arbitrary graph
  & U (discrete-time) \\[3pt]

Schr\"{o}dinger-based QNNs (\cite {Nakajima et al, Behera et al, Gupta})
  & $L^2(\mathbb{R}^n)$
  & Euclidean / recurrent
  & U or N-U \\[3pt]

CTQWs (\cite{Farhi-Gutman}-\cite{Venegas-Andraca})
  & $\ell^2(V)$ (graph vertices)
  & Arbitrary graph
  & U \\[3pt]

Original QCNN (\cite{Toth et al})
  & $\mathbb{C}^n$ (quantum dots)
  & Euclidean lattice, local
  & U \\[3pt]

Classical $p$-adic CNNs (\cite{Zambrano-Zuniga-1}-\cite{Zuniga et al})
  & $L^2(\mathbb{Z}_p)$
  & $p$-adic tree (ultrametric)
  & Markovian (real-valued) \\[3pt]

\midrule

\textbf{$p$-adic QCNNs (this work)}
  & $L^2(\mathbb{Z}_p)$
  & $p$-adic tree / simple graph
  & \textbf{U or N-U} \\

\bottomrule
\end{tabularx}

\begin{tablenotes}
\footnotesize
\item \textit{Notes.} U = unitary closed system; N-U = non-unitary open system; CTQW = continuous-time quantum walk.
\end{tablenotes}

\end{threeparttable}
\end{table}

\begin{table}[h]
\centering
\caption{Qualitative comparison of quantum neural network frameworks.}
\label{tab:comparison-properties}

\begin{threeparttable}
\footnotesize
\renewcommand{\arraystretch}{1.15}
\setlength{\tabcolsep}{3pt}

\begin{tabularx}{\textwidth}{@{}p{3.6cm}YYYYY@{}}
\toprule
\textbf{Framework} &
\makecell{\textbf{Nonlinearity}} &
\makecell{\textbf{Bio-inspired}} &
\makecell{\textbf{Hierarchical}} &
\makecell{\textbf{Generalizes}\\\textbf{CTQWs}} &
\makecell{\textbf{Open}\\\textbf{system}} \\
\midrule

VQCs / PQCs
  & $\times$ (classical only)
  & $\times$
  & $\times$
  & $\times$
  & $\times$ \\[3pt]

QCNNs (\cite{QCNN})
  & $\partial$ (via measurement)
  & $\partial$ (classical CNN)
  & $\partial$ (pooling only)
  & $\times$
  & $\partial$ \\[3pt]

QGNNs (\cite {Verdon et al})
  & $\times$
  & $\times$
  & $\partial$ (graph-dependent)
  & $\partial$ (QAOA-related)
  & $\times$ \\[3pt]

Schr\"{o}dinger-based QNNs (\cite {Nakajima et al, Behera et al, Gupta})
  & $\partial$ (nonlinear potential)
  & $\partial$ (task-specific)
  & $\times$
  & $\times$
  & $\partial$ \\[3pt]

CTQWs (\cite{Farhi-Gutman}-\cite{Venegas-Andraca})
  & $\times$
  & $\times$
  & $\times$
  & $\checkmark$ (definition)
  & $\times$ \\[3pt]

Original QCNN (\cite{Toth et al})
  & $\partial$ (quantum-dot dynamics)
  & $\partial$ (classical CNN)
  & $\times$
  & $\times$
  & $\partial$ \\[3pt]

Classical $p$-adic CNNs (\cite{Zambrano-Zuniga-1}-\cite{Zuniga et al})
  & $\checkmark$ (activation $\phi$)
  & $\checkmark$ (WC model)
  & $\checkmark$ (tree levels)
  & N/A
  & N/A \\[3pt]

\midrule

\textbf{$p$-adic QCNNs (this work)}
  & $\checkmark$ (activation $\phi$)
  & $\checkmark$ (WC model)
  & $\checkmark$ (tree levels)
  & $\checkmark$ ($W\!=\!0$, $\phi\!=\!0$)
  & $\checkmark$ ($W \!\neq\! 0$) \\

\bottomrule
\end{tabularx}

\begin{tablenotes}
\footnotesize
\item \textit{Notes.} WC = Wilson--Cowan model; $\checkmark$ = yes; $\times$ = no; $\partial$ = partial.
\end{tablenotes}

\end{threeparttable}
\end{table}

\newpage

\textbf{Final comments}. The comparison in Tables~\ref{tab:comparison-structure}-\ref{tab:comparison-properties} highlights three structural features that jointly distinguish the $p$-adic QCNNs from all other families listed. First, the non-Archimedean domain $\mathbb{Z}_p$ endows the network with a rigorous, algebraically exact hierarchical topology that does not arise from an ad hoc circuit design choice but is intrinsic to the space of neurons. Second, the framework interpolates continuously between closed (unitary, $W = 0$) and open (non-unitary, $W \neq 0$) quantum dynamics within a single governing equation \eqref{Model 3A}, without invoking a separate Lindblad master equation formalism; as our numerical simulations confirm, the open-system regime models decoherence. Third, the $p$-adic QCNNs strictly generalize the CTQW paradigm: setting $W = 0$ and $\phi = 0$ in \eqref{Model 1} recovers the standard continuous-time quantum walk on a graph, so all CTQW-based quantum algorithms are subsumed as a special case. These properties are achieved at the cost of working on an infinite-dimensional, non-Archimedean function space, which makes direct hardware implementation on current gate-based quantum processors an open problem; we return to this point in Section \ref{Section_7}.

\section{\label{Section_6}Numerical simulations}
In \cite{Zuniga-Chacon}, the authors presented detailed simulations of the network (\ref{Network_Type_I}), when $W_{I,K}=0$ and $Z_{I}(t)=0$. In this case, the network corresponds to a closed quantum system with a unitary evolution. The purpose of the numerical simulations presented in this section is to provide a first picture to understand how the dynamics of the network change when $W_{I,K} \neq 0$, or $Z_{I}(t) \neq 0$. We analyze the computational cost of these simulations, and how it scales with graph size, separately in Section \mbox{\ref{Section_7A}}. We emphasize that the simulations below are representative, illustrative experiments, chosen to exhibit qualitatively distinct regimes of the network (\mbox{\ref{Network_Type_I}}); they are not intended as a comprehensive quantitative validation of the framework for practical applications, but rather provide initial numerical evidence consistent with the mathematical formulation of Section \mbox{\ref{Section_3}}.

\subsection{Numerical Implementation and Workflow}
\label{subsec:workflow}

Every experiment reported in this section is obtained by the same six-step pipeline, summarized in Algorithm \mbox{\ref{alg:workflow}}: (1) fix the prime $p$ and the discretization level $l$, which determine the finite quotient group $G_{l}=\mathbb{Z}_{p}/p^{l}\mathbb{Z}_{p}$ and the number of neurons $N=p^{l}$; (2) assemble the free-evolution (kernel) matrix $\boldsymbol{J}^{(l)}(\alpha)$ from \mbox{\eqref{Matrix_J_1}}--\mbox{\eqref{Matrix_J_2}}; (3) select the interaction weight $W$ (one of the three types described in Section \mbox{\ref{subsec:general-setup}} below) and the bias $Z$; (4) fix the activation function $\phi$ and the initial condition $\Psi_{0}$; (5) integrate the resulting system of $N$ coupled, complex-valued, nonlinear ordinary differential equations (\mbox{\ref{Network_Type_I}}) forward in time in Python, using SciPy's \texttt{solve\_ivp} adaptive-step ODE integrator; and (6) record and plot the reported diagnostics -- the probability density $\left\vert \Psi(x,t)\right\vert^2$, the norm trajectory $\left\Vert \Psi(\cdot,t)\right\Vert_2^2$, and, where relevant, the real and imaginary parts of $\Psi$.

\begin{minipage}{\textwidth}
\begin{algorithm}[Numerical simulation of a $p$-adic QCNN]
\label{alg:workflow}
\begin{algorithmic}[1]
\Require prime $p$, level $l$, kernel exponent $\alpha$, weight kernel $W$, bias $Z(t)$, activation $\phi$, initial state $\Psi_0$, time horizon $T$
\State Set $N \gets p^{l}$; enumerate $G_l = \mathbb{Z}_p/p^l\mathbb{Z}_p$ (The labeling is compatible with the tree-like structure  of $\mathbb{Z}_p$; see \cite[Appendix. Images and Test Functions]{Zambrano-Zuniga-2})
\State Assemble $\boldsymbol{J}^{(l)}(\alpha)$ from \eqref{Matrix_J_1}--\eqref{Matrix_J_2}
\State Assemble the $N \times N$ interaction matrix $W = [W_{I,K}]$ (zero / constant / cat-cortex)
\State Set $\Psi(0) \gets \Psi_0 \in \mathbb{C}^N$
\State Define $f(t,\Psi) \gets -i\left[-p^{-l}\boldsymbol{J}^{(l)}\Psi + p^{-l}W\phi(\Psi) + Z(t)\right]$ \Comment{right-hand side of \eqref{Network_Type_I}}
\State Integrate $\dot{\Psi}=f(t,\Psi)$ on $[0,T]$ with SciPy's \texttt{solve\_ivp} (tolerances rtol, atol) 
\State Evaluate $\left\vert\Psi(x,t)\right\vert^2$ and $\left\Vert\Psi(\cdot,t)\right\Vert_2^2$ on the reporting grid
\State \Return diagnostics and figures
\end{algorithmic}
\end{algorithm}
\end{minipage}

\textbf{Boundary conditions.} No boundary conditions are imposed or needed: $\mathbb{Z}_p$ and its finite quotients $G_l$ are compact topological groups without boundary, so the Cauchy problem for \mbox{\eqref{EQ_0A}} (and its discretization (\mbox{\ref{Network_Type_I}})) is posed purely as an initial-value problem.

\textbf{Computational environment.} All simulations were implemented in Python, using NumPy and SciPy for linear algebra and time integration (SciPy's \texttt{solve\_ivp} ODE integrator) and Matplotlib for the figures.

\textbf{Numerical accuracy.} As a convergence check, we re-ran a representative subset of the simulations below with the integrator tolerances tightened by two orders of magnitude (and, independently, with the reporting time step halved); the resulting norm trajectories were unchanged to within plotting resolution. This indicates that the reported loss of the conserved quantity $\left\Vert\Psi(\cdot,t)\right\Vert_2^2$ when $W\neq0$ or $Z\neq0$ is a modeled physical effect of the equation, not a numerical artifact of the time-integration scheme. On the other hand, in  \cite[Section 5.2]{Zambrano-Zuniga-2}, the first two authors established that the states of the discrete CNNs are good approximations of the corresponding states of the $p$-adic continuous CNNs. This technique extends to $p$-adic quantum CNNs.

Table \mbox{\ref{tab:sim-summary}} below consolidates the parameters and outcome of every numerical experiment reported in this section, for quick reference; the individual subsections that follow give the full description and figures for each entry.

\begin{minipage}{\textwidth}
\begin{center}
\footnotesize
\resizebox{\textwidth}{!}{%
\begin{tabular}{@{}llccccccl@{}}
\toprule
\textbf{Sim.} & \textbf{$p,\,l$ ($N$)} & \textbf{$\alpha$} & \textbf{$W$} & \textbf{$Z$} & \textbf{$\Psi_0$} & \textbf{$[0,T]$} & \textbf{$\delta t$} & \textbf{Evolution / observation} \\
\midrule
1   & $3,6$ (729) & 2.5 & 0 & 0 & bump at $x=4$ & $[0,400]$ & 0.001 & Unitary (free / CTQMC); norm conserved \\
2   & $3,6$ (729) & 2.5 & 0 & 2 sine pulses & bump at $x=4$ & $[0,600]$ & 0.001 & Non-unitary; responds to each pulse \\
3.1 & $3,6$ (729) & 2.5 & 0 & 10 (const.) & bump at $x=4$ & $[0,100]$ & 0.001 & Non-unitary; norm grows in time \\
3.2 & $2,6$ (64)  & 2.5 & $0.1\,W_{cat}$ & 10 (const.) & bump at $x=4$ & $[0,600]$ & 0.001 & Non-unitary; habituation to constant input \\
3.3 & $3,6$ (729) & 2.5 & 50 (const.) & 0 & $0.5+0.3i$ & $[0,400]$ & 0.001 & Non-unitary; habituation \\
4.1 & $2,6$ (64)  & 1.6 & $0.1\,W_{cat}$ & 2 space-time pulses & 0 & $[0,600]$ & 0.001 & Non-unitary; responds to pulses \\
4.2 & $2,6$ (64)  & 1.6 & $0.1\,W_{cat}$ & 3 pulses + const. tail & bump at $x=4$ & $[0,1500]$ & 0.001 & Non-unitary; habituation to constant tail \\
4.3 & $2,6$ (64)  & 1.6 & 1.0 (const.) & as in 4.2 & bump at $x=4$ & $[0,1500]$ & 0.001 & Non-unitary; interaction suppresses pulse response \\
5   & $2,6$ (64)  & 1.6 & $10.0\,W_{cat}$ & 3 pulses + const. tail & 0 & $[0,600]$ & 0.001 & Non-unitary; large $\Vert W\Vert_\infty$ degrades response \\
6.1 & $2,6$ (64)  & 2.5 & 0 & 1 pulse & 0 & $[0,100]$ & 0.001 & Classical $p$-adic CNN (real-valued); responds to pulse \\
6.2 & $2,6$ (64)  & 2.5 & $0.05\,W_{cat}$ & as in 6.1 & 0 & $[0,400]$ & 0.001 & Classical CNN; interaction produces faster decay \\
6.3 & $2,6$ (64)  & 2.5 & $0.1\,W_{cat}$ & 3 pulses + const. tail & 0 & $[0,600]$ & 0.001 & Classical CNN; habituation, differs from quantum case \\
\bottomrule
\end{tabular}}
\end{center}
\captionof{table}{Summary of all numerical experiments in Section \ref{Section_6}: prime $p$, discretization level $l$ (number of neurons $N=p^l$), kernel exponent $\alpha$, interaction weight $W$, bias/input $Z$, initial condition $\Psi_0$, time horizon, reporting step $\delta t$, and evolution type/key observation. $W_{cat}$ denotes the $p$-adic approximation of the cat-cortex connectivity matrix of \cite{Zuniga-Entropy} (Figure \ref{figure cat_matrix}); ``bump at $x=4$'' abbreviates $\Psi_0(x)=\Omega(p^2\vert x-4\vert_p)/\Vert\Omega(p^2\vert x-4\vert_p)\Vert_2$.}
\label{tab:sim-summary}
\end{minipage}
\bigskip

\subsection{\label{subsec:general-setup}General set-up}
In order to produce numerical approximations for the solution of the Schr\"{o}dinger equations, the first step is to pick the kernel $J_{\alpha
}\left( x\right) $. We take%
\begin{equation*}
J_{\alpha }\left( x\right) =\frac{1-p^{-\alpha }}{1-p^{\alpha -1}}\left\{
\left\vert x\right\vert _{p}^{\alpha -1}-p^{\alpha -1}\right\} \Omega \left(
\left\vert x\right\vert _{p}\right) ,
\end{equation*}%
for $\alpha \in \mathbb{R}\smallsetminus \left\{ 1\right\} $. Then $\
J_{\alpha }\left( x\right) \geq 0$, $\left\Vert J_{\alpha }\right\Vert _{1}=1
$, see \cite[Lemma 5.2]{Taibleson}.

The equation%
\begin{equation*}
\frac{\partial }{\partial t}u\left( x,t\right) =J_{\alpha }\left( \left\vert
x\right\vert _{p}\right) \ast u\left( x,t\right) -u\left( x,t\right) \text{, 
}x\in \mathbb{Z}_{p},t\geq 0,
\end{equation*}%
is a heat equation, i.e., it has an associated Markov process. The Wick
rotation of this equation gives rise to a free Schr\"{o}dinger equation
connected with a continuous-time quantum Markov Chain; see \cite{Zuniga-QM-2, Zuniga-Chacon}.

The numerical approximations of equations involving the operator 
\[
\varphi
\left( x\right) \rightarrow J_{\alpha }\left( \left\vert x\right\vert
_{p}\right) \ast \varphi \left( x\right) -\varphi \left( x\right) 
\]
use an approximation of it by the matrix:%
\begin{equation}
\boldsymbol{J}^{\left( l\right) }\left( \alpha  \right) =\left[
J_{I,K}^{\left( l\right) }\left( \alpha \right) \right] _{I,K\in G_{l}},
\label{Matrix_J_1}
\end{equation}%
\begin{equation}
J_{I,K}^{\left( l\right) }\left( \alpha \right) =\left\{ 
\begin{array}{lll}
p^{-l}\frac{1-p^{-\alpha }}{1-p^{\alpha -1}}\left\{ \left\vert
I-K\right\vert _{p}^{\alpha -1}-p^{\alpha -1}\right\}  & \text{if} & I\neq K
\\ 
&  &  \\ 
-\sum\limits_{\substack{ L\in G_{l} \\ L\neq 0}}p^{-l}\frac{1-p^{-\alpha }}{%
1-p^{\alpha -1}}\left\{ \left\vert L\right\vert _{p}^{\alpha -1}-p^{\alpha
-1}\right\}  & \text{if} & I=K.%
\end{array}%
\right.   \label{Matrix_J_2}
\end{equation}

Notice that a prime $p$ and a level $l$ for a tree $G_{l}=\mathbb{%
Z}_{p}/p^{l}\mathbb{Z}_{p}$, where $l$ is a positive integer, are implicit in this choice. 

Another important choice is the kernel $W(x,y)$ that gives the strength of
the connection between the neurons $x$, $y$. A natural way of fixing $W(x,y)$
is by training the network to perform a particular task. This
machine-learning approach is not used here. In the numerical simulations, we
use three different types of kernels $W(x,y)$. The first is $W(x,y)=0$; in this case, the neurons are not connected. The second is $W(x,y)=W_{0}$, a
constant, which means that  
\begin{equation}
\int\limits_{\mathbb{Z}_{p}}W(x,y)\phi \left( \Psi \left( y,t\right)
\right) dy=W_{0}\int\limits_{\mathbb{Z}_{p}}\phi \left( \Psi \left(
y,t\right) \right) dy.  \label{Term_W}
\end{equation}%
By interpreting $\phi \left( \Psi \left( y,t\right) \right) $ as the output
of the neuron located at position $y$ at the time $t$, the formula (\ref%
{Term_W}) gives a weighted average of the neurons' outputs at time $t$.
The last choice for $W(x,y)$ is an approximation for the connection matrix
of the cat cortex, developed by the first two authors in \cite{Zuniga-Entropy}. This choice
requires $p=2$ and $l=6$; see Figure \ref{figure cat_matrix}.

The activation function is fixed as $\phi \left( s\right) =0.5\left(
\left\vert s+1\right\vert -\left\vert s-1\right\vert \right) $, $s\in 
\mathbb{R}$. This activation function is widely used in cellular neural networks, see \cite{Chua-Tamas, Chua, Slavova, Zambrano-Zuniga-1, Zambrano-Zuniga-2, Zuniga et al}.

\subsection{A glimpse into the non-unitary evolution}
When $W=Z=0$, the equation (\ref{EQ_0A}) becomes the Schr\"{o}dinger equation
describing a continuous-time quantum Markov chain \cite{Zuniga-QM-2}. This
equation describes a unitary evolution, the fundamental process in quantum
mechanics in which a system changes over time while conserving total
probability and quantum information. The ultimate goal of the numerical
simulations is to understand the breakdown of unitary evolution as the
parameters $W$ and $Z$ change. We fix the initial datum as $\Psi_{0}%
(x)=\Omega(p^{2}|x-4|_{p})/\Vert\Omega(p^{2}|x-4|_{p})\Vert_{2}$. The
numerical simulations in this section confirm that when $W\neq0$ or $Z\neq0$, the
norm $\left\Vert \Psi\left(  \cdot,t\right)  \right\Vert _{2}^{2}=1$ is no
longer conserved, and consequently, the network models an open quantum system
subject to environmental interaction. The way in which $W\neq0$, or $Z\neq0$, controls
the behavior of the norm $\left\Vert \Psi\left(  \cdot,t\right)  \right\Vert
_{2}^{2}$ is not completely understood at this moment. The numerical
simulations show that the equation (\ref{EQ_0A}) describes a quantum system,
where decoherence is controlled by the parameters $W$ and $Z$.

\begin{figure}
    \centering
    \includegraphics[width=0.75\linewidth]{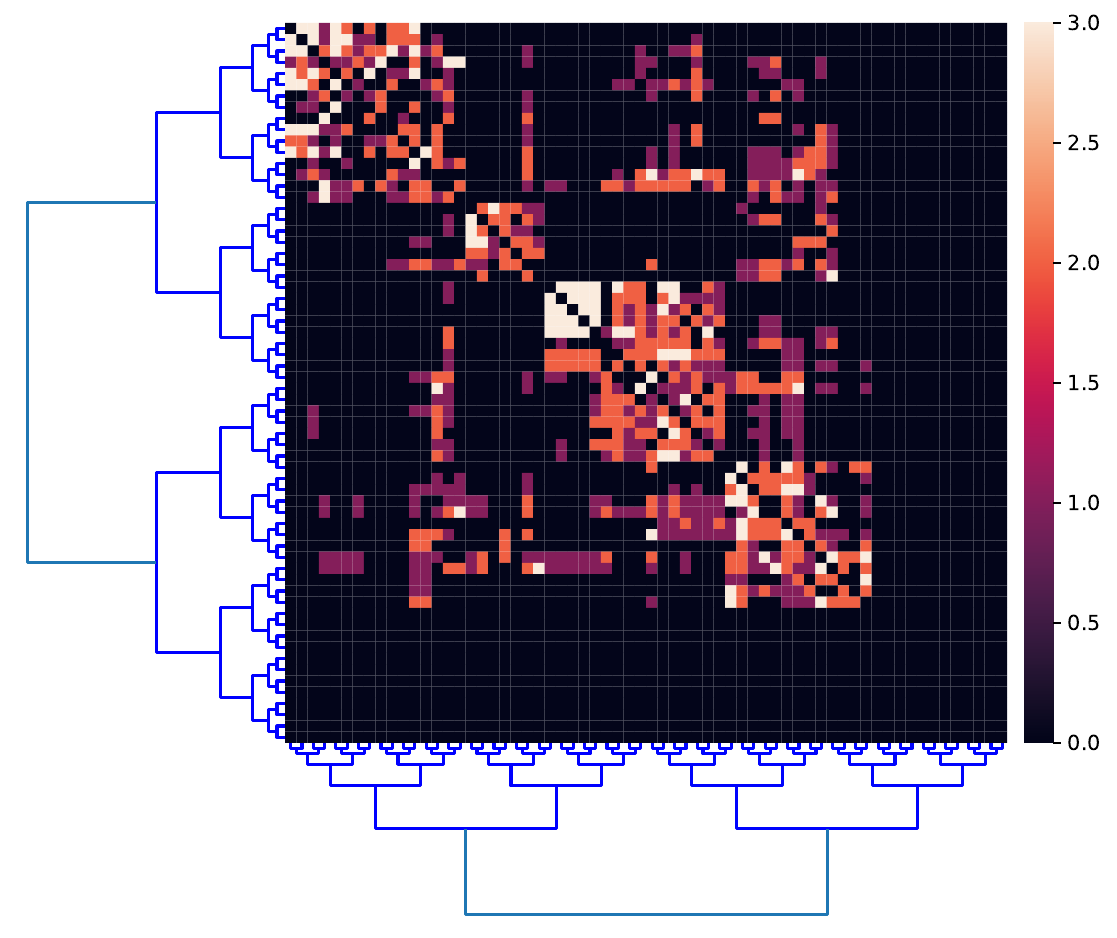}
    \caption{$p$-Adic approximation for the cat cortex. The  parameters are $ p=2$ and $l=6$, \cite{Zuniga-Entropy}. The connectivity matrix for the cat cortex is  available at  https://terpconnect.umd.edu/~cherniak/SuppData/cat.htm}
    \label{figure cat_matrix}
\end{figure}


\medskip
\textbf{Group A: Unitary evolution (baseline).} The following simulation establishes the baseline, closed-system behavior of the network ($W=Z=0$), against which all subsequent (open-system) groups are compared.

\subsection{Numerical simulation 1}
(Validates the unitary, norm-conserving special case of equation \mbox{\eqref{EQ_0A}} discussed in Section \mbox{\ref{Section_3}}.)
The parameters used are  $p=3$, $l=6$, $G_{6}=\mathbb{Z}_{3}/3^{6}\mathbb{Z}_{3}$, $W=0$, $Z=0$, and $\alpha=2.5$. The initial condition is $\Psi_0(x) = \Omega(p^2\vert x-4\vert_p)/\Vert\Omega(p^2\vert x -4\vert_p)\Vert_2$. The time runs over $[0,400]$, with step $\delta t = 0.001$. This case corresponds to a free Schr\"odinger equation, so the network undergoes a unitary evolution. Figure \ref{figure1}-(A) shows a numerical approximation for the probability density $\left\vert \Psi \left( x,t\right) \right\vert ^{2}$, and Figure \ref{figure1}-(B) shows $\left\Vert\Psi \left( \cdot,t\right) \right\Vert _{2}^{2}=\int_{\mathbb{Z}_{p}}\left\vert
\Psi \left( x,t\right) \right\vert ^{2}dx$. In this case, the  Schr\"odinger equation describes a continuous-time quantum Markov chain; see \cite{Zuniga-QM-2, Zuniga-Mayes, Zuniga-Chacon}. Figure \mbox{\ref{figure1}}-(B) is, by construction, flat at 1: this is the quantitative conservation diagnostic used throughout this section, and every departure from a flat line in later groups is the signature of non-unitary (open-system) behavior.

\begin{figure}[H]
    \centering

    \begin{subfigure}[t]{0.48\linewidth}
        \centering
        \includegraphics[width=\linewidth]{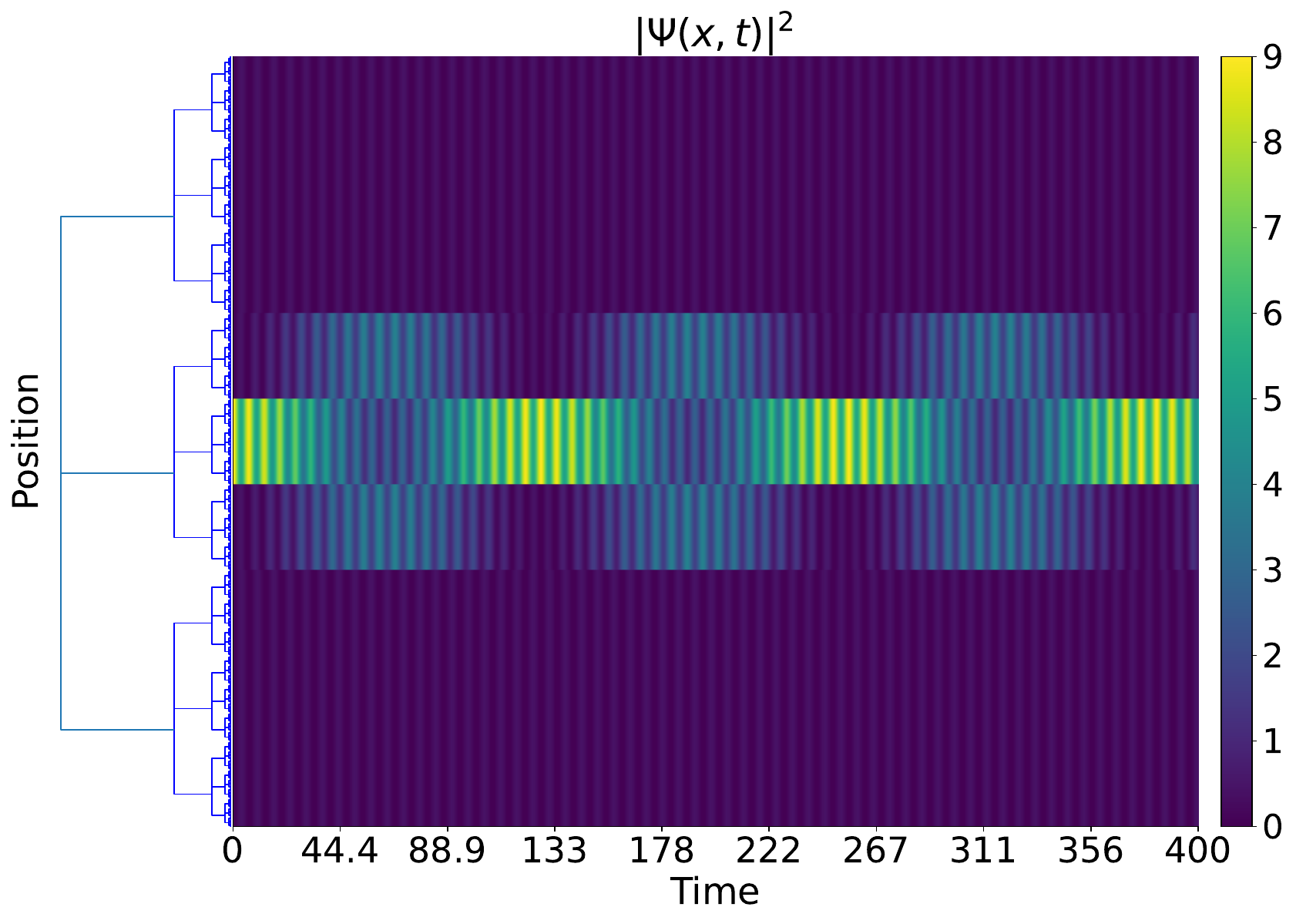}
        \caption{A numerical approximation for  $\left\vert \Psi \left( x,t\right) \right\vert ^{2}$. The position variable runs over the tree $G_{6}$, and the time runs over $[0,400]$.}
        \label{fig:sim1_l1_norm_sq}
    \end{subfigure}
    \hfill
    \begin{subfigure}[t]{0.48\linewidth}
        \centering
        \includegraphics[width=\linewidth]{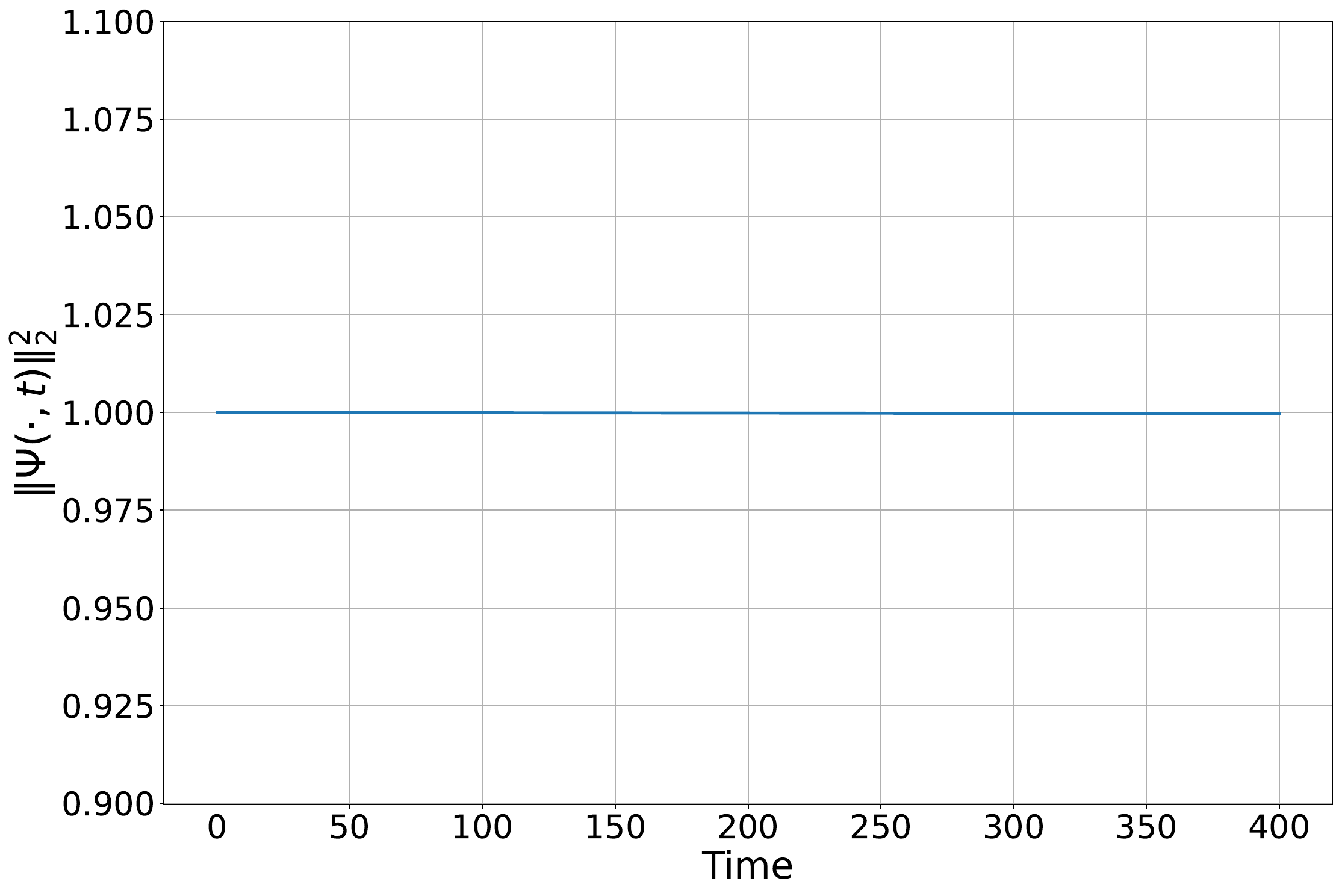}
        \caption{A numerical approximation for $\left\Vert\Psi \left( \cdot,t\right) \right\Vert _{2}^{2}$ for $x\in G_{6}$, and  $t\in[0,400]$. This calculation confirms that the evolution is unitary.}
        \label{fig:sim1_norm_trajectories}
    \end{subfigure}
    \caption{Results of the numerical simulation 1}
    \label{figure1}
\end{figure}

\medskip
\textbf{Group B: Non-unitary evolution driven by an external input, without neuron interaction ($W=0$, $Z\neq0$).} This group isolates the effect of the bias/input term $Z$ on the norm, with the interaction weight kept at $W=0$.

\subsection{Numerical simulation 2}
(Validates the non-unitary regime of equation \mbox{\eqref{EQ_0A}} driven purely by an external input, i.e. $W=0$, $Z\neq0$.)
The parameters are $p=3$, $l=6$, $\alpha=2.5$, $W=0$,  $Z(x,t)=\sin(0.01 \pi t) 1_{[25,50]}(t) + \sin(0.01 \pi t) 1_{[200,225]}(t)$. The initial condition is  $\Psi_0(x) = \Omega(p^2\vert x -4\vert_p)/\Vert\Omega(p^2\vert x -4\vert_p)\Vert_2$. The time runs over $[0,600]$ with a step  $\delta t = 0.001$. In this case, the network has two inputs, $\Psi_0(x)$ and $Z(t)$. The input $Z(t)$ corresponds to two pulses, see  Figure \ref{figure2A}. Figure \ref{figure2}-(B) shows that the evolution is non-unitary, so the network behaves as an open quantum system. Figures \ref{figure2}-(A) and (B) show that the network responds to each pulse. The condition $W=0$ means that there is no interaction between the neurons.

\begin{figure}[H]
    \centering

    \begin{subfigure}[t]{0.49\linewidth}
        \centering
        \includegraphics[width=\linewidth]{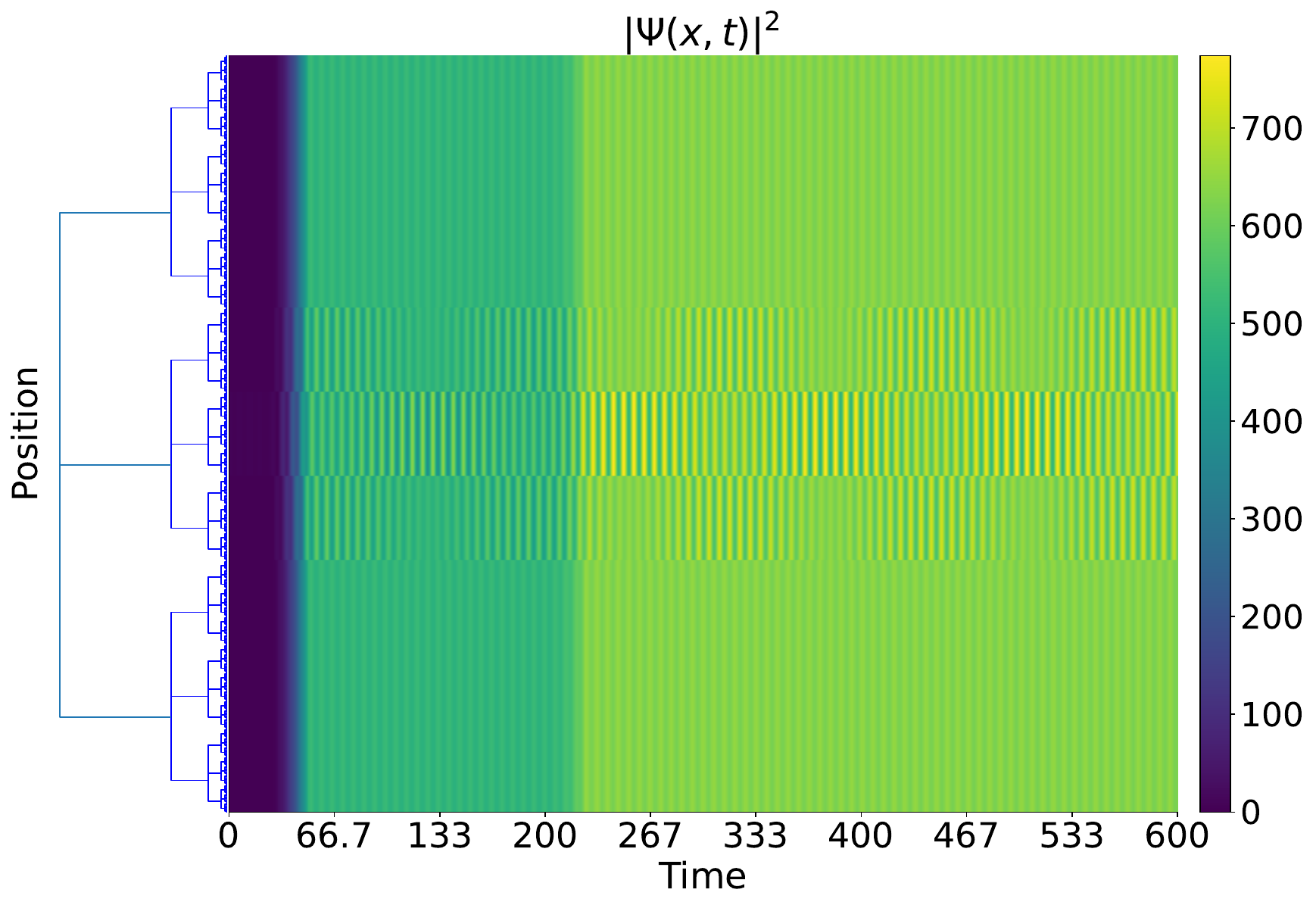}
        \caption{An approximation for $\left\vert \Psi \left( x,t\right) \right\vert ^{2}$. The position runs over $G_{6}=\mathbb{Z}_{3}/3^{6}\mathbb{Z}_{3}$. The network responds to each pulse.}
        \label{fig:simulation_3_1_l1_norm_sq}
    \end{subfigure}
    \hfill
    \begin{subfigure}[t]{0.48\linewidth}
        \centering
        \includegraphics[width=\linewidth]{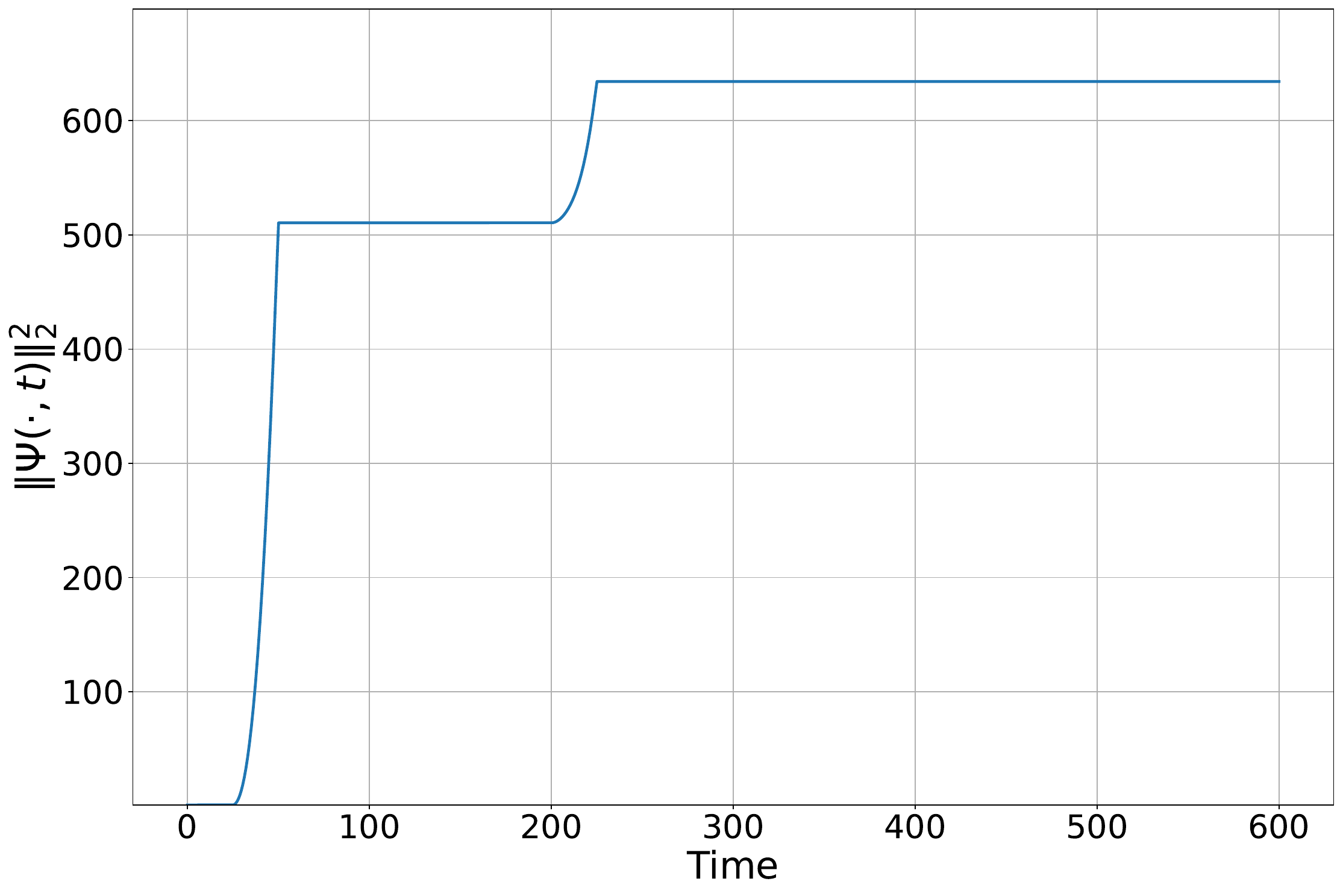}
        \caption{A numerical approximation for $\left\Vert\Psi \left( \cdot,t\right) \right\Vert _{2}^{2}$ for $x\in G_{6}$. The network behaves as an open system.}
        \label{fig:sim3_1_norm_trajectories}
    \end{subfigure}
    
    

    \caption{Numerical simulation 2}
    \label{figure2}
\end{figure}
\begin{figure}[htbp]
    \centering
    \includegraphics[width=0.5\linewidth]{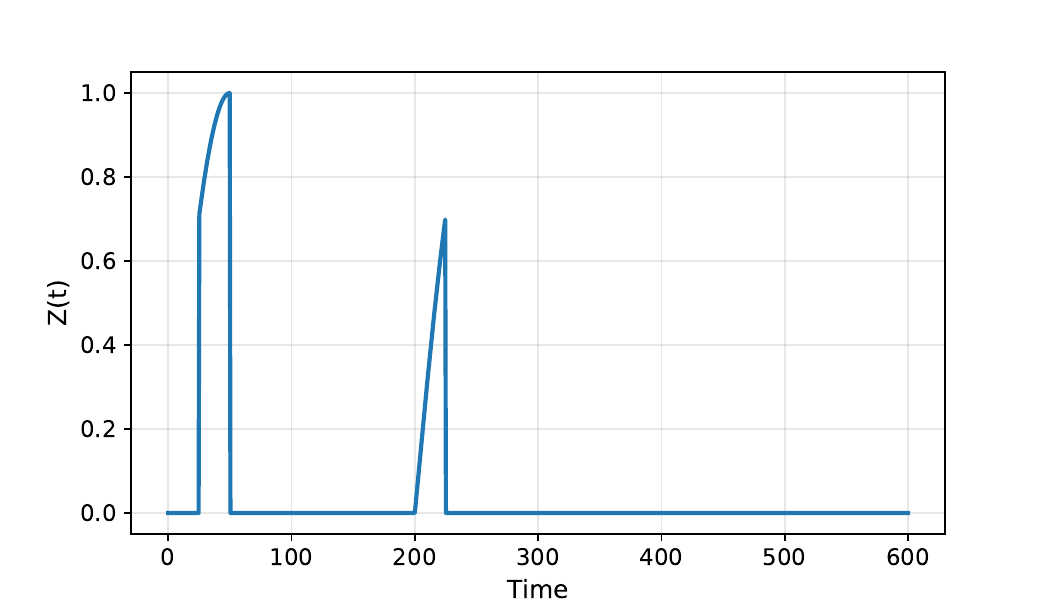}
    \caption{Pulse used in Simulation 2}
    \label{figure2A}
\end{figure}

\subsection{\label{Num_Sim_3}Numerical simulation 3}
(Validates the non-unitary regime driven by a constant external input, $W=0$, $Z\neq0$; continues Group B.)
The parameters are $p=3$, $l=6$, $\alpha=2.5$, $W=0$, $Z=10$. The  initial condition is $\Psi_0(x) = \Omega(p^2\vert x -4\vert_p)/\Vert\Omega(p^2\vert x -4\vert_p)\Vert_2$. The time runs over $[0,100]$ with step $\delta t = 0.001$. Every neuron gets a constant input $Z=10$, and there are no interactions between the neurons. This causes the norm  $\left\Vert\Psi \left( \cdot,t\right) \right\Vert _{2}^{2}$  to grow in time, see Figures \ref{Figure3}-(A) and (B). This calculation confirms that the evolution is not unitary.

\begin{figure}[H]
    \centering
    \begin{subfigure}[t]{0.50\linewidth}
        \centering
        \includegraphics[width=\linewidth]{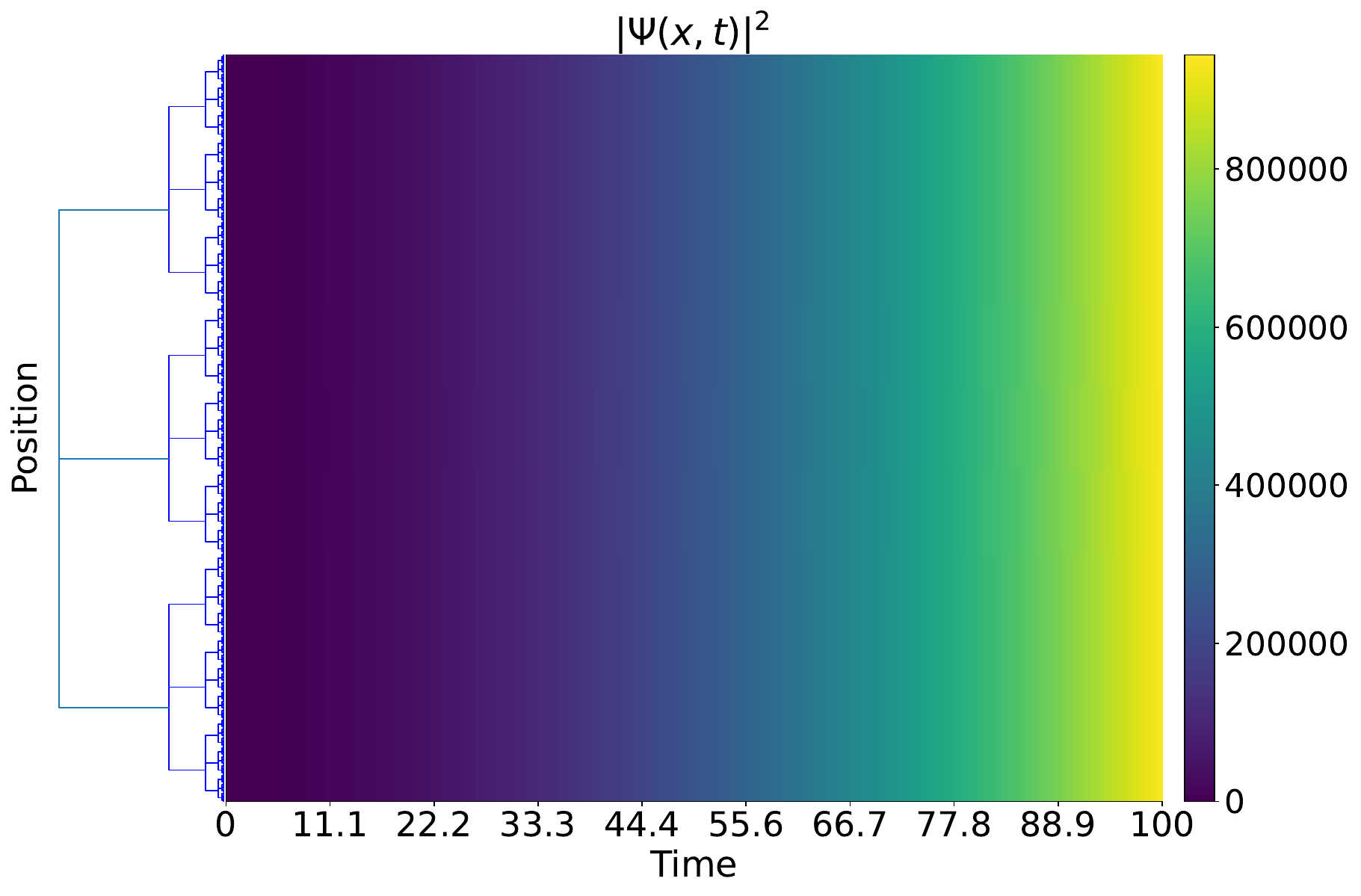}
        \caption{A numerical approximation for $\left\vert \Psi \left( x,t\right) \right\vert ^{2}$ for $x\in G_{6}$, and  $t\in[0,100]$.}
        \label{fig:sim21_l1_norm_sq}
    \end{subfigure}
    \hfill
    \begin{subfigure}[t]{0.48\linewidth}
        \centering
        \includegraphics[width=\linewidth]{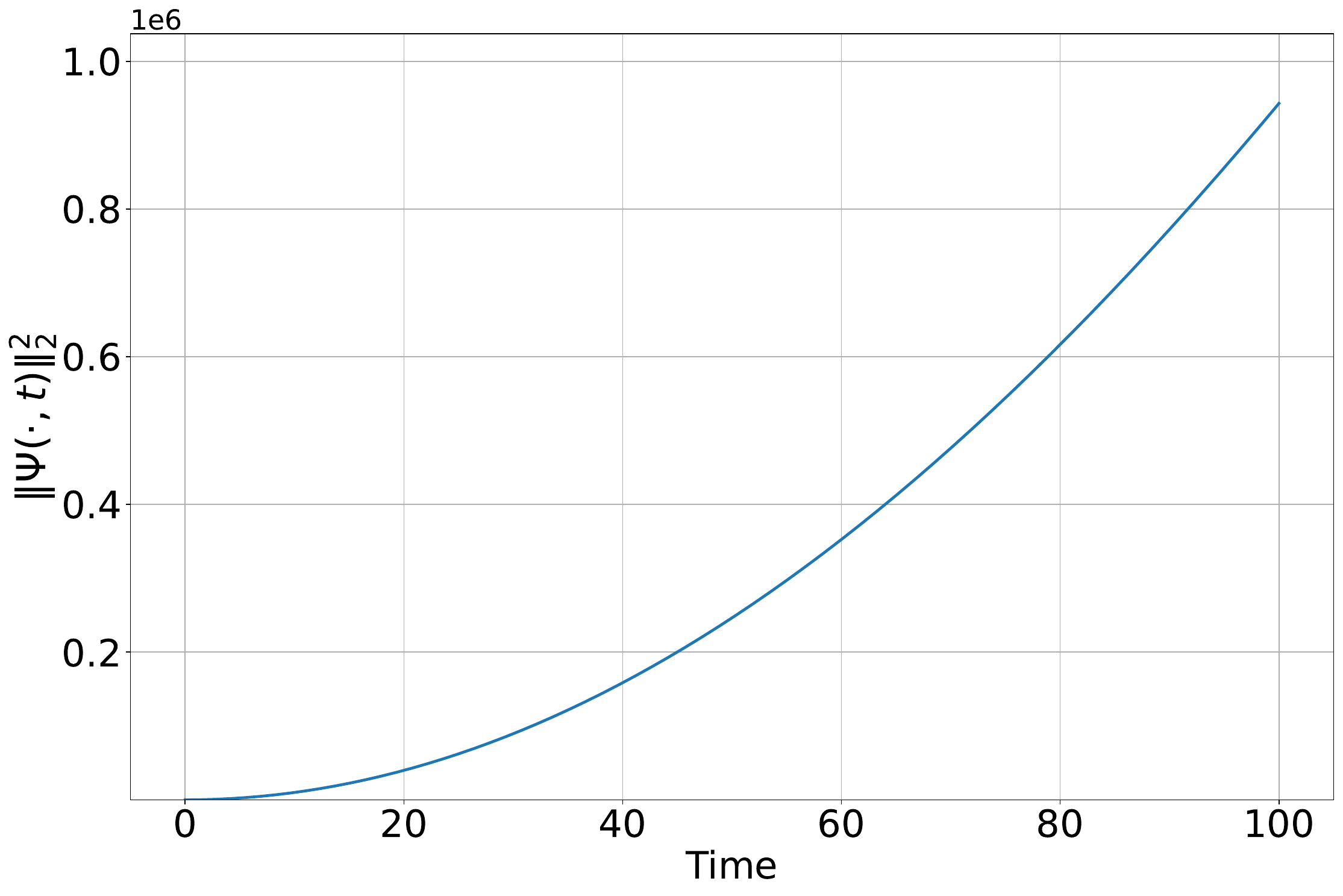}
        \caption{A numerical approximation for $\left\Vert\Psi \left( \cdot,t\right) \right\Vert _{2}^{2}$. The network behaves as an open quantum system.}
        \label{fig:sim21_norm_trajectories}
    \end{subfigure}
    \caption{Results for Simulation 3.1}
    \label{Figure3}
\end{figure}

\medskip
\textbf{Group C: Non-unitary evolution driven by neuron interaction ($W\neq0$).} The remaining simulations in this section switch on the interaction weight $W$ and study how it reshapes the network's response to an input, including the habituation phenomenon.

We now investigate the network's response with $W\neq 0$. We use the approximation of the cat connection matrix $W_{cat}$ developed in \cite{Zuniga-Entropy}. The parameters are $p=2$, $l=6$, $\alpha=2.5$, $W=0.1W_{cat}$, see  Figure \ref{figure cat_matrix}, $Z=10$, and the time runs over  $[0,600]$ with a step $\delta t = 0.001$. The initial condition is   $\Psi_0(x) = \Omega(p^2\vert x -4\vert_p)/\Vert\Omega(p^2\vert x -4\vert_p)\Vert_2$.  The use of prime $p=2$ is required due to the fact that the approximation for the cat connection matrix uses this prime. The interaction between the neurons changes the network's response to the constant input $Z=10$; see Figures \ref{Figure4}-(A) and (B). We interpret Figure \ref{Figure4}-(B) as showing that the network habituates to the constant signal. Habituation is a fundamental form of non-associative learning in which an organism reduces its behavioral response or attention to a repeated, harmless stimulus, \cite{Habituation}. The habituation appears when $W\neq 0$.

\begin{figure}[H]
    \centering

    \begin{subfigure}[t]{0.50\linewidth}
        \centering
        \includegraphics[width=\linewidth]{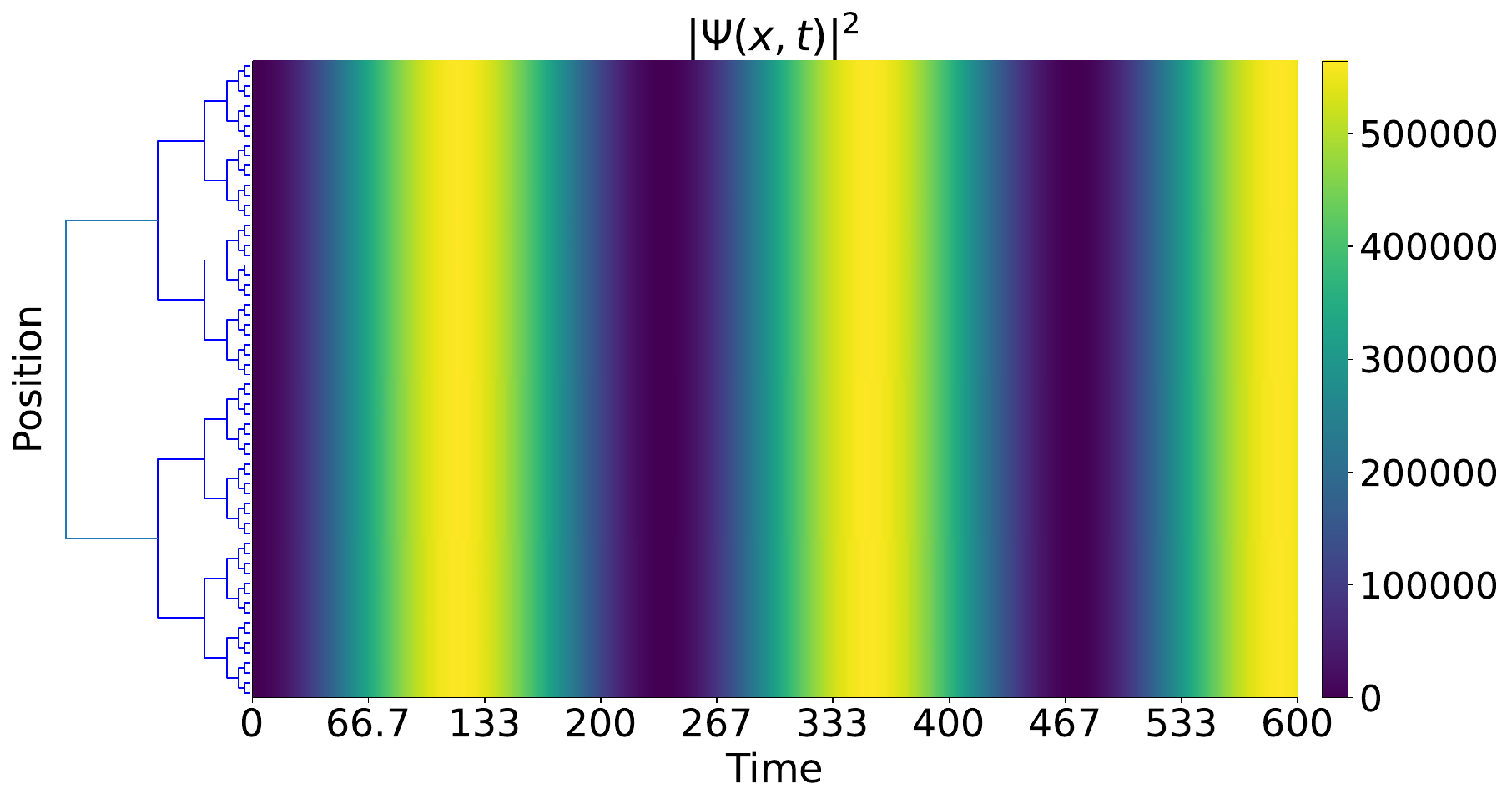}
        \caption{A numerical approximation for $\left\vert \Psi \left( x,t\right) \right\vert ^{2}$ for $x\in G_{6}=\mathbb{Z}_{2}/2^{6}\mathbb{Z}_{2}$, and  $t\in[0,600]$.}
        \label{fig:sim32_l1_norm_sq}
    \end{subfigure}
    \hfill
    \begin{subfigure}[t]{0.48\linewidth}
        \centering
        \includegraphics[width=\linewidth]{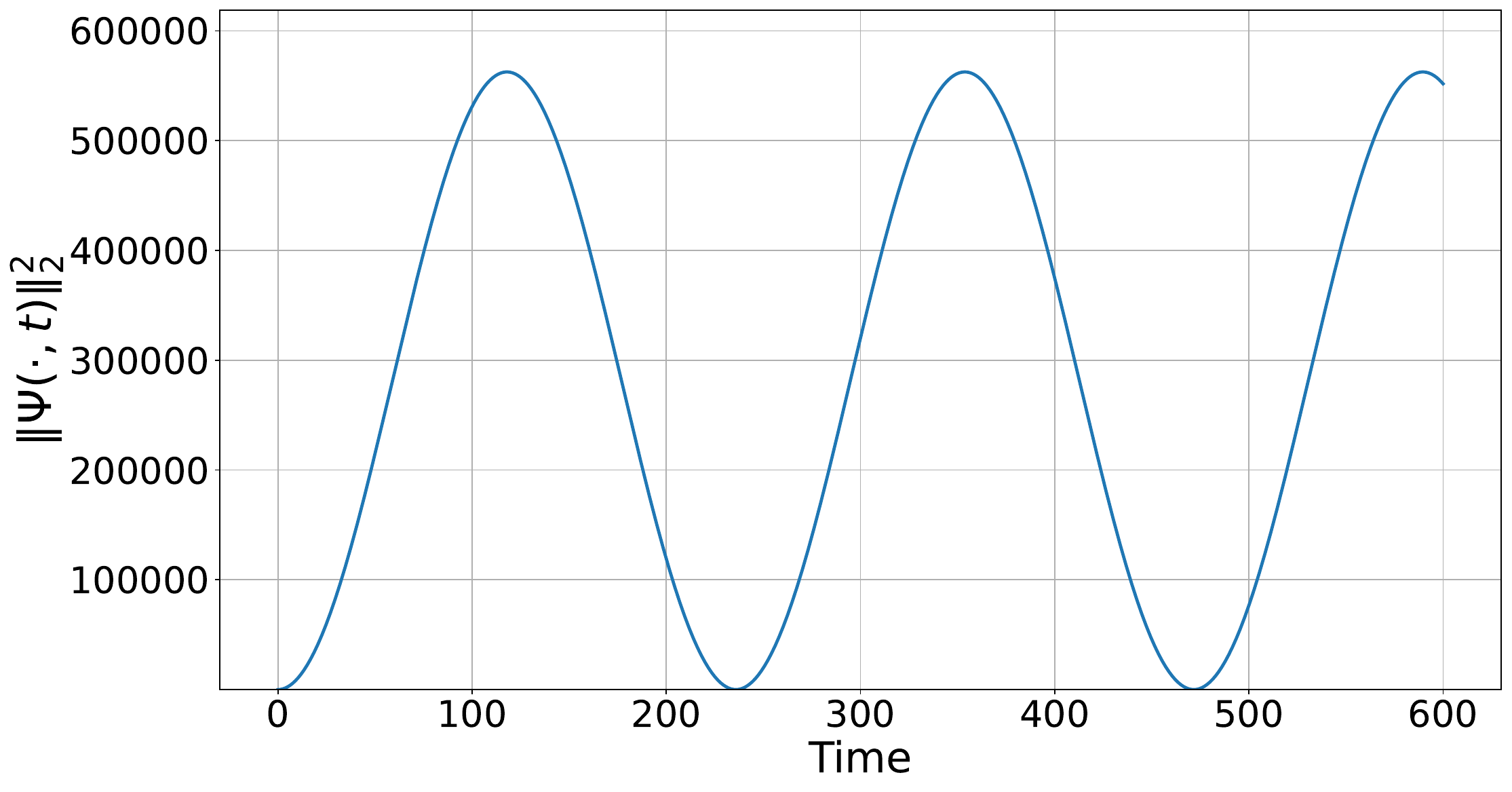}
        \caption{A numerical approximation for $\left\Vert\Psi \left( \cdot,t\right) \right\Vert _{2}^{2}$.}
        \label{fig:sim32_norm_trajectories}
    \end{subfigure}

    \caption{Results for Simulation 3.2}
    \label{Figure4}
\end{figure}

We now study the case where $W(x,y)=W_{0}$ is a constant function, which means that the interaction between neurons is controlled by the term%
\begin{equation*}
W_{0}\int\limits_{\mathbb{Z}_{p}}\phi \left( \Psi \left( y,t\right) \right)
dy.
\end{equation*}
The parameters are $p=3$, $l=6$, $\alpha=2.5$, $W=50$, $Z=0$, and the time runs over $[0,400]$ with step $\delta t = 0.001$. The initial condition is $\Psi_0(x) =0.5 + 0.3i$.  In this case, the initial datum serves as the input.
The results in Figures \ref{Figure5}-(A) and (B) indicate habituation in the network.
\begin{figure}[H]
    \centering

    \begin{subfigure}[t]{0.48\linewidth}
        \centering
        \includegraphics[width=\linewidth]{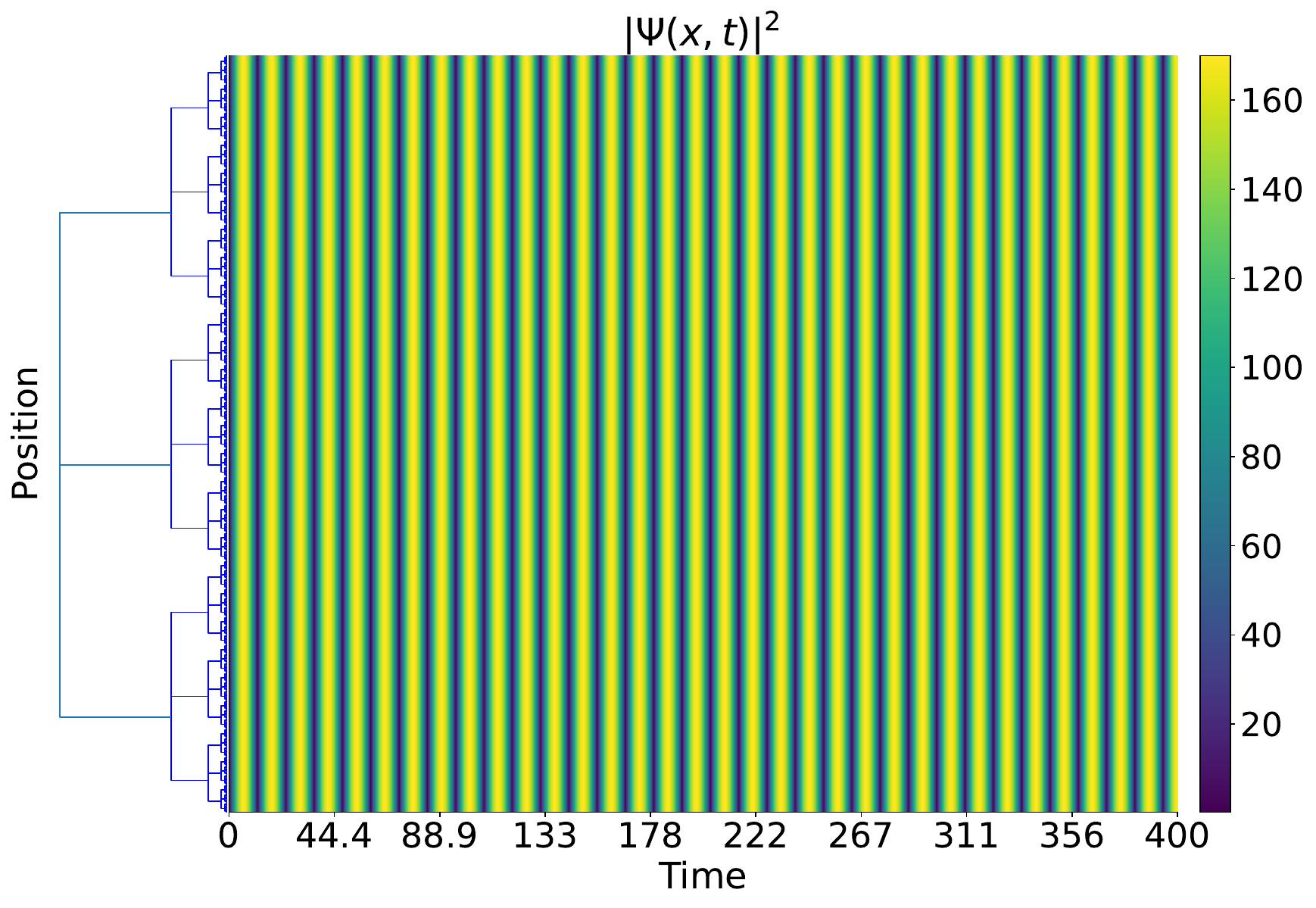}
        \caption{A numerical approximation for  $\left\vert \Psi \left( x ,t\right) \right\vert ^{2}$.}
        \label{fig:sim17_1_l1_norm_sq}
    \end{subfigure}
    \hfill
    \begin{subfigure}[t]{0.48\linewidth}
        \centering
        \includegraphics[width=\linewidth]{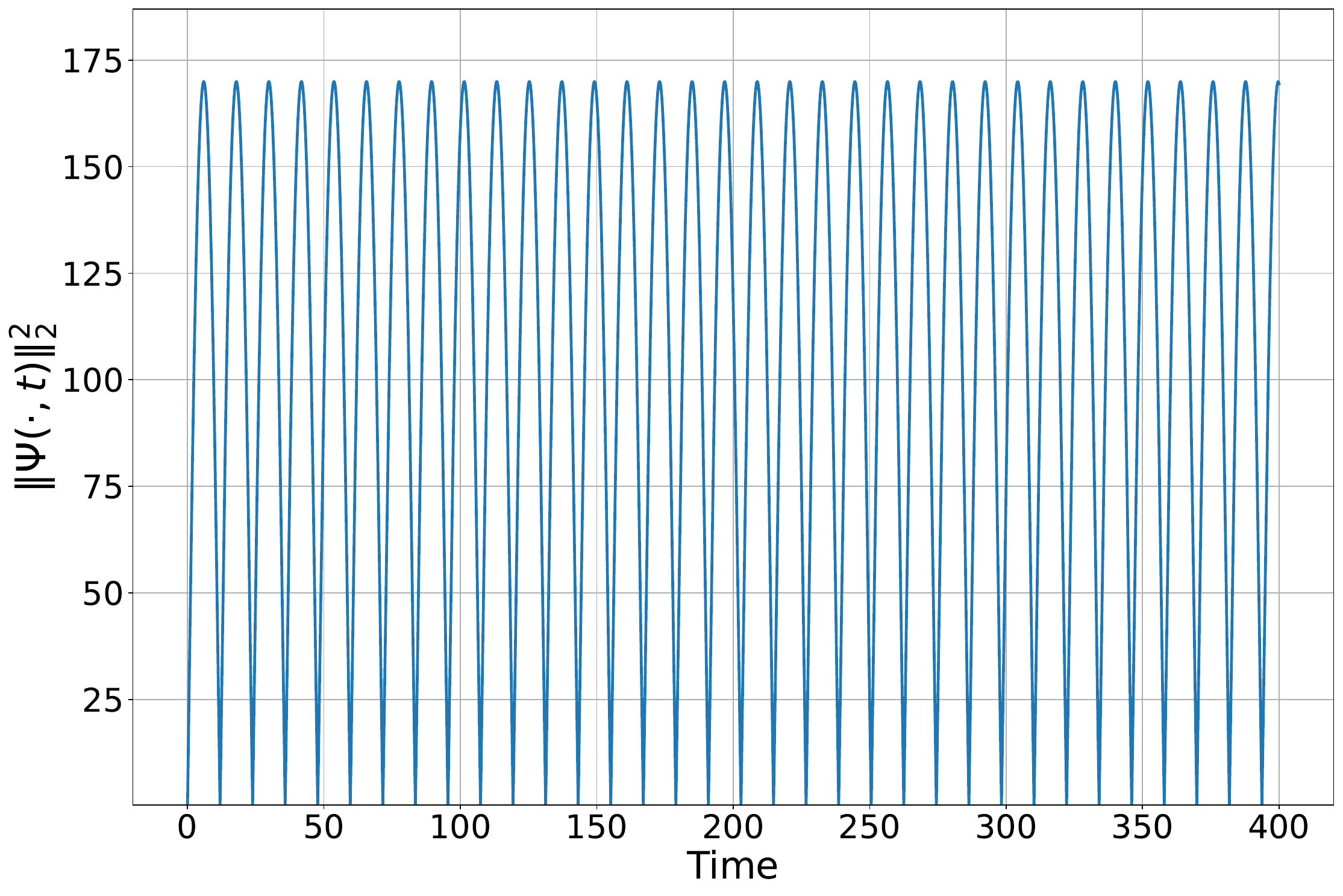}
        \caption{A numerical approximation for $\left\Vert\Psi \left( \cdot,t\right) \right\Vert _{2}^{2}$ }
        \label{fig:sim17_1_norm_trajectories}
    \end{subfigure}

    \caption{Results for Simulation 3.3}
    \label{Figure5}
\end{figure}

\subsection{\label{Num_Sim_4}Numerical simulations 4}
(Validates the non-unitary, interaction-driven regime of equation \mbox{\eqref{EQ_0A}}, using the cat-cortex approximation of $W$; part of Group C.)
The parameters are  $p=2$, $l=6$,  $W=0.1W_{cat}$ (see  Figure \ref{figure cat_matrix}), $Z(x,t)=\sin(0.1 \pi t) \Omega(p^2\vert x - 3\vert_p)1_{[25,50]}(t) + \sin(0.1 \pi t) 1_{[200,225]}(t)\Omega(p^2\vert x \vert_p)$, and $\alpha=1.6$. The initial condition is $\Psi_0(x) =0 + 0i$, and the time runs over  $[0,600]$ with a step $\delta t = 0.001$. In this case, the input $Z(x,t)$ consists of two pulses in time and space; see Figure \ref{Figure6}-(A). The numerical simulations indicate that the network responds to the pulses; see Figures \ref{Figure6}-(A) and (B).

\begin{figure}[H]
    \centering
    \begin{subfigure}[t]{0.50\linewidth}
        \centering
        \includegraphics[width=\linewidth]{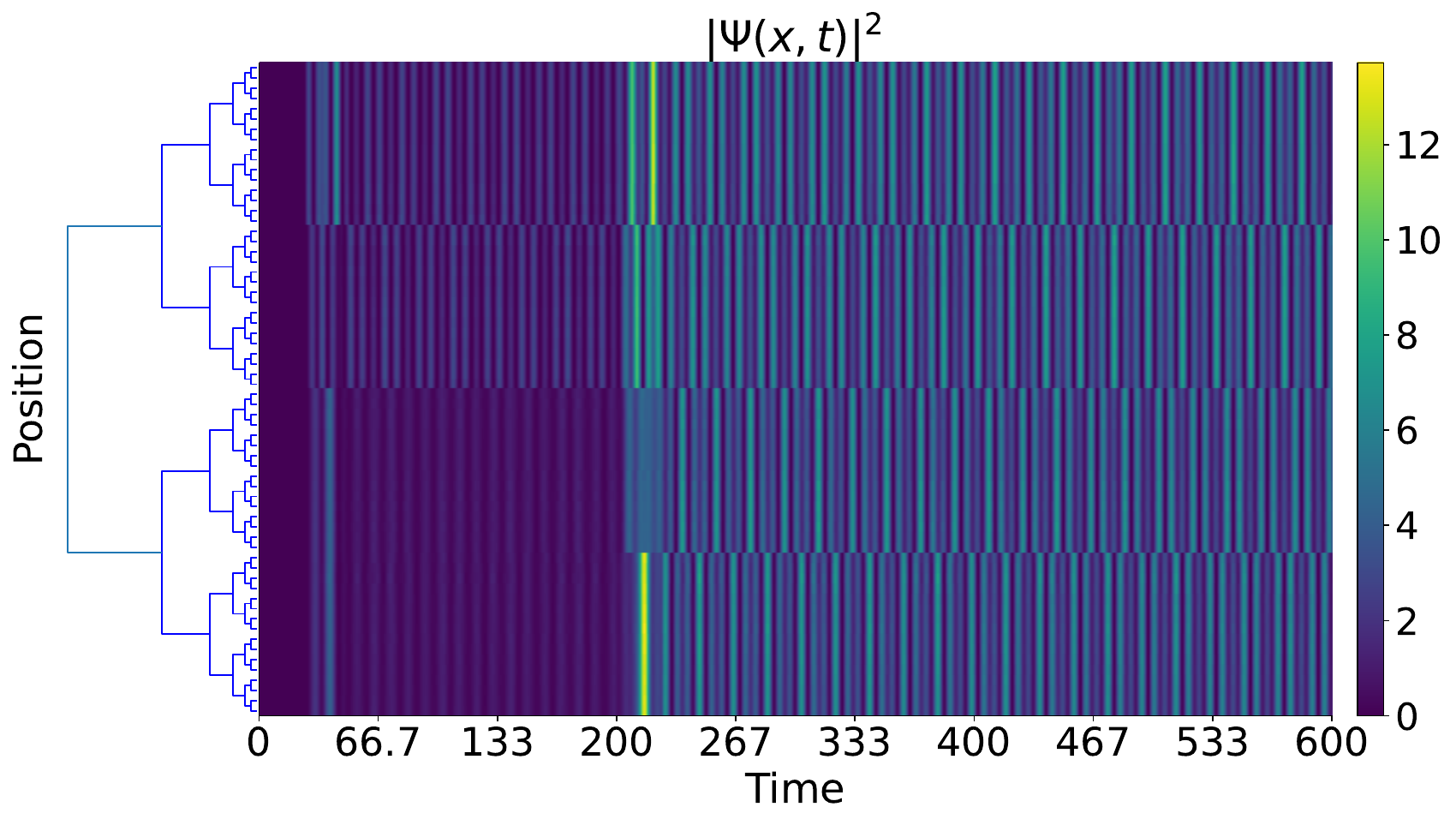}
        \caption{A numerical approximation for $\left\vert \Psi \left( x ,t\right) \right\vert ^{2}$.}
        \label{fig:sim3_12_l1_norm_sq}
    \end{subfigure}
    \hfill
    \begin{subfigure}[t]{0.48\linewidth}
        \centering
        \includegraphics[width=\linewidth]{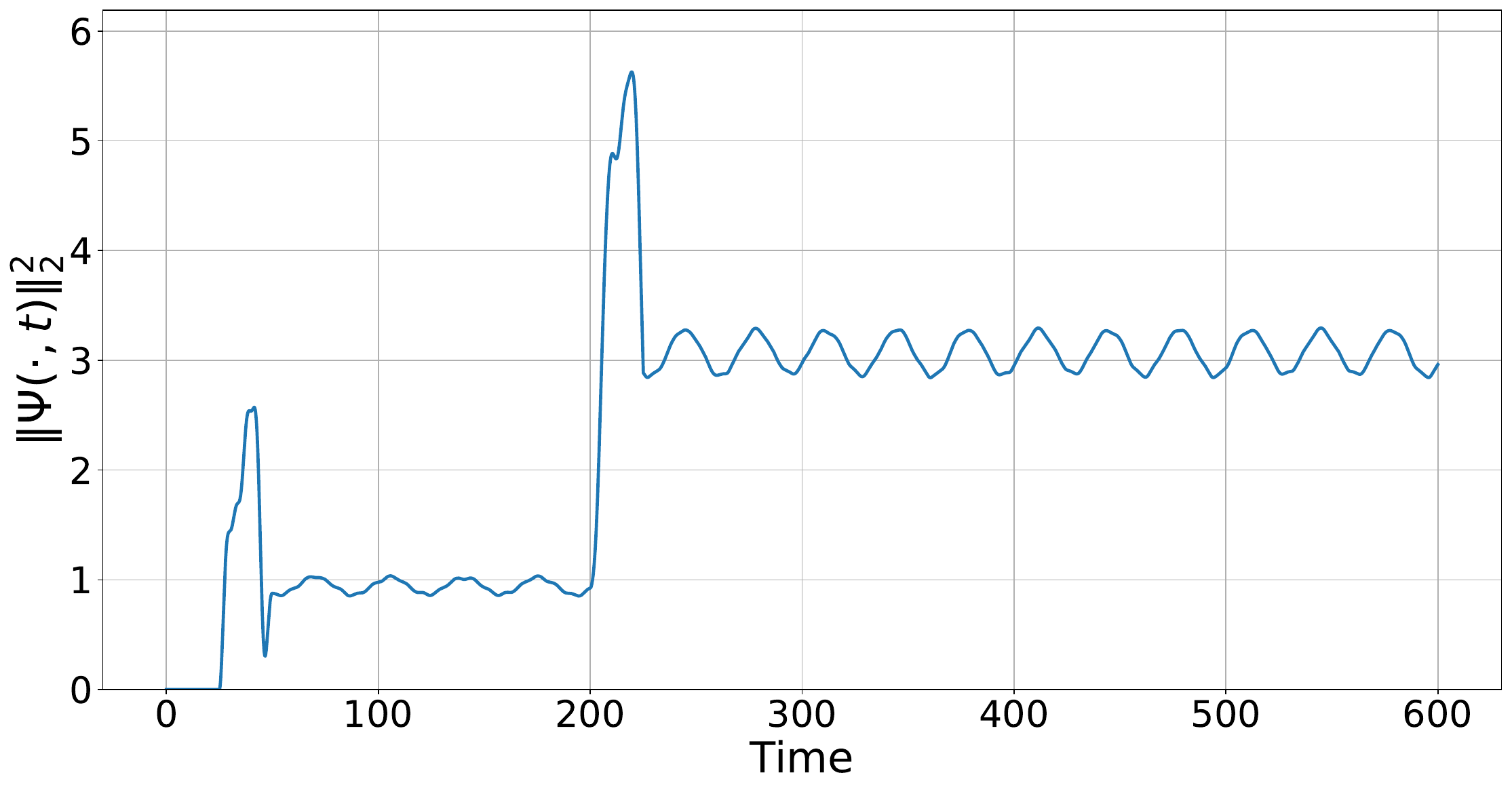}
        \caption{A numerical approximation for $\left\Vert\Psi \left( \cdot,t\right) \right\Vert _{2}^{2}$. }
        \label{fig:simulation_7_norm_trajectories}
    \end{subfigure}

    \caption{Results for Simulation 4.1}
    \label{Figure6}
\end{figure}

In the next simulation, we use the same values for the parameters $p, l, W, \alpha$. The time runs over $[0,1500]$ with step $\delta t = 0.001$. But now there are two inputs, 

\[
Z(x,t)=\sin(0.1 \pi t) 1_{[25,50]}(t) + \sin( \pi t) 1_{[200,225]}(t)+ \sin(10 \pi t) 1_{[800,1225]}(t)  + (0.1)  1_{[1225, \infty)}(t),
\]
and the initial datum 
\[
\Psi_0(x) = \Omega(p^2\vert x -4\vert_p)/\Vert\Omega(p^2\vert x -4\vert_p)\Vert_2.
\]
Notice that for $t\geq 1225$, $Z(x,t)=0.1$. This allows us to test the network's habituation. Figures \ref{Figure7}-(A) and (B) show the network's responses to each pulse, and indicate habituation to the constant input  $Z(x,t)=0.1$, for  $t\geq 1225$. Figure \ref{figure7A} shows the pulse used in the simulation.
\begin{figure}[H]
    \centering

    \begin{subfigure}[t]{0.49\linewidth}
        \centering
        \includegraphics[width=\linewidth]{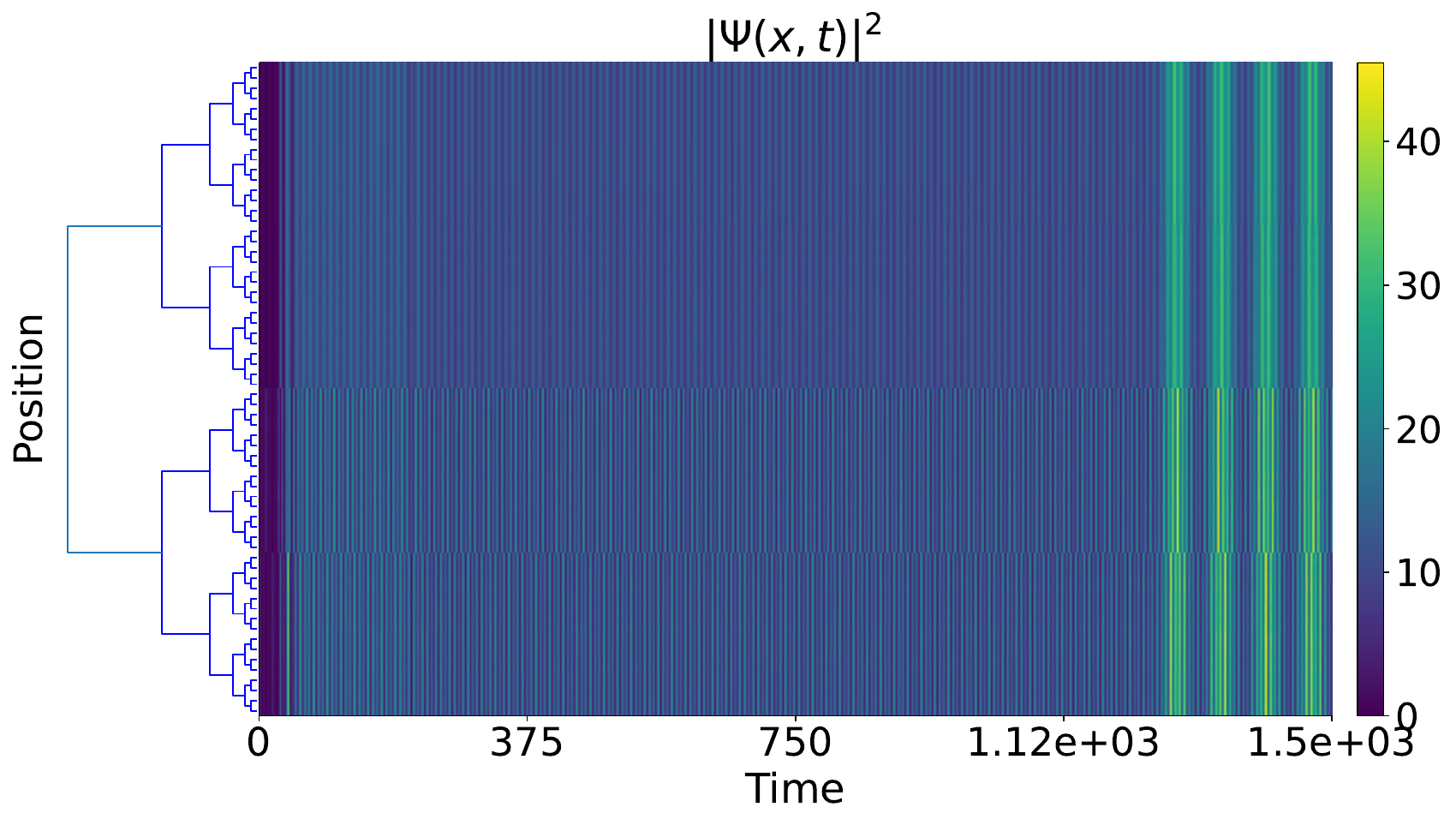}
        \caption{A numerical approximation for $\left\vert \Psi \left( x ,t\right) \right\vert ^{2}$.}
        \label{fig:simulation_8_A_l1_norm_sq}
    \end{subfigure}
    \hfill
    \begin{subfigure}[t]{0.48\linewidth}
        \centering
        \includegraphics[width=\linewidth]{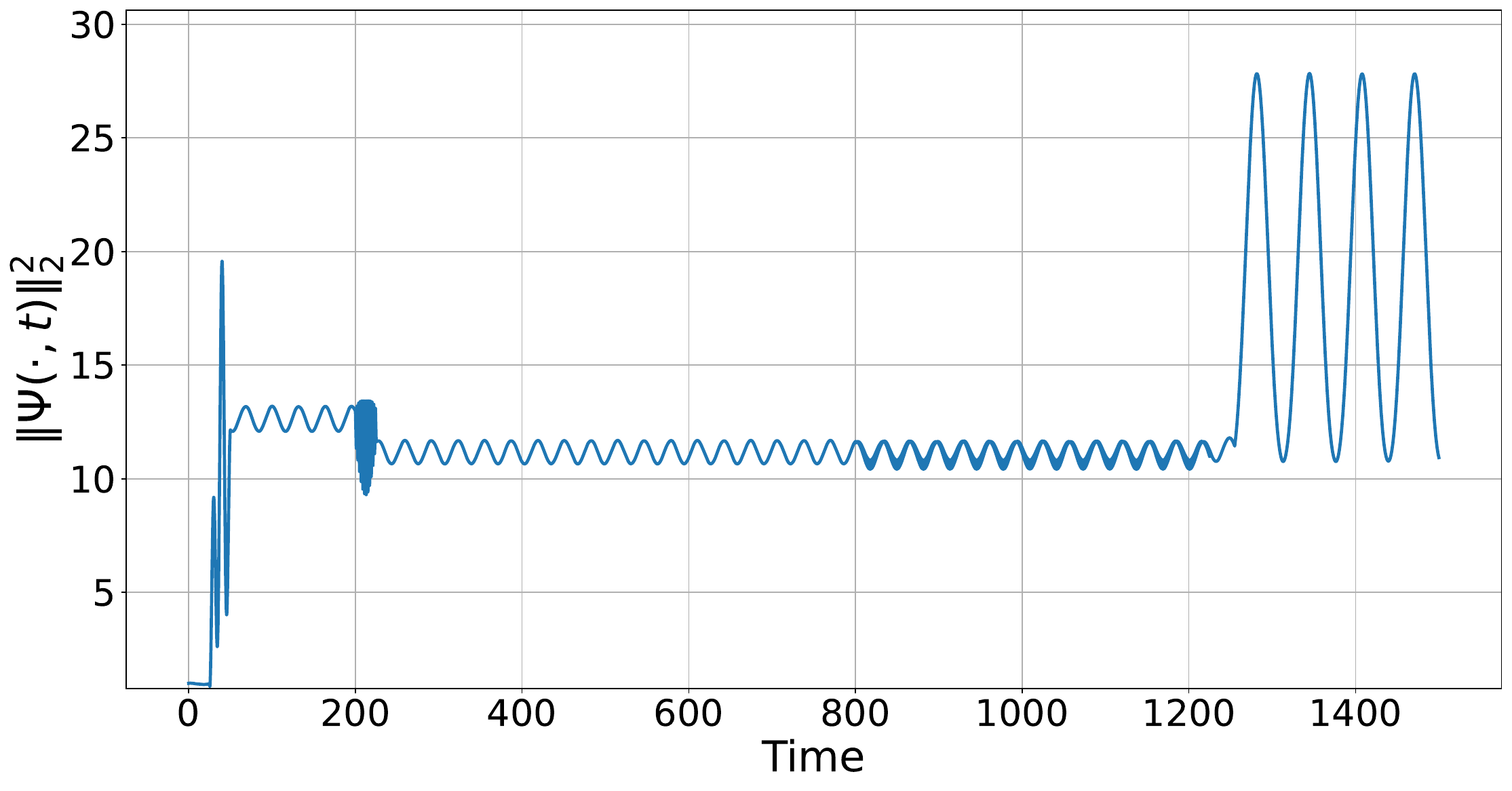}
        \caption{A numerical approximation for $\left\Vert\Psi \left( \cdot,t\right) \right\Vert _{2}^{2}$.}
        \label{fig:simulation_8_A_norm_trajectories}
    \end{subfigure}
    
    

    \caption{Results for Simulation 4.2}
    \label{Figure7}
\end{figure}
\begin{figure}[htbp]
    \centering
    \includegraphics[width=0.75\linewidth]{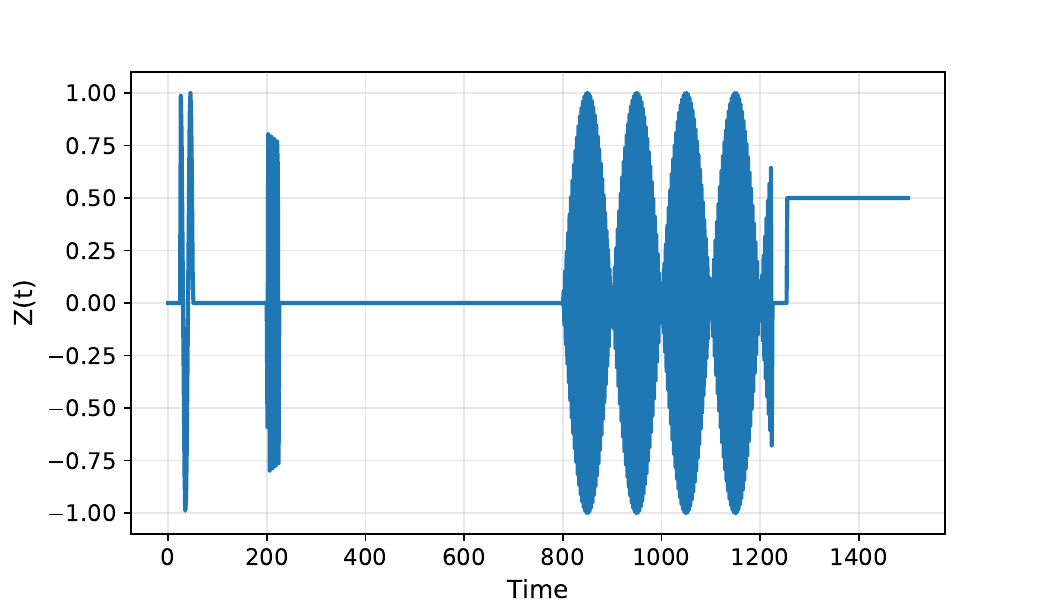}
    \caption{Pulse used in Simulation 4.2}
    \label{figure7A}
\end{figure}

In the next simulation, we keep all the parameters but set $W=1.0$, indicating that the neurons interact. In this case, the network responds poorly to the pulses; see Figures \ref{Figure8}-(A) and (B).

\begin{figure}[H]
    \centering

    \begin{subfigure}[t]{0.49\linewidth}
        \centering
        \includegraphics[width=\linewidth]{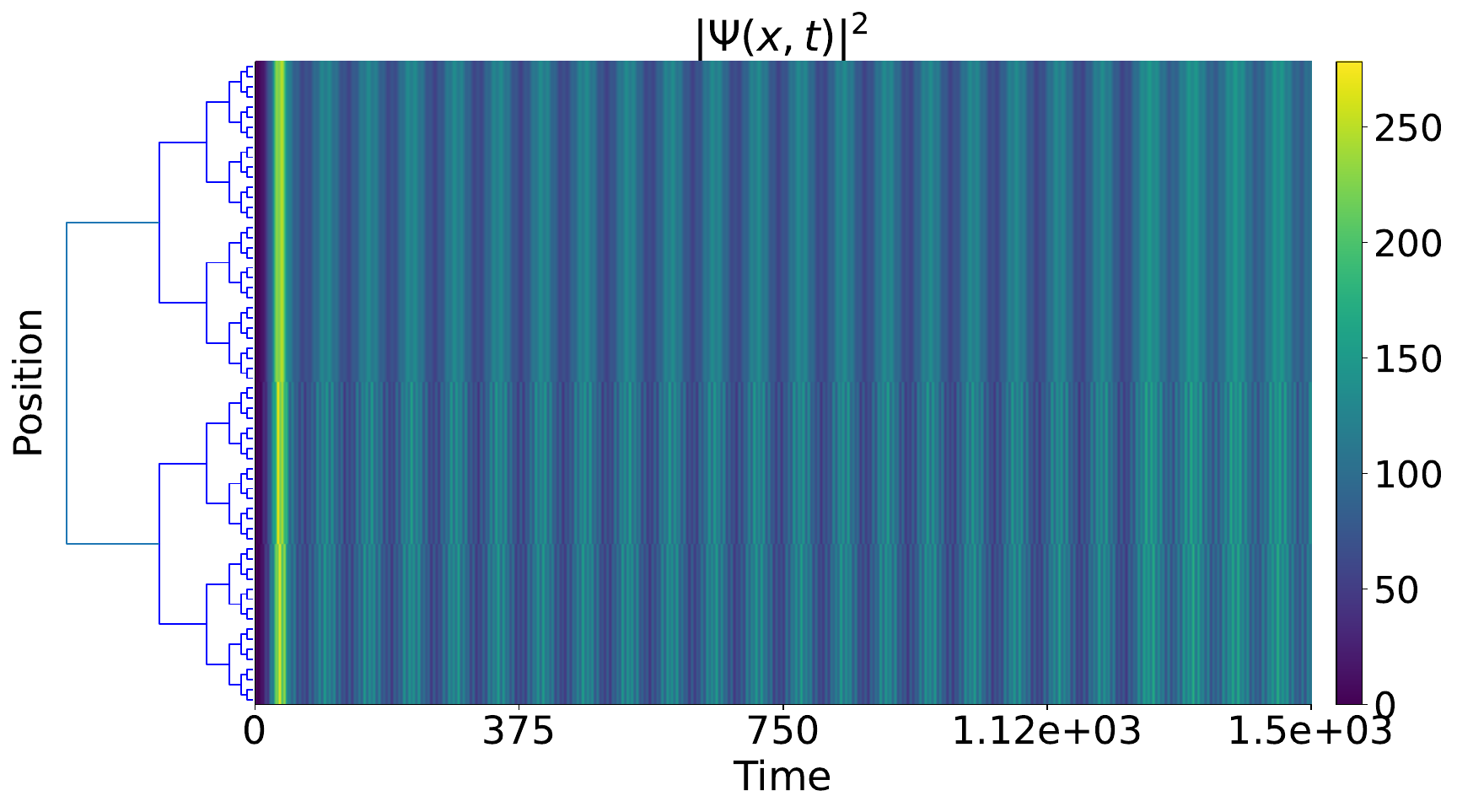}
        \caption{A numerical approximation for $\left\vert \Psi \left( x ,t\right) \right\vert ^{2}$.}
        \label{fig:simulation_9_B_l1_norm_sq}
    \end{subfigure}
    \hfill
    \begin{subfigure}[t]{0.48\linewidth}
        \centering
        \includegraphics[width=\linewidth]{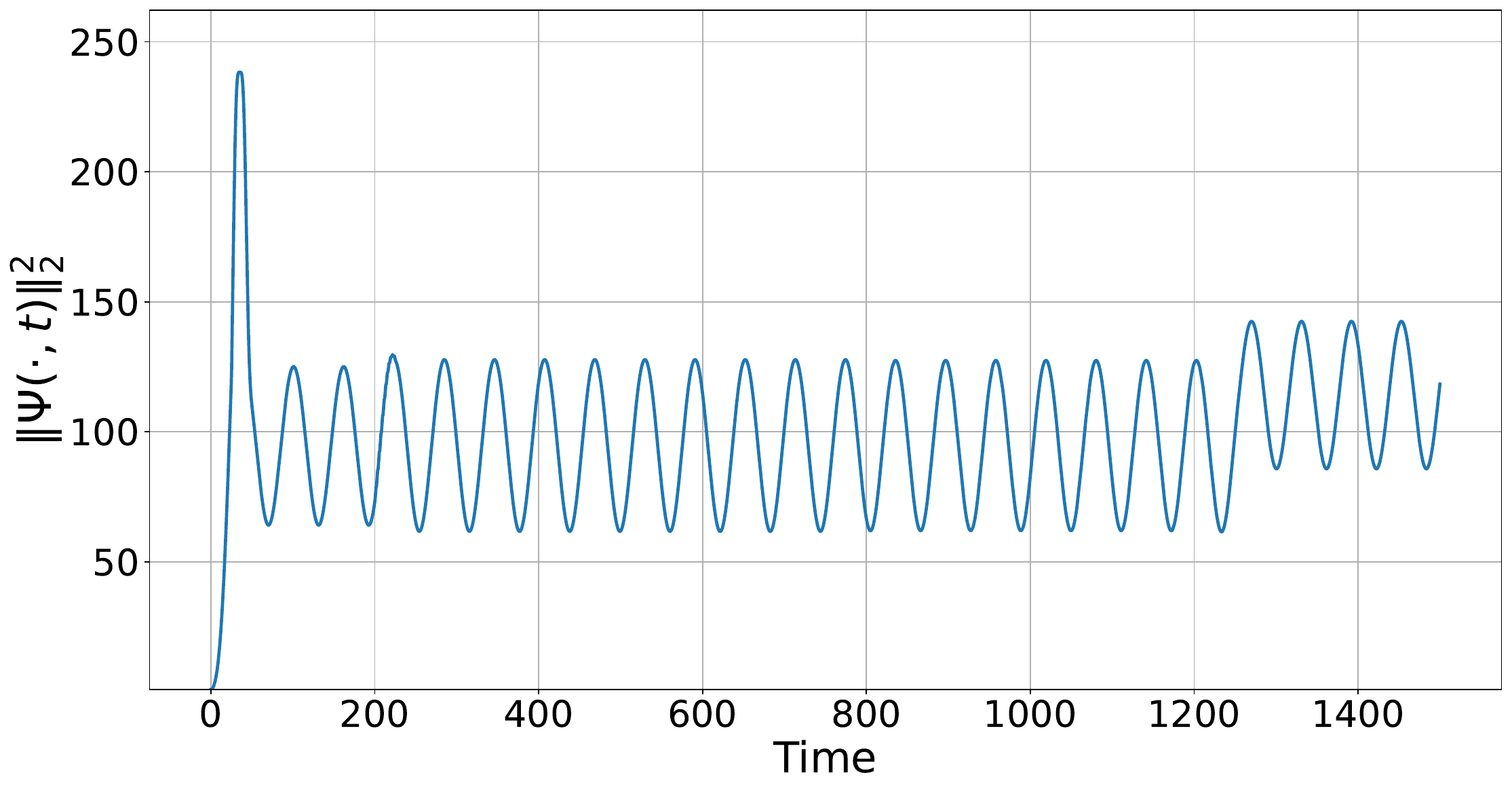}
        \caption{A numerical approximation for $\left\Vert\Psi \left( \cdot,t\right) \right\Vert _{2}^{2}$.}
        \label{fig:simulation_9_B_norm_trajectories}
    \end{subfigure}
    \caption{Results for Simulation 4.3}
    \label{Figure8}
\end{figure}

\subsection{\label{Num_Sim_5}Numerical simulations 5}
(Validates the sensitivity of the interaction-driven regime to the magnitude $\Vert W\Vert_\infty$; part of Group C -- see also the sensitivity discussion of response 3.2 in the accompanying rebuttal.)
The parameters are $p=2$, $l=6$, and $\alpha=1.6$, 
\[Z(x,t)=\sin(0.1 \pi t)1_{[25,50]}(t) + \sin(0.1 \pi t) 1_{[200,225]}(t) + (0.5) 1_{(225,\infty)}(t),
\]
and the initial condition is $\Psi_0(x) =0 + 0i$. The time runs over $[0,600]$ with a step size $\delta t = 0.001$. We perform a simulation with $W=10.0W_{cat}$ (see  Figure \ref{figure cat_matrix}). The simulations (see Figure \ref{Figure9}) show that the network's response is affected by $\left\Vert W\right\Vert_{\infty}$. In this simulation, the network responds poorly to the pulses $Z(x,t)$.

\begin{figure}[H]
    \centering

    \begin{subfigure}[t]{0.50\linewidth}
        \centering
        \includegraphics[width=\linewidth]{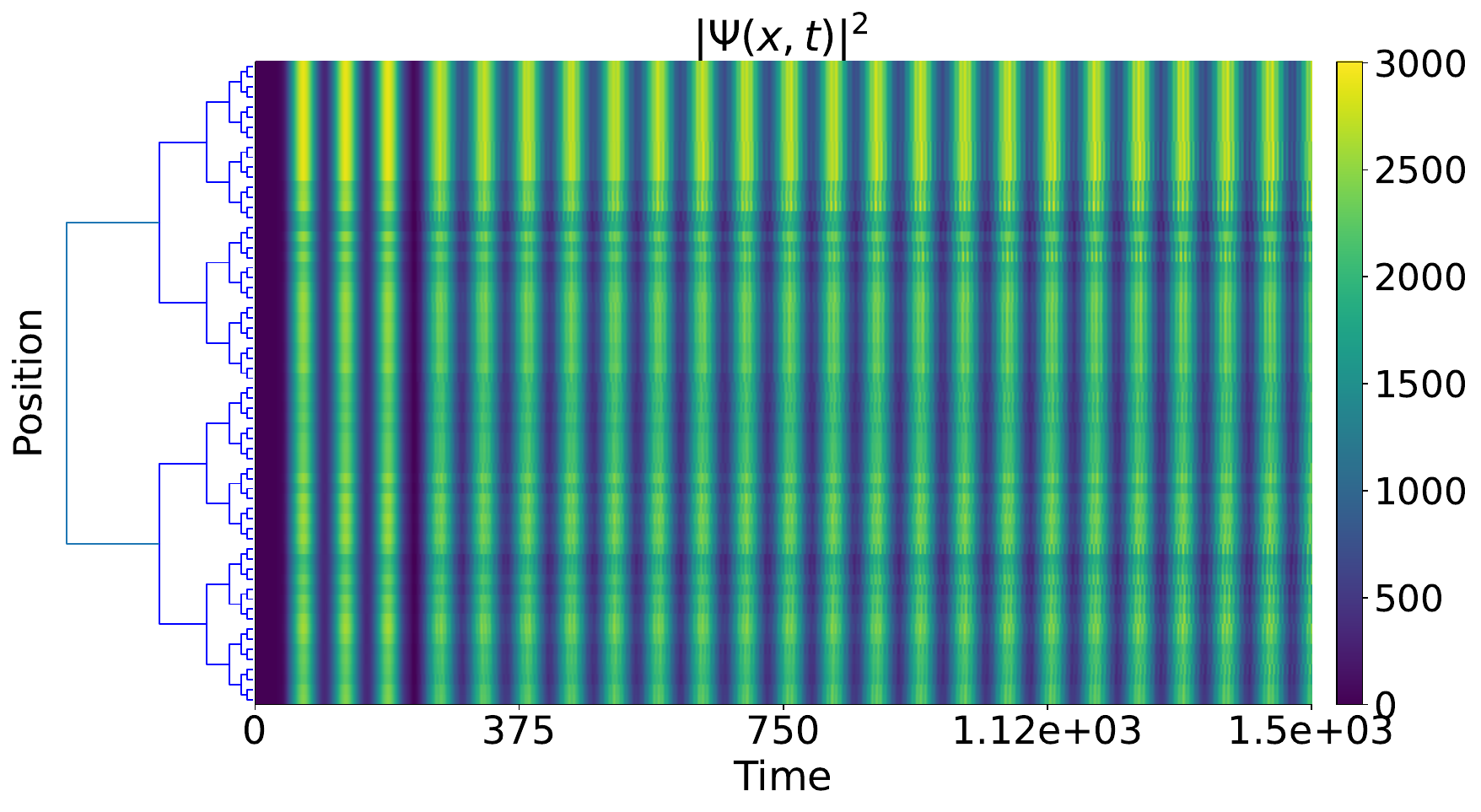}
        \caption{A numerical approximation for $\left\vert \Psi \left( x ,t\right) \right\vert ^{2}$.}
        \label{fig:sim3_10_l1_norm_sq}
    \end{subfigure}
    \hfill
    \begin{subfigure}[t]{0.48\linewidth}
        \centering
        \includegraphics[width=\linewidth]{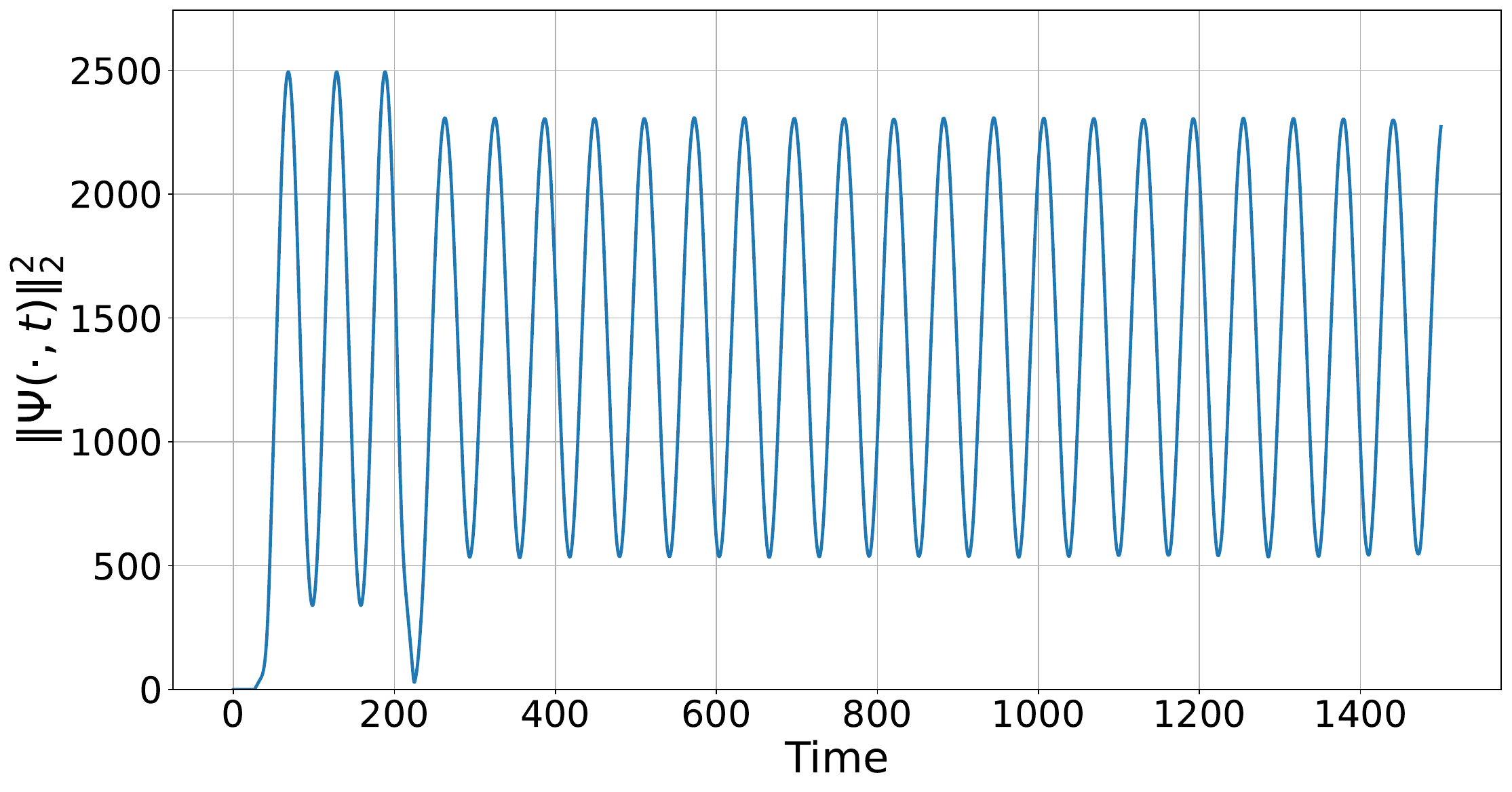}
        \caption{A numerical approximation for $\left\Vert\Psi \left( \cdot,t\right) \right\Vert _{2}^{2}$.}
        \label{fig:sim3_10_norm_trajectories}
    \end{subfigure}

    \caption{Results for Simulation 5}
    \label{Figure9}
\end{figure}

\medskip
\textbf{Group D: Classical (real-valued) $p$-adic CNN, for comparison.} The final group returns to the classical (pre-Wick-rotation) $p$-adic CNN of \mbox{\cite{Zambrano-Zuniga-1}}, i.e. the real-valued heat-type equation \mbox{\eqref{Model 3}} below, so that the quantum habituation phenomena of Group C can be compared directly against their classical counterpart.

\subsection{Numerical simulations 6}
(Validates that the quantum habituation phenomenon of Group C has a classical, real-valued analog, and highlights how the two differ; see the discussion after equation \mbox{\eqref{Model 3}}.)
This subsection provides simulations for $p$-adic CNNs of type
\begin{equation}
\frac{\partial }{\partial t}u\left( x,t\right) =J\left( \left\vert
x\right\vert _{p}\right) \ast u\left( x,t\right) -u\left( x,t\right)
+\int\limits_{\mathbb{Z}_{p}}W\left( x,y\right) \phi \left( u\left(
y,t\right) \right) dy+Z(x,t),  \label{Model 3}
\end{equation}%
where $J:\left[ 0,1\right] \rightarrow \mathbb{R}_{\geq 0}$, and $\int_{%
\mathbb{Z}_{p}}J\left( \left\vert x\right\vert _{p}\right) dx=1$. The quantum CNNs are obtained by Wick-rotating the heat equation above. 
In the simulations, we use the parameters $p=2$, $l=6$, $\alpha=2.5$, the initial condition $u_0(x) =0 $.
We first consider the case $W=0$, $Z(x,t)=\sin(0.01 \pi t) 1_{[2,10]}(t)$,   $t\in [0,100]$, with $\delta t = 0.001$. Figures \ref{Figure11}-(A) and (B) show the response of the network to the pulse under the condition that the neurons do not interact. In the next numerical experiment, we keep all the parameters but change $W$ to $0.05W_{cat}$, and $t\in[0, 400]$. The results are shown in Figures \ref{Figure12}-(A) and (B). The interaction between the neurons produces a faster decay of the output measured using $\left\Vert\Psi \left( \cdot,t\right) \right\Vert _{2}^{2}$.
In the next experiment, we keep all the parameters but change 
\[Z(x,t)=\sin(0.01 \pi t) 1_{[25,50]}(t) + \sin(0.01 \pi t) 1_{[200,225]}(t) + (0.5) 1_{[225, \infty)}(t),\]
for $t\in [0,600]$, and $W=0.1W_{cat}$. The results are shown in Figures \ref{Figure13}-(A) and (B). The network has a constant steady state, which we interpret as habituation, that differs from that shown by the quantum CNN.

\begin{figure}[H]
    \centering

    \begin{subfigure}[t]{0.50\linewidth}
        \centering
        \includegraphics[width=\linewidth]{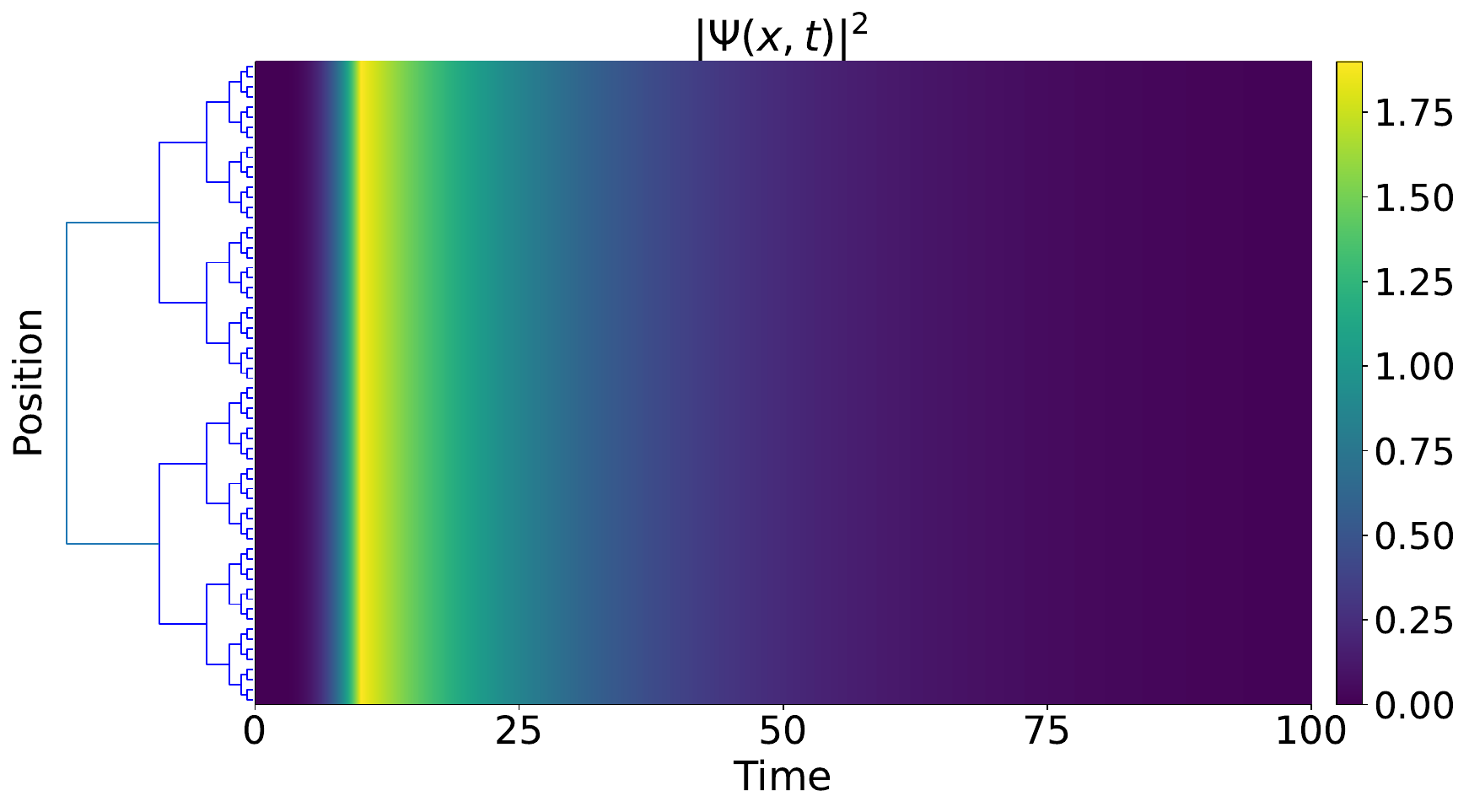}
        \caption{A numerical approximation for $\left\vert \Psi \left( x ,t\right) \right\vert ^{2}$.}
        \label{fig:sim2_5_3_l1_norm_sq}
    \end{subfigure}
    \hfill
    \begin{subfigure}[t]{0.48\linewidth}
        \centering
        \includegraphics[width=\linewidth]{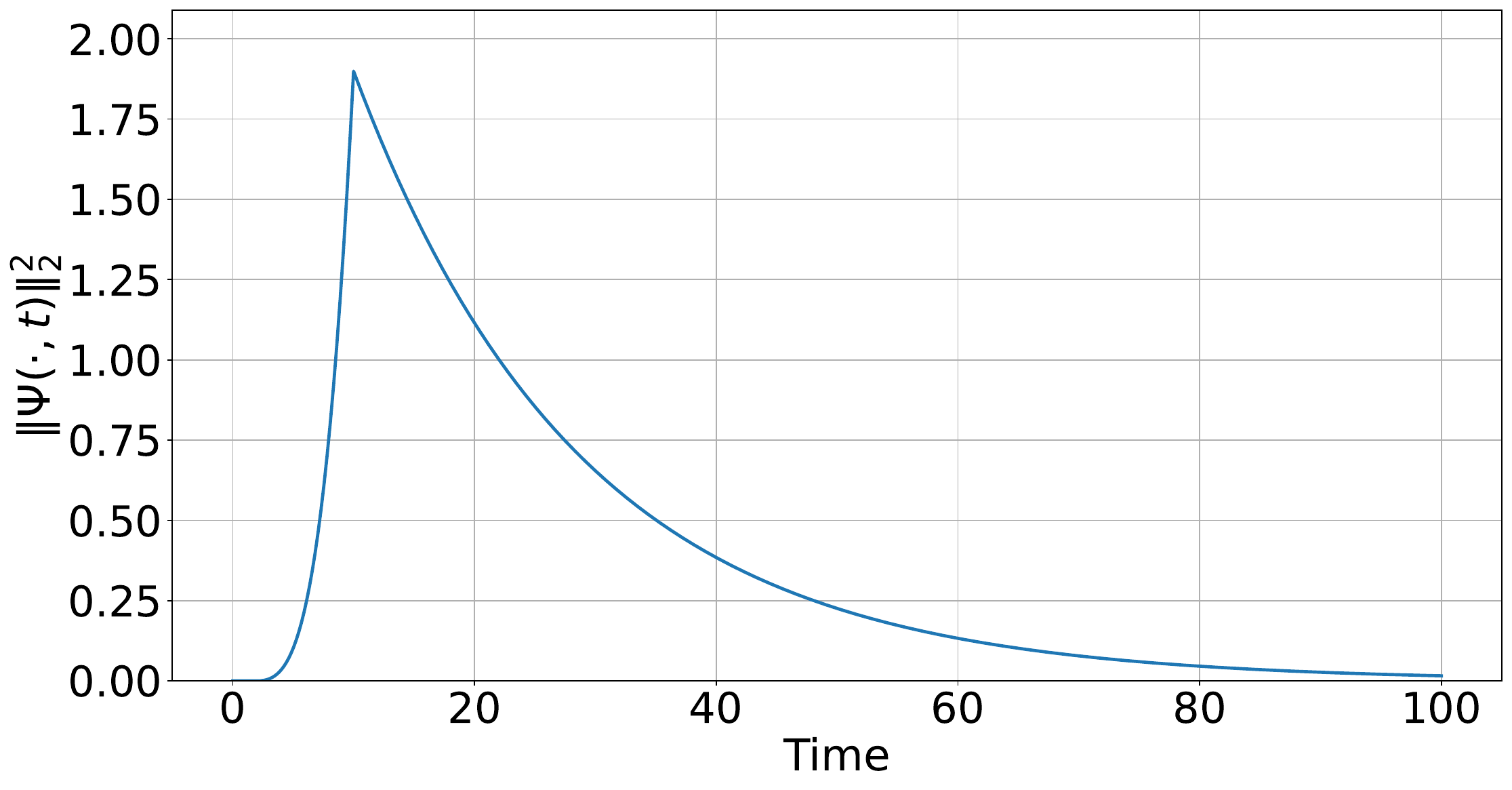}
        \caption{A numerical approximation for $\left\Vert\Psi \left( \cdot,t\right) \right\Vert _{2}^{2}$.}
        \label{fig:sim2_5_3_norm_trajectories}
    \end{subfigure}

    \caption{Results for Simulation 6.1}
    \label{Figure11}
\end{figure}

\begin{figure}[H]
    \centering

    \begin{subfigure}[t]{0.50\linewidth}
        \centering
        \includegraphics[width=\linewidth]{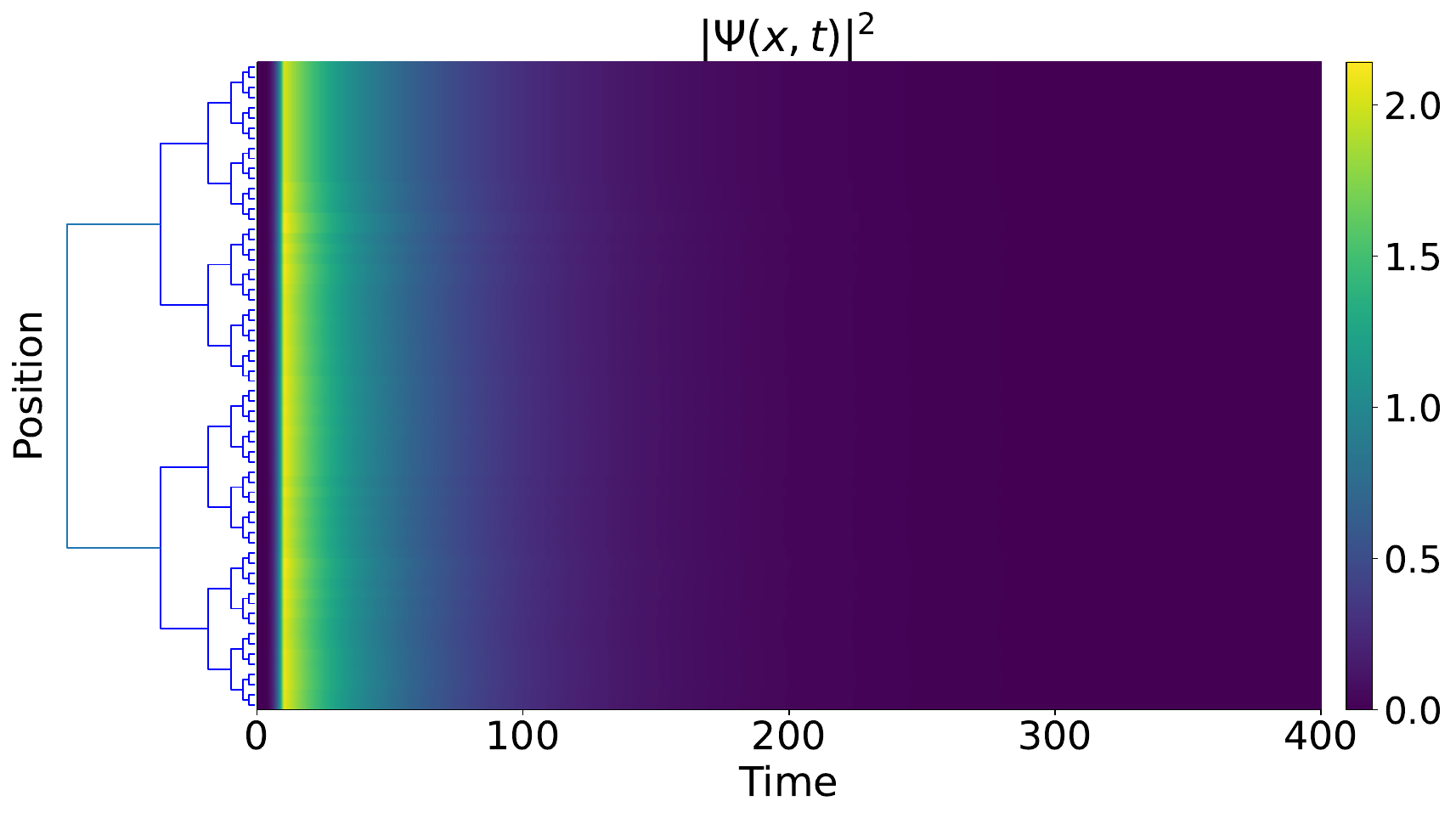}
        \caption{A numerical approximation for $\left\vert \Psi \left( x ,t\right) \right\vert ^{2}$.}
        \label{fig:catmatrix_sim1_3_l1_norm_sq}
    \end{subfigure}
    \hfill
    \begin{subfigure}[t]{0.48\linewidth}
        \centering
        \includegraphics[width=\linewidth]{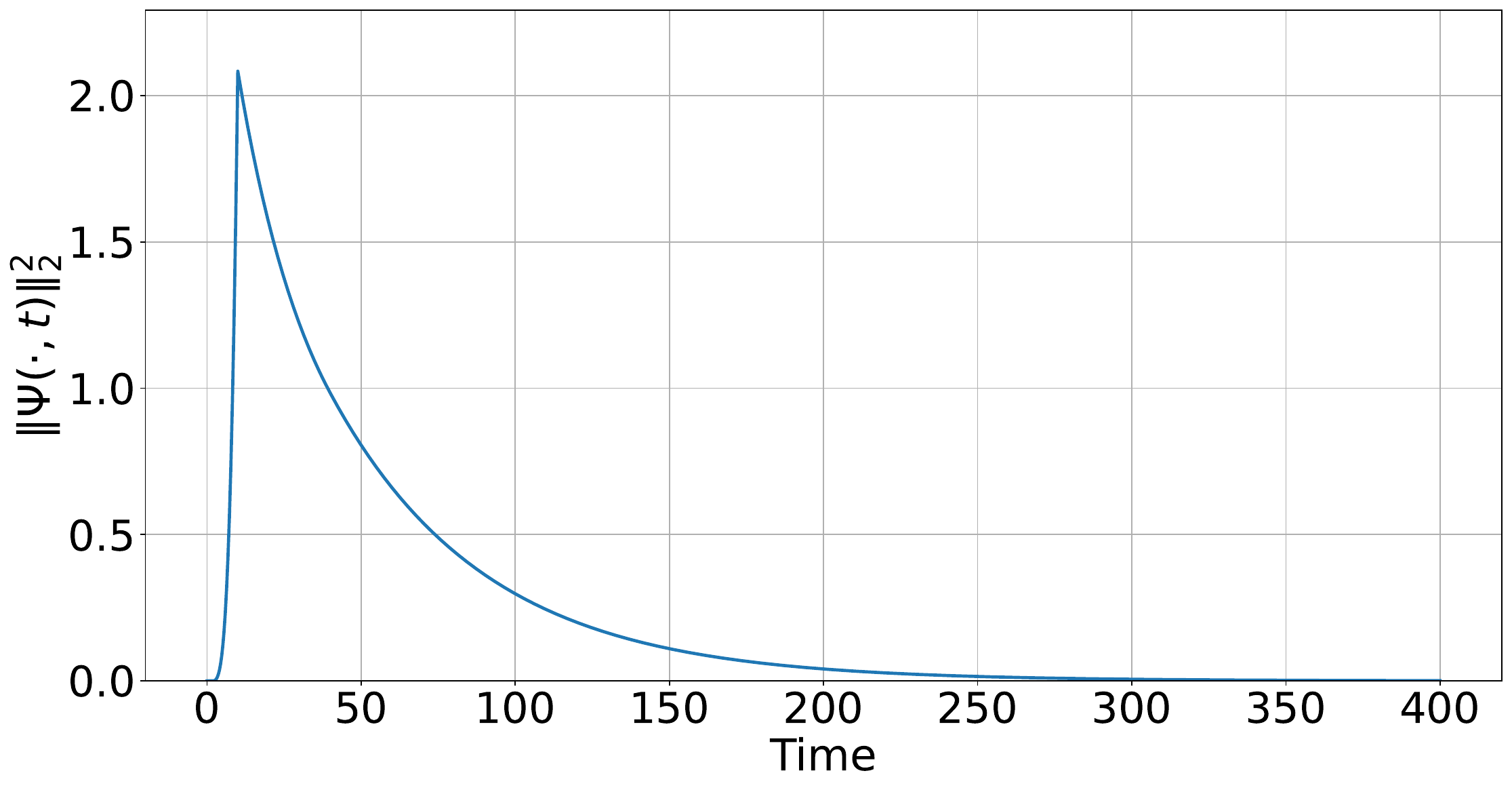}
        \caption{A numerical approximation for $\left\Vert\Psi \left( \cdot,t\right) \right\Vert _{2}^{2}$}
        \label{fig:catmatrix_sim1_3_norm_trajectories}
    \end{subfigure}

    \caption{Numerical simulation 6.2}
    \label{Figure12}
\end{figure}

\begin{figure}[H]
    \centering

    \begin{subfigure}[t]{0.50\linewidth}
        \centering
        \includegraphics[width=\linewidth]{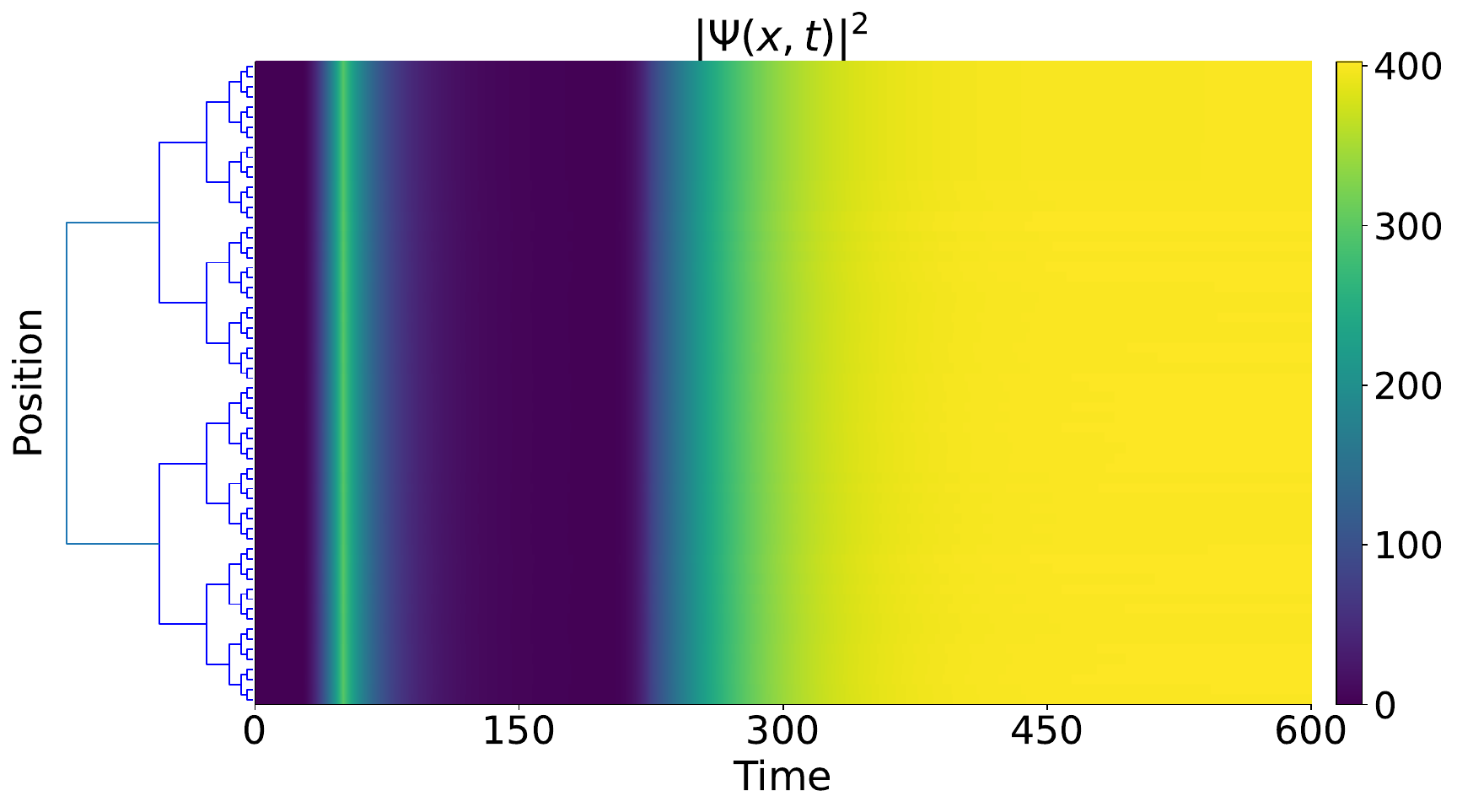}
        \caption{A numerical approximation for $\left\vert \Psi \left( x ,t\right) \right\vert ^{2}$.}
        \label{fig:catmatrix_sim1_1_l1_norm_sq}
    \end{subfigure}
    \hfill
    \begin{subfigure}[t]{0.48\linewidth}
        \centering
        \includegraphics[width=\linewidth]{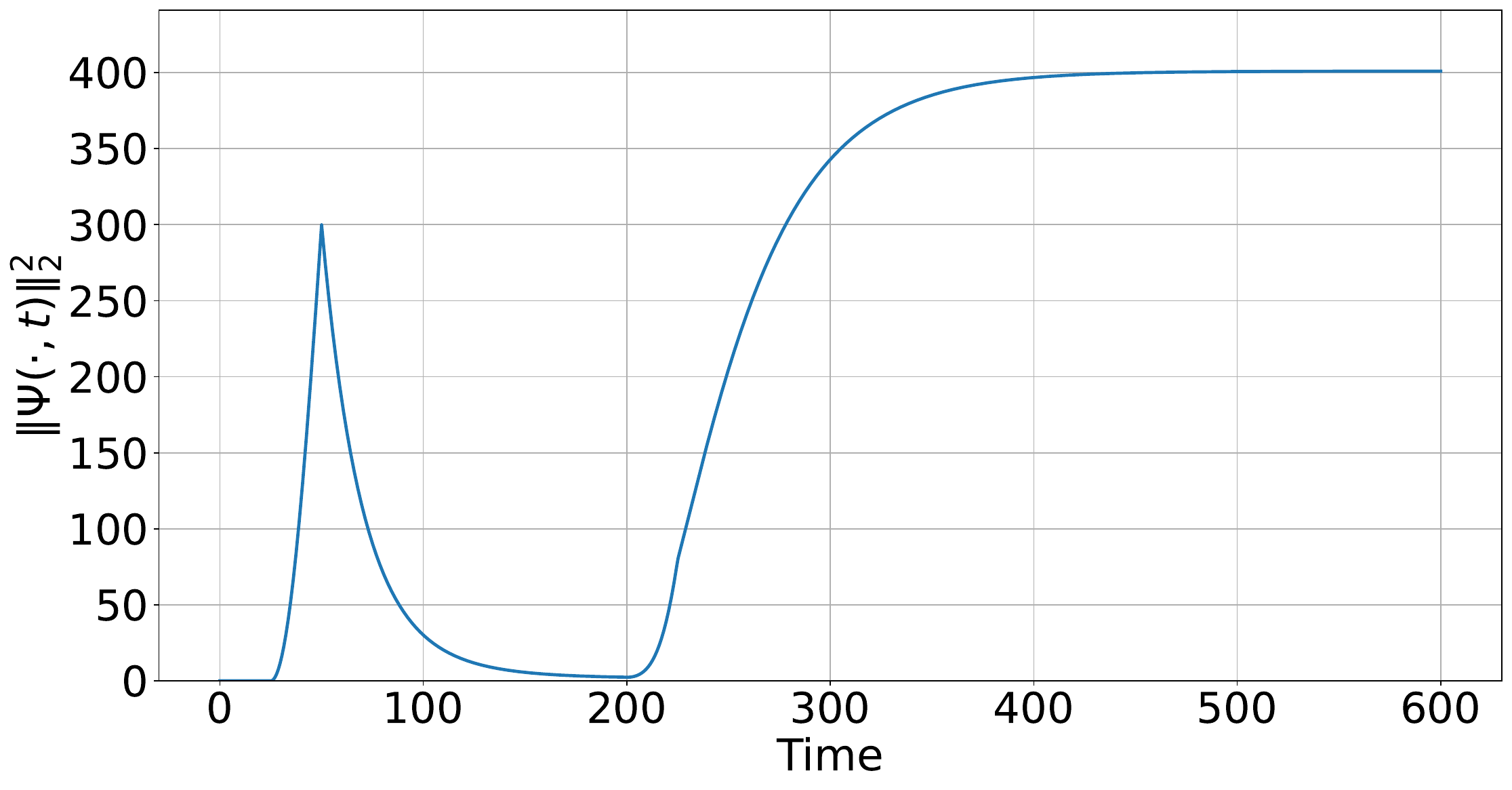}
        \caption{A numerical approximation for $\left\Vert\Psi \left( \cdot,t\right) \right\Vert _{2}^{2}$}
        \label{fig:catmatrix_sim1_1_norm_trajectories}
    \end{subfigure}

    \caption{Results for the simulation 6.3}
    \label{Figure13}
\end{figure}

\section{\label{Section_7A}Computational Complexity and Scalability}

We discuss the computational cost of simulating the $p$-adic QCNNs
classically, identify where the hierarchical structure of the $p$-adic
domain yields savings relative to a na\"{\i}ve implementation, and comment
on the prospects for quantum implementation and future training.

\textbf{State space dimension}.
After discretization at level $l$, the network state is a vector
\[
\left[  {\Psi_{I}(t)}\right] _{I\in G_{l}}\in\mathbb{C}{^{p^{l}}},%
\]
where $G_{l}=\mathbb{Z}_{p}/p^{l}\mathbb{Z}_{p}$ has cardinality $N=p^{l}$.
Each additional level of the $p$-adic tree multiplies the number
of neurons by the prime $p$, so $N$ grows exponentially in $l$ and
polynomially in $p$ (for $l$ fixed).
In the numerical simulations of Section \ref{Section_6},  we use $l = 6$ throughout,
with $p \in \{2,3\}$, giving $N = 64$ and $N = 729$ neurons,
respectively.
Note that for $p = 2$ the dimension equals $2^l$, so the network with
$l$ levels has the same Hilbert-space dimension as a system of $l$ qubits; this observation is relevant to the discussion of quantum implementation
below.

\textbf{Free evolution: exploiting convolution structure}.
The free Hamiltonian $\boldsymbol{H}_{l}=\left[  h_{I,K}\right]  _{I,K\in G_{l}}$ defined
in (\ref{Matrix_J_l})  has entries
\begin{equation}
h_{I,K}\;=\;p^{-l}\bigl(J(\left\vert I-K\right\vert _{p})-\delta
_{I,K}\bigr),\qquad I,K\in G_{l},\label{eq:hamiltonian_entries}%
\end{equation}
which depend only on the $p$-adic distance $\left\vert I-K\right\vert _{p}$
between the two nodes.
Consequently, ${\boldsymbol{H}}_{l}$ acts as a convolution operator on the finite
additive group $G_{l} \cong \mathbb{Z}/p^l\mathbb{Z}$.
The matrix-vector product $\boldsymbol{H}_{l}  \left[  {\Psi_{I}(t)}\right]$ can therefore be computed in
\begin{equation}
  O(N\log N) \;=\; O(p^l \cdot l\log p)
\end{equation}
operations using the fast Fourier transform (FFT) on $G_{l}$, rather than
the $O(N^2) = O(p^{2l})$ operations required by a generic
dense matrix-vector multiplication.
For the largest network used in our simulations
($p = 3$, $l = 6$, $N = 729$), the FFT-based approach reduces the
free-part cost per time step from $O(729^2) \approx 5.3\times 10^5$
to $O(729 \times 6) \approx 4.4\times 10^3$ arithmetic operations,
a reduction by a factor of roughly $120$.
This saving is a direct consequence of the ultrametric, tree-based
geometry of $\mathbb{Z}_{p}$, which endows $\boldsymbol{H}_{l}$ with an algebraically exact
hierarchical block structure that is absent in Euclidean-domain QNNs.

\textbf{Non-linear interaction term}.
The non-linear interaction term in the discretized equation (\ref{Network_Type_I}) reads
\begin{equation}\label{eq:interaction_term}
  \sum_{K \in G_{l}} W_{I,K}\,\phi(\Psi_K(t)),
  \qquad I \in G_{l},
\end{equation}
and requires $O(N^2)$ operations per time step for a general weight
matrix $W = [W_{I,K}]$.
In our simulations, three structural choices for $W$ are considered.

\begin{enumerate}
\item $W(x,y)=0$ (no interaction). The interaction term vanishes, and the cost
per step is $O(N\log N)$ from the free part alone.

\item $W(x,y)=W_{0}$ constant. The sum in \eqref{eq:interaction_term}
collapses to $W_{0}\sum_{K}\phi(\Psi_{K}(t))$, which costs only $O(N)$ per step.

\item $W(x,y)=c\,W_{cat}$ (cat-cortex approximation) of \cite{Zuniga-Entropy},
used with $p=2$, $l=6$. Here $W_{cat}$ is a general $64\times64$ matrix, so
the interaction term costs $O(N^{2})$ per step. In this regime, the interaction
term dominates, and the total cost per step is $O(N^{2})$.
\end{enumerate}

In cases where the weight kernel $W(x,y)$ inherits structure from the
$p$-adic geometry --- for instance, if $W$ is itself a convolution
operator on $G_{l}$ --- the interaction term is also reducible to
$O(N\log N)$ per step, matching the cost of the free part.
Whether biologically or task-motivated, the tendency of weight matrices to possess such structure is an interesting question for future investigation.

\textbf{Total simulation cost}.
All simulations in Section \ref{Section_6} use a fixed time step $\delta t = 0.001$.
Over a time horizon $[0,T]$ the number of steps is $T/\delta t$, and the
total classical simulation cost is
\begin{equation}\label{eq:total_cost}
  C_{\mathrm{sim}}
  \;\sim\;
  \frac{T}{\delta t}
  \times
  \begin{cases}
    O(N\log N)
      & \text{if } W = 0
        \text{ or } W = W_0
        \text{ (constant),}\\[4pt]
    O(N^2)
      & \text{if } \mathbf{W}
        \text{ is a general dense matrix.}
  \end{cases}
\end{equation}
For the longest simulation in Section \ref{Section_6}
($T = 1500$, $p = 2$, $l = 6$, $W = c\,W_{cat}$)
formula~\eqref{eq:total_cost} yields approximately
\[
  1.5 \times 10^6 \;\times\; 64^2
  \;=\; 1.5\times 10^6 \times 4{,}096
  \;\approx\; 6\times 10^9
  \text{ arithmetic operations,}
\]
a computation that is entirely feasible on modern workstation hardware.
Scaling to $l = 10$ with $p = 2$ (i.e., $N = 1{,}024$ neurons) would
increase the dense-$W$ cost by a factor of
$(1{,}024/64)^2 = 256$ relative to $l = 6$,
placing such simulations at the boundary of what is practical without
structured sparsity, low-rank approximations, or GPU-accelerated linear
algebra.
Scaling in $l$ beyond approximately $10$--$12$ will therefore require
either structured weight matrices, sparse approximations to $W$,
or a transition to quantum hardware.

\textbf{Scalability in $p$ and $l$}.
The parameters $p$ and $l$ play distinct structural roles.
Increasing $l$ refines the resolution of the $p$-adic tree by adding
one more hierarchical layer of neurons, multiplying $N$ by $p$ at each
step.
Increasing $p$ changes the \emph{branching factor} of the tree --- that
is, the number of children each node has at each level --- without
changing the depth $l$.
For fixed $l$, a larger $p$ produces a wider hierarchy with
more neurons per level, whereas a smaller $p$ produces a narrower,
deeper hierarchy.
The choice $p = 2$ is the natural analog of binary trees and aligns
with qubit-based quantum hardware (see below), while $p = 3$ provides
finer intra-level resolution at the same depth.
The computational cost scales as $O(p^{2l})$ in the dense-$W$
regime and as $O(p^l \cdot l)$ in the structured or free regime,
so moderate values of $p$ ($p = 2$ or $p = 3$) are most practical for
classical simulation at depths $l \leq 10$.

\textbf{Prospects for quantum implementation}.
The free Schr\"odinger evolution $e^{it\boldsymbol{H_{l}}}$ is a unitary operator
and is, in principle, implementable as a quantum circuit.
The convolution structure of $\boldsymbol{H}_{l}$ on
$G_{l} $ suggests that a quantum Fourier
transform (QFT) based decomposition could yield an efficient circuit;
for $p = 2$ the QFT on $G_{l}$ coincides with the Walsh-Hadamard
transform on $l$ qubits, implementable in $O(l^2)$ two-qubit gates.
This means that the free evolution of the $p$-adic QCNN could, in
principle, achieve an exponential speedup over its classical simulation
on $l$ qubits, analogously to the quantum advantage established for
CTQWs \cite{Farhi-Gutman}-\cite{Venegas-Andraca}.
The non-linear term $\phi(\Psi)$, however, does not correspond to a
unitary operation, and its direct implementation on standard gate-based
quantum hardware is non-trivial.
This is a fundamental challenge shared by all nonlinear QNN proposals
\cite{Nakajima et al}-\cite{Behera et al}, \cite{Gupta, Schuld et al}, not specific to the $p$-adic framework.
For the open-system dynamics ($W \neq 0$), a Lindblad-type master-equation
simulation on quantum hardware would be required, which is an active area of current research.
We therefore regard the quantum implementation of the full nonlinear
$p$-adic QCNN as an important open problem, and note that the free
($W = 0$, $Z = 0$) case --- which already recovers CTQWs on
graphs --- is the most immediately accessible to near-term quantum
hardware.

\textbf{Memory cost}. The complexity analysis above concerns arithmetic (time) cost; memory requirements follow a parallel pattern. Storing the state vector $\Psi(t)\in\mathbb{C}^N$ requires $O(N)$ memory. The free-evolution matrix $\boldsymbol{H}_l$ need not be stored densely: because it is a convolution operator on $G_l$ (paragraph 2 above), only its $O(N)$ generating values $J^{(l)}_{I,0}$ need to be kept in memory, with the matrix-vector product computed on the fly via the FFT; storing $\boldsymbol{H}_l$ as a dense $N\times N$ matrix would instead cost $O(N^2)$. The memory cost of the interaction weight matrix $W$ mirrors the three structural cases of paragraph 3: $O(1)$ when $W=0$, $O(1)$ when $W=W_0$ is constant (since only the scalar $W_0$ need be stored), and $O(N^2)$ for a general dense matrix such as the cat-cortex approximation $W_{cat}$ used in several of the simulations of Section \mbox{\ref{Section_6}}. For the largest network simulated here ($N=729$), a dense $N\times N$ complex matrix occupies a few megabytes in double precision, well within the memory of a standard workstation; dense storage of $W$ becomes the binding memory constraint, rather than $\boldsymbol{H}_l$ or $\Psi$, as $N$ grows.

\textbf{Training complexity}.
The present work does not address the training of network parameters;
the weight kernel $W$, the diffusion kernel $J$, and the bias $Z$ are
chosen manually for each numerical experiment.
Were a gradient-based training scheme to be applied --- for instance,
by differentiating the ODE solution with respect to $W$ using
adjoint (backpropagation-through-time) methods --- the additional cost
per gradient step would scale as $O(N^2)$ for a general weight
matrix, matching the forward simulation cost.
The number of trainable parameters is $O(N^2)$ for a general
$W$, $O(N)$ for a diagonal or radially symmetric $W$, and
$O(1)$ for the single scale parameter $\alpha$ of the kernel
$J_\alpha$.
Developing efficient quantum machine learning algorithms \cite{QML} for
optimizing the parameters of networks of type (\ref{EQ_0A}) remains a central
open problem that we identify as a primary direction for future work.

\section{\label{Section_7}Final discussion}
The construction of quantum analogs of biological neural networks has been
intensively studied over the last 30 years. The first ideas on quantum neural
computation were published independently in 1995 by S. Kak and R. Chrisley 
\cite{Kak, Chrisley}. Since then, multiple architectures for such networks have been
proposed. Here, we introduce a new class of QNNs that are quantum analogs of
$p$-adic CNNs. The mathematical formulation is rigorous and is framed in the
Dirac--von Neumann formalism. The new $p$-adic QCNNs are a generalization of the
QNNs introduced by the first author in \cite{Zuniga-QM-2}, which are stochastic automata.
Such networks can be used to construct continuous-time quantum walks on graphs,
which are widely used in quantum computing. Furthermore, the construction of the
new networks is bio-inspired by the Wilson--Cowan model. This work raises several
research questions across many areas.

\textbf{Key achievements.} Before turning to the broader discussion, we summarize explicitly what this manuscript establishes. Mathematically, we (i) introduce a new class of nonlinear $p$-adic Schr\"{o}dinger equations, obtained as the Wick rotation of the state equation of the $p$-adic CNNs of \mbox{\cite{Zambrano-Zuniga-1}}; (ii) prove global existence and uniqueness of solutions to the associated Cauchy problem under the hypothesis $\phi\in L^\infty(\mathbb{R})$ (Theorem \mbox{\ref{Theorem_main}}, Appendix); and (iii) give a rigorous discretization scheme that produces finite-dimensional QNNs on simple graphs (Sections \mbox{\ref{Section_4}}--\mbox{\ref{Section_5}}). Structurally, we (iv) provide a systematic, criterion-by-criterion comparison against six existing QNN families that clarifies the novelty of the present approach (Section \mbox{\ref{Section_6A}}). Computationally, we (v) characterize, through representative numerical experiments, how the interaction weight $W$ and bias $Z$ drive the transition from unitary to non-unitary evolution, including a habituation-like phenomenon (Section \mbox{\ref{Section_6}}); and (vi) analyze the time and memory complexity of classical simulation and the prospects for quantum hardware implementation (Section \mbox{\ref{Section_7A}}).

It is important to distinguish, among the above, results that are \emph{rigorously established} from those that are \emph{numerical observations}. Items (i)--(iv) are proved or constructed rigorously and hold unconditionally (within the stated hypotheses). Items (v)--(vi) are, respectively, an empirical/illustrative characterization drawn from a finite set of representative simulations (Section \mbox{\ref{Section_6}}) and an analytical-but-asymptotic complexity estimate; neither constitutes a theorem. In particular, the decoherence and habituation phenomena reported in Section \mbox{\ref{Section_6}} are \emph{interpretations of numerical output}, consistent with the mathematical formulation but not proved from it, and their generality across all parameter choices, graph structures, and activation functions has not been established -- see item (7) of the Limitations below. We maintain this distinction throughout the remainder of this section.

\subsection{Physical and Computational Advantages of the \texorpdfstring{$p$}{p}-Adic Design}

We identify four structural properties of the $p$-adic framework that jointly
distinguish the networks introduced here from existing QNN architectures and that
confer specific physical and computational advantages.

\smallskip
\noindent\textbf{Algebraically exact hierarchical topology.}
The $p$-adic integers $\mathbb{Z}_p$ carry a natural ultrametric topology induced by the $p$-adic absolute value $\left\vert \cdot\right\vert _{p}$. Two neurons $x, y \in \mathbb{Z}_p$ satisfy
$\left\vert x-y\right\vert _{p}\leq p^{-k}$ if and only if they belong to the same ball of radius
$p^{-k}$, i.e., to the same subtree at depth $k$ of the $p$-adic tree.
Hierarchical organization is therefore an intrinsic algebraic feature of the
domain ${\mathbb{Z}}_{p}$, not an externally imposed circuit layout or graph structure.
This contrasts with Euclidean-domain QNNs, where hierarchy must be artificially
introduced --- for instance, through pooling layers in quantum convolutional neural
networks \cite{QCNN} or through a prescribed graph in quantum graph neural networks
\cite{Verdon et al, Ceschini et al}.
In the $p$-adic framework, the range of interaction between neurons is determined
entirely by their $p$-adic distance, and the prime $p$ controls the branching
factor of the hierarchy, providing a mathematically clean one-parameter family
of hierarchical architectures.
This property is particularly well-suited to modeling biological neural systems,
in which neurons are organized hierarchically across spatial scales (cells, columns,
areas, lobes), a structure that is naturally captured by the successive quotients
$G_{l}=\mathbb{Z}_{p}/p^{l}\mathbb{Z}_{p}$ of the $p$-adic integers.

\smallskip
\noindent\textbf{Continuous interpolation between closed and open quantum dynamics.}
A single governing equation (\ref{EQ_0A}) interpolates continuously between two
physically distinct regimes controlled by the weight kernel $W$ and the bias $Z$.
When $W=Z = 0$, the network undergoes strictly unitary evolution, corresponding to a continuous-time quantum Markov chain \cite{Zuniga-QM-2}.
When $W \not= 0$ or $Z \not= 0$, the norm $\left\Vert \Psi\left(  \cdot,t\right)  \right\Vert _{2}^{2}$ is no longer conserved
(as confirmed in all numerical simulations of Section \ref{Section_6}), and the network
models an open quantum system subject to environmental interaction. 
We interpret this as a Lindblad-type dynamics (see Section \ref{Section_0}), in which the coupling strength $W$ and the bias $Z$ control the degree of decoherence.
This unified treatment --- in which closed and open dynamics are special cases
of the same equation --- is not available in VQC or QGNN architectures, where
decoherence must be modeled separately through an auxiliary formalism.
The biological interpretation is equally natural: the parameter $W$ encodes the
influence of one neuronal population on another, so setting $W = 0$ isolates
individual populations (no cross-talk) while $W \not= 0$ activates the
inter-population interaction that drives habituation and pattern formation,
as observed in Sections \ref{Num_Sim_3}-\ref{Num_Sim_5}. We stress that this description of habituation is an interpretation of the numerical output of Section \mbox{\ref{Section_6}}, not a proved property of equation \mbox{\eqref{EQ_0A}}; establishing it rigorously, and determining how general it is across parameter choices and graph structures, is left for future theoretical and computational work (see Limitations item (7) and Future Directions).

\smallskip
\noindent\textbf{Strict generalization of continuous-time quantum walks.}
Setting $W = 0$ and $Z = 0$ in equation (\ref{EQ_0A}) recovers the
continuous-time quantum Markov chain on $\mathbb{Z}_{p}$ introduced in \cite{Zuniga-QM-2}, which,
after discretization to $G_{l}$, reduces to a standard continuous-time quantum
walk (CTQW) on the associated graph \cite{Farhi-Gutman}-\cite{Venegas-Andraca}.
CTQWs have been used to design quantum algorithms with provable speedup over
their classical counterparts \cite{Farhi-Gutman}. The $p$-adic QCNNs therefore subsume
all CTQW-based quantum algorithms as a special case, and extend them in two
directions: the inclusion of a nonlinear activation function $\phi$, which
breaks the linearity of the Schr\"odinger equation and introduces
interaction-driven pattern formation; and the coupling $W \not= 0$, and the bias $Z \not= 0$, which transitions the system from a closed to an open quantum network.
Whether these two extensions can be exploited to obtain new quantum algorithmic
advantages --- beyond those already known for CTQWs --- is an important open problem that we leave for future investigation.

\smallskip
\noindent\textbf{Convolution structure and computational efficiency.}
The free Hamiltonian $\boldsymbol{H}_{l}$ defined in (\ref{Matrix_J_l}) is a convolution operator on the finite group $G_{l} \cong \mathbb{Z}/p^l\mathbb{Z}$, because its entries depend
only on the $p$-adic distance $\left\vert I-K\right\vert _{p}$ between nodes.
As a consequence, the matrix-vector product ${\boldsymbol{H}}_{l} \Psi$ is computable in
$O(p^l \cdot l\log p)$ operations via the fast Fourier transform on $G_{l}$,
compared with $O(p^{2l})$ for a generic dense matrix.
This saving is a direct consequence of the ultrametric geometry and has no
analog in Euclidean-domain QNNs, whose Hamiltonians do not generally admit such
factorization.
For the $p = 2$ case the FFT on $G_{l}$ coincides with the Walsh-Hadamard
transform on $l$ bits, which is implementable in $O(l^2)$ two-qubit gates
on a quantum processor, opening a path toward exponential quantum speedup in
the linear regime.
A detailed discussion of the full computational complexity of the $p$-adic
QCNNs --- including the interaction term and the training cost --- is given in
Section \ref{Section_7A}.

\subsection{Connection to Quantum Models of Brain Activity}

Currently, there is a strong interest in studying quantum networks as analogs of
biological neuronal networks, to explore whether brain activity is rooted in
quantum mechanics. We mention two approaches here. The Quantum-Like Modeling
(QLM) framework treats neuronal networks as complex systems that deviate from
classical logic \cite {Khrennikov-1, Khrennikov-2}; it uses quantum probability theory to describe cognitive effects such as decision interference and bistable perception.
The second approach is based on the Lipkin--Meshkov--Glick quantum
Hamiltonian \cite{LMG}. The $p$-adic QCNNs introduced here are bio-inspired by
the Wilson--Cowan model, which has been widely used to describe brain activity;
consequently, the present work is related to \cite{Khrennikov-1}-\cite{LMG}, but the specific connection constitutes an open problem that we plan to investigate in a future
contribution.

\subsection{Limitations}
\label{subsec:limitations}

We collect the limitations of the present work here in one place; readers most interested in the framework's current scope and practical implementability may wish to consult this subsection immediately after the Key Achievements paragraph above. We identify the following limitations of the present work, which also define
the boundaries within which the results should be interpreted.

\begin{enumerate}

\item \textbf{Purely theoretical framework.} The paper establishes the
mathematical foundations of the $p$-adic QCNNs and validates the framework
through numerical simulation on networks of size $N = p^{l} \leq 729$.  No
hardware implementation --- whether on gate-based, photonic, or analog
quantum processors --- is provided or claimed. The question of which quantum
platform could most naturally realize the $p$-adic Schr\"odinger dynamics
remains open; see item~(4) below.

\item \textbf{Well-posedness under restricted hypotheses.} Theorem \ref{Theorem_main}
establishes global existence ($T = \infty$) and uniqueness of solutions to
the Cauchy problem (\ref{EQ_IVP}) under the hypothesis $\phi\in L^{\infty
}(\mathbb{R})$. Global well-posedness for the general nonlinear
equation \eqref{EQ_0A} --- without the boundedness assumption on $\phi$, and for
more singular initial data or weight kernels $W$ --- has not been established
and remains a central mathematical open problem.

\item \textbf{Untrained parameters.} In all numerical experiments of
Section \ref{Section_6}, the parameters $J$, $W$, and $Z$ in (\ref{EQ_0A}) are chosen manually on
the basis of physical or biological considerations. A systematic training
procedure --- selecting these parameters via an optimization process adapted
to a specific task --- has not been developed for the $p$-adic QCNN
architecture. This is the most significant practical limitation of the
present work in relation to deployed machine learning systems. We emphasize that this scope is deliberate: the present manuscript focuses on the framework's mathematical foundations, not data-driven learning; see the ``Training and quantum machine learning'' paragraph in Future Directions below.

\item \textbf{Quantum hardware implementation.} The nonlinear activation term
$\phi(\Psi)$ in (\ref{EQ_0A}) does not correspond to a unitary gate operation,
making direct implementation on standard gate-based quantum hardware
non-trivial. This difficulty is not unique to the $p$-adic framework but is
shared by all proposals for nonlinear QNNs \cite{Nakajima et al}-\cite{Behera et al}, \cite{Gupta,Schuld et al}. For the
open-system regime ($W \not = 0$, $Z \not = 0$), implementation would require a Lindblad
master-equation simulator on quantum hardware, which is an active research
direction. The linear, unitary special case ($W = 0$, $Z= 0$) is the
most immediately accessible to near-term devices.

\item \textbf{Entanglement and quantum correlations.} The present analysis
does not characterize the entanglement structure of the quantum states
produced by (\ref{Model 1}) or (\ref{EQ_0A}). Understanding whether and how the $p$-adic
hierarchical architecture generates multipartite entanglement across tree
levels would be essential for establishing any quantum computational
advantage beyond the CTQW special case, and constitutes an important
theoretical gap.

\item \textbf{Classical simulation scalability.} Classical simulation of the
$p$-adic QCNNs with a general dense weight matrix $W$ costs $O(N^{2})$ per
time step, limiting direct simulation to networks of size $N\lesssim10^{3}$ on
current hardware without structured sparsity or GPU acceleration. Scaling to
$l\geq12$ (with $p=2$, $N=4096$) will require either low-rank or hierarchical
approximations to $W$, or quantum hardware. More broadly, a comprehensive quantitative evaluation of the framework -- covering robustness to perturbations of the initial data or parameters, and dependence on the graph structure/topology rather than only on the graph size $N$ -- has not yet been carried out and is left for future work.

\textbf{Generality of the numerical observations and absence of an experimental benchmark.} The simulations of Section \mbox{\ref{Section_6}} use representative, hand-chosen values of $\alpha$, of the weight kernel $W$ (zero, constant, or the specific cat-cortex approximation $W_{cat}$ of \mbox{\cite{Zuniga-Entropy}}), of the activation function $\phi$, and of a small number of graph discretizations ($p\in\{2,3\}$, $l=6$); establishing that the reported phenomena -- loss of unitarity, habituation, pattern formation -- persist across the full parameter space, across other activation functions, and across more general graph topologies requires the broader sensitivity analysis outlined in Future Directions below, and has not been carried out here. Relatedly, while Section \mbox{\ref{Section_6A}} situates the $p$-adic QCNNs structurally and qualitatively against six existing QNN families, it does not provide an experimental or quantitative benchmark against actual implementations of those architectures; such a comparison remains for future work.
\end{enumerate}

\subsection{Future Directions}
\label{subsec:future}

The limitations identified above suggest a concrete agenda for future research,
which we organize into four directions.

\smallskip
\noindent\textbf{Mathematical analysis.}
The most pressing theoretical problem is the global well-posedness of the
Cauchy problem (\ref{EQ_IVP}) without the boundedness assumption $\phi \in L^\infty(\mathbb{R})$.
A second open problem is the existence and stability of stationary patterns
and traveling wave solutions of the form $\psi(\left\vert x\right\vert _{p}+vt)$;
we plan to address this in a forthcoming paper, using techniques from
the $p$-adic theory of reaction-diffusion equations.
A third problem is the characterization of the long-time behavior of
solutions in the open-system regime ($W \not= 0$ or $Z \not= 0$), including the possible
convergence to equilibrium states and the onset of pattern formation
driven by the interaction kernel $W$.

\smallskip
\noindent\textbf{Training and quantum machine learning.}
For practical applications, the parameters $J$, $W$, and $Z$ in (\ref{EQ_0A}) should
be selected by an optimization process adapted to a specific task.
For instance, $J$ could be optimized over the one-parameter family
$J_\alpha(\left\vert x\right\vert _{p})$ (see Section \ref{Section_6}), and $W$ could be learned from data
using gradient-based methods with adjoint (backpropagation-through-time)
sensitivity analysis, at a cost of $O(N^2)$ per gradient step for a
general dense $W$.
The development of quantum ML algorithms \cite{QML} for networks of
type (\ref{EQ_0A}) --- potentially exploiting the convolution structure of ${\boldsymbol{H}}_{l}$ via
a quantum Fourier transform --- constitutes a central open problem.

\smallskip
\noindent\textbf{Quantum algorithmic applications.}
Since the $p$-adic QCNNs with $W = 0$ and $Z = 0$ reduce to CTQWs
on graphs, all existing CTQW-based quantum algorithms (search,
state transfer, graph isomorphism testing) are immediately available in this
framework. The more interesting question is whether the open-system extension ($W \not= 0$ or $Z \not= 0$) can lead to
new quantum speedups that are not attainable with linear, unitary CTQWs.
Identifying tasks for which the hierarchical $p$-adic structure provides
a provable computational advantage is an important direction for future work.

\smallskip
\noindent\textbf{Model generality and robustness.}
As noted in Limitations item (7), the numerical study of Section \mbox{\ref{Section_6}} uses a small number of representative choices for the activation function $\phi$, the weight kernel $W$, and the underlying graph. Extending the sensitivity analysis to (i) other nonlinear activation functions beyond the piecewise-linear $\phi$ used here, (ii) more general graph topologies (e.g., non-$p$-adic trees, or graphs not arising as quotients of $\mathbb{Z}_p$), and (iii) an experimental benchmark against implementations of the QNN families compared structurally in Section \mbox{\ref{Section_6A}}, is an important and tractable direction for future work, and would place the qualitative conclusions of Section \mbox{\ref{Section_6}} on firmer quantitative footing.

\bigskip

\textbf{Data and code availability.} The simulation code used to produce every figure in this paper is available in the following  GitHub repository: https://github.com/B-A-Zambrano-Luna/ PATTERN-FORMATION-IN-QUANTUM-HIERARCHICAL-CELLULAR-NEURAL-NETWORKS.git
\section{\label{Appendix}Appendix}
We set
\[
\boldsymbol{H}_{0}\psi\left(  x\right)  =i\left\{  J\left(  \left\vert
x\right\vert _{p}\right)  \ast\psi\left(  x\right)  -\psi\left(  x\right)
\right\}  \text{, for }\psi\left(  x\right)  \in L^{2}\left(  \mathbb{Z}%
_{p}\right)  .
\]
Notice that $\left\Vert \boldsymbol{H}_{0}\psi\right\Vert _{2}\leq\left(
\left\Vert J\right\Vert _{1}+1\right)  \left\Vert \psi\right\Vert _{2}$, and
consequently $\boldsymbol{H}_{0}:L^{2}\left(  \mathbb{Z}_{p}\right)
\rightarrow L^{2}\left(  \mathbb{Z}_{p}\right)  $ is a bounded linear
operator. Using the fact that $H_{0}$ is a pseudo-differential operator,%
\[
\boldsymbol{H}_{0}\psi=\mathcal{F}^{-1}\left(  i\left(  \widehat{J}\left(
\left\vert \xi\right\vert _{p}\right)  -1\right)  \widehat{\psi}\right)  ,
\]
and the fact that the Fourier transform preserves the inner product in
$L^{2}\left(  \mathbb{Q}_{p}\right)  $, we conclude that $\boldsymbol{H}%
_{0}$ is a self-adjoint operator. For further details, the reader may consult
\cite[Section 4.2]{Zuniga-QM-2}. By Stone's theorem on one-parameter unitary
groups, $e^{it\boldsymbol{H}_{0}}$ is a unitary operator on $L^{2}\left(
\mathbb{Z}_{p}\right)$, for $t\geq0$.

We take $J\left(  \left\vert x\right\vert _{p}\right)  \in L^{1}\left(
\mathbb{Z}_{p}\right)  $, $W\left(  x,y\right)  \in L^{2}\left(
\mathbb{Z}_{p}\times\mathbb{Z}_{p}\right)  $, $Z(x)\in L^{2}\left(
\mathbb{Z}_{p}\right)  $, $\phi:\mathbb{R}\rightarrow\mathbb{R}$ is a Lipschitz
function, i.e., $\left\vert \phi\left(  s\right)  -\phi\left(  t\right)
\right\vert \leq L_{\phi}\left\vert s-t\right\vert $, for any $s,t\in
\mathbb{R}$.

We now consider the following initial value problem:%
\begin{equation}
\left\{
\begin{array}
[c]{l}%
\Psi\left(  x,t\right)  \in C^{1}\left(  \left[  0,T\right)  ,L^{2}\left(
\mathbb{Z}_{p}\right)  \right)  \text{, \ \ }x\in\mathbb{Z}_{p}\text{, }%
t\in\left[  0,T\right)  ;\\
\\
\frac{\partial}{\partial t}\Psi\left(  x,t\right)  =\boldsymbol{H}_{0}%
\Psi\left(  x,t\right)  +F(\Psi\left(  x,t\right)  );\\
\\
\Psi\left(  x,0\right)  =\psi_{0}\left(  x\right)  \in L^{2}\left(
\mathbb{Z}_{p}\right)  ,
\end{array}
\right.  \label{EQ_IVP}%
\end{equation}
where%
\[
F(\Psi\left(  x,t\right)  )=-i%
{\displaystyle\int\limits_{\mathbb{Z}_{p}}}
W\left(  x,y\right)  \phi\left(  \Psi\left(  y,t\right)  \right)  dy-iZ(x).
\]
The function
\[%
\begin{array}
[c]{llll}%
F: & L^{2}\left(  \mathbb{Z}_{p}\right)   & \rightarrow & L^{2}\left(
\mathbb{Z}_{p}\right)  \\
&  &  & \\
& \psi\left(  x\right)   & \rightarrow & -i%
{\displaystyle\int\limits_{\mathbb{Z}_{p}}}
W\left(  x,y\right)  \phi\left(  \psi\left(  y\right)  \right)  dy-iZ(x)
\end{array}
\]
$F:L^{2}\left(  \mathbb{Z}_{p}\right)  \rightarrow L^{2}\left(  \mathbb{Z}%
_{p}\right)  $ is Lipschitz, i.e., there exists a positive constant $L_{F}$ such that
\[
\left\Vert F\left(  \varphi\right)  -F\left(  \psi\right)  \right\Vert
_{2}\leq L_{F}\left\Vert \varphi-\psi\right\Vert _{2}.
\]
Indeed,%
\begin{gather*}
\left\vert F\left(  \varphi\left(  x\right)  \right)  -F\left(  \psi\left(
x\right)  \right)  \right\vert =\left\vert \text{ }%
{\displaystyle\int\limits_{\mathbb{Z}_{p}}}
W\left(  x,y\right)  \left\{  \phi\left(  \varphi\left(  y\right)  \right)
-\phi\left(  \psi\left(  y\right)  \right)  \right\}  dy\right\vert \leq\\%
{\displaystyle\int\limits_{\mathbb{Z}_{p}}}
\left\vert W\left(  x,y\right)  \right\vert \text{ }\left\vert \phi\left(
\varphi\left(  y\right)  \right)  -\phi\left(  \psi\left(  y\right)  \right)
\right\vert dy\leq\\
L_{\phi}%
{\displaystyle\int\limits_{\mathbb{Z}_{p}}}
\left\vert W\left(  x,y\right)  \right\vert \text{ }\left\vert \varphi\left(
y\right)  -\psi\left(  y\right)  \right\vert dy,
\end{gather*}
where $L_{\phi}$ is the Lipschitz constant of $\phi$. By applying the
Cauchy-Schwarz inequality,%
\[
\left\vert F\left(  \varphi\left(  x\right)  \right)  -F\left(  \psi\left(
x\right)  \right)  \right\vert \leq L_{\phi}\left\Vert W\left(  x,\cdot
\right)  \right\Vert _{2}\left\Vert \varphi-\psi\right\Vert _{2}.
\]
Then
\[
\left\Vert F\left(  \varphi\right)  -F\left(  \psi\right)  \right\Vert
_{2}\leq L_{\phi}\left\Vert W\right\Vert _{2}\left\Vert \varphi-\psi
\right\Vert _{2}.
\]
Now, by using \cite[Lemma 4.1.1, Proposition 4.3.3]{Cazaneve-Haraux}, any
solution of (\ref{EQ_IVP}) satisfies
\begin{equation}
\Psi\left(  x,t\right)  =e^{it\boldsymbol{H}_{0}}\psi_{0}\left(  x\right)  +%
{\displaystyle\int\limits_{0}^{t}}
e^{i\left(  t-s\right)  \boldsymbol{H}_{0}}F(\Psi\left(  x,s\right)
)ds. \label{Duhamel-Form}%
\end{equation}



\begin{lemma}
\label{Lemma}If $\phi\in
L^{\infty}\left(  \mathbb{R}\right)  $, then \[\int_0^T \Vert F(\Psi(x,t)) \Vert_2^2 dt\leq TC(F),\] for all $\Psi(x,t)\in \mathcal{C}([0,T], L^2(\mathbb Z_p))$.
\end{lemma}

\begin{proof}
We first use the Cauchy-Schwarz inequality and the hypothesis $\phi\in
L^{\infty}\left(  \mathbb{R}\right)  $, to estimate%
\begin{gather*}
\left\Vert F\left(  \Psi\left(  x,t\right)  \right)  \right\Vert _{2}^{2}=%
{\displaystyle\int\limits_{\mathbb{Z}_{p}}}
\text{ }\left\vert F\left(  \Psi\left(  x,t\right)  \right)  \right\vert
^{2}dx=\\%
{\displaystyle\int\limits_{\mathbb{Z}_{p}}}
\text{ }\left\vert -i%
{\displaystyle\int\limits_{\mathbb{Z}_{p}}}
W\left(  x,y\right)  \phi\left(  \Psi\left(  y,t\right)  \right)
dy-iZ(x)\right\vert ^{2}dx\leq
\end{gather*}%
\begin{gather*}%
{\displaystyle\int\limits_{\mathbb{Z}_{p}}}
\text{ }\left\{  \left\vert \text{ }%
{\displaystyle\int\limits_{\mathbb{Z}_{p}}}
W\left(  x,y\right)  \text{ }\phi\left(  \Psi\left(  y,t\right)  \right)
dy\right\vert +\left\vert Z(x)\right\vert \right\}  ^{2}dx\\%
{\displaystyle\int\limits_{\mathbb{Z}_{p}}}
\text{ }\left\{  \left\Vert W\left(  x,\cdot\right)  \right\Vert
_{2}\left\Vert \phi\left(  \Psi\left(  \cdot,t\right)  \right)  \right\Vert
_{2}+\left\vert Z(x)\right\vert \right\}  ^{2}dx\leq\\%
{\displaystyle\int\limits_{\mathbb{Z}_{p}}}
\text{ }\left\{  \left\Vert W\left(  x,\cdot\right)  \right\Vert
_{2}\left\Vert \phi\right\Vert _{\infty}+\left\vert Z(x)\right\vert \right\}
^{2}dx.
\end{gather*}

By using the hypothesis $W\left(  x,y\right)  \in L^{2}\left(  \mathbb{Z}%
_{p}\times\mathbb{Z}_{p}\right)  $, and Fubini's theorem,
\[
\left\Vert \phi\right\Vert _{\infty}\left\vert W\left(  x,\cdot\right)
\right\vert +\left\vert Z(x)\right\vert \in L^{2}\left(  \mathbb{Z}%
_{p}\right)  \text{, for almost all }x\text{,}%
\]
and consequently
\[
\left\Vert F\left(  \Psi\left(  x,t\right)  \right)  \right\Vert _{2}^{2}\leq
{\displaystyle\int\limits_{\mathbb{Z}_{p}}}
\text{ }\left\{  \left\Vert W\left(  x,\cdot\right)  \right\Vert
_{2}\left\Vert \phi\right\Vert _{\infty}+\left\vert Z(x)\right\vert \right\}
^{2}dx=C(F)<\infty,
\]
where $C(F)$ is a positive constant independent of $t$. Finally,%
\[%
{\displaystyle\int\limits_{0}^{T}}
\left\Vert F\left(  \Psi\left(  x,t\right)  \right)  \right\Vert _{2}%
^{2}dt\leq TC(F)<\infty.
\]

\end{proof}

\begin{proposition}\label{Prop: Existence solution}
With the above hypotheses. For any $\psi_{0}\left(  x\right)  \in L^{2}\left(
\mathbb{Z}_{p}\right)  $, there exists $T=T\left(  \psi_{0}\right)  >0$, such
that the Cauchy problem (\ref{EQ_IVP}) has a unique solution. Furthermore,
$T=\infty$, or $\lim_{t\rightarrow T}\left\Vert \Psi\left(  \cdot,t\right)
\right\Vert _{2}=\infty$.
\end{proposition}

\begin{proof}
By \cite[Proposition 4.1.6]{Cazaneve-Haraux} and Lemma \ref{Lemma}, the
mild solution (\ref{Duhamel-Form}) is a classical solution, i.e., a solution
of (\ref{EQ_IVP}). The fact that $T=\infty$, or $\lim_{t\rightarrow
T}\left\Vert \Psi\left(  \cdot,t\right)  \right\Vert _{2}=\infty$ follows from
\cite[Theorem 4.3.4]{Cazaneve-Haraux}.
\end{proof}

\begin{theorem} \label{Theorem_main}
Assume that $\phi\in L^{\infty}\left(  \mathbb{R}\right)  $ is a Lipschitz
function. Then, the Cauchy problem (\ref{EQ_IVP}) has a unique solution with $T=\infty$
for any $\psi_{0}\in L^{2}(\mathbb{Z}_{p})$.
\end{theorem}

\begin{proof}
First, by Lemma \ref{Lemma} and \cite[Proposition 1.4.14]{Cazaneve-Haraux},
\[
\Vert\int_{0}^{t}e^{i(t-s)\boldsymbol{H}_{0}}F(\Psi(x,s))ds\Vert_{2}\leq
\int_{0}^{t}\left\Vert e^{i(t-s)\boldsymbol{H}_{0}}F(\Psi(x,s))\right\Vert
_{2}ds.
\]
Now, by Proposition \ref{Prop: Existence solution}, $T(\psi_{0})=\infty$ or
$T(\psi_{0})<\infty$ and $\lim_{t\rightarrow T}\left\Vert \Psi\left(
\cdot,t\right)  \right\Vert _{2}=\infty$. By Lemma \ref{Lemma} and the
fact that $\Vert e^{it\boldsymbol{H}_{0}}\Vert=e^{\Vert\boldsymbol{H}_{0}%
\Vert}$, and (\ref{Duhamel-Form}), we obtain that
\begin{align*}
\Vert\Psi(x,t)\Vert_{2} &  \leq\Vert e^{it\boldsymbol{H}_{0}}\psi_{0}\Vert
_{2}+\Vert\int_{0}^{t}e^{i(t-s)\boldsymbol{H}_{0}}F(\Psi(x,s))ds\Vert_{2}\\
&  \leq e^{t\Vert\boldsymbol{H}_{0}\Vert}\Vert\psi_{0}\Vert_{2}+\int_{0}%
^{t}e^{(t-s)\Vert\boldsymbol{H}_{0}\Vert}\Vert F(\Psi(x,s))\Vert_{2}\text{
}ds\\
&  \leq e^{t\Vert\boldsymbol{H}_{0}\Vert}\Vert\psi_{0}\Vert+\sqrt{C(F)}\int_{0}%
^{t}e^{(t-s)\Vert\boldsymbol{H}_{0}\Vert}ds\\
&  =e^{t\Vert\boldsymbol{H}_{0}\Vert}\Vert\psi_{0}\Vert+\frac{\sqrt{C(F)}}%
{\Vert\boldsymbol{H}_{0}\Vert}(1-e^{-t\Vert\boldsymbol{H}_{0}\Vert}).
\end{align*}
Notice that the last inequality is finite for all $T$. This
implies that $T(\psi_{0})=\infty$.
\end{proof}

\end{document}